\documentclass[sigconf]{NewEDBTStyle-2025/acmart}

\AtBeginDocument{%
  }

\AtBeginMaketitle{%

\setcopyright{none}
\copyrightyear{\copyright\ 2025 Copyright held by the owner/author(s).
  Published on \url{OpenProceedings.org} under ISBN 978-3-98318-102-5, series ISSN 2367-2005.
  Distribution of this paper is permitted under the terms of the
  Creative Commons license CC-by-nc-nd 4.0}
\acmYear{2025}
\acmDOI{}
\acmISBN{}

\acmConference[EDBT '26]{Extending Database Technology}{24-27 March 2026}{Tampere (Finland)}

\settopmatter{printacmref=false, printccs=false, printfolios=false}

\newsavebox{\ximagebox}
\newlength{\ximageheight}
\newsavebox{\xglyphbox}
\newlength{\xglyphheight}
\newcommand{\xbox}[1]%
  {\savebox{\ximagebox}{#1}%
  \settoheight{\ximageheight}{\usebox{\ximagebox}}%
  \savebox{\xglyphbox}{\color{white}\char32}%
  \settoheight{\xglyphheight}{\usebox{\xglyphbox}}%
  \raisebox{\ximageheight}[0pt][0pt]{\raisebox{-\xglyphheight}[0pt][0pt]{%
    \makebox[0pt][l]{\usebox{\xglyphbox}}}}%
    \usebox{\ximagebox}%
    \raisebox{0pt}[0pt][0pt]{\makebox[0pt][r]{\usebox{\xglyphbox}}}}

\newsavebox{\LogoBox}
\sbox{\LogoBox}{\includegraphics[height=1cm]{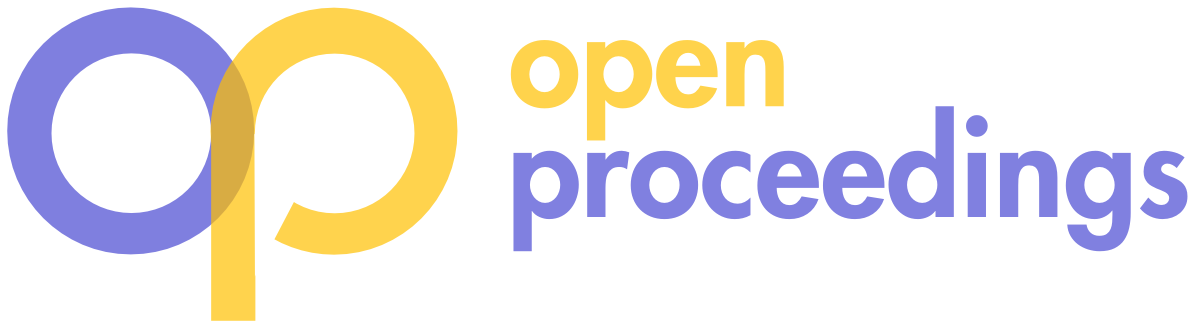}}

\fancypagestyle{firstpagestyle}{%
  \fancyhf{}%
  \fancyfoot{}%
  \fancyhead[L,C]{}
  \fancyhead[R]{\href{https://OpenProceedings.org/}{\raisebox{15pt}[0pt][0pt]{\xbox{\usebox{\LogoBox}}}}}
  }

\hypersetup{%
  pdfcopyright={CC-by-nc-nd}
}

}

\usepackage[inline]{enumitem}
\usepackage{caption}
\usepackage{subcaption}
\usepackage{balance}
\usepackage[binary-units]{siunitx}
\usepackage{graphicx}

\usepackage[nobiblatex]{xurl}

\usepackage{hyperref}
\usepackage{cleveref}
\usepackage{makecell}
\usepackage{listings}
\usepackage{xspace}
\usepackage{natbib}
\usepackage{xspace}
\usepackage[title]{appendix}

\newtheorem{theorem}{Theorem}

\newtheorem{example}{Example}

\newcommand{\para}[1]{\paragraph{\textbf{#1}}}

\newcommand{\MName}[1]{\textsc{#1}}

\newcommand{\OCC}{\MName{OCC}}

\newcommand{\XP}{Thunderbolt}
\newcommand{\TPL}{\MName{2PL-No-Wait}}

\newcommand{\sysname}{\XP\xspace}

\newcommand{\SST}{$Single-shard TX$}
\newcommand{\CST}{$Cross-shard TX$}

\usepackage[noend]{algorithmic}
\newenvironment{myprotocol}{
    \hrule
    \smallskip
    \footnotesize
    \algsetup{linenosize=\tiny}
    \begin{algorithmic}[1]

        \makeatletter
        \newcommand{\EVENT}[1]{\STATE \textbf{event} ##1 \textbf{do}\begin{ALC@g}}
                \newcommand{\ENDEVENT}{\end{ALC@g}}
        \makeatother

        \makeatletter
        \newcommand{\FUNC}[2]{\STATE \textbf{function} \textbf{##1} (##2) \begin{ALC@g}}
                \newcommand{\ENDFUNC}{\end{ALC@g}}
        \makeatother

        }{
    \end{algorithmic}
    \smallskip
    \hrule
}

\usepackage{tikz,pgfplots,pgfplotstable}
\usetikzlibrary{arrows.meta,calc,decorations.pathreplacing,patterns,shapes}
\tikzset{
    dot/.append style={circle,scale=0.35,draw=black,fill=black},
    sdot/.append style={scale=0.45,draw=black,fill=black},
    label/.append style={align=center,font=\strut\footnotesize},
    >=Stealth,
    every edge/.append style={semithick},
    thread/.append style={align=center,draw,thick,rectangle,text width=1cm,text height=2ex,text depth=.25ex,minimum height=0.75cm,font=\strut},
}

\definecolor{colA}{RGB}{230,159,0}
\definecolor{colB}{RGB}{86,180,233}
\definecolor{colC}{RGB}{0,158,115}
\definecolor{colD}{RGB}{240,228,66}
\definecolor{colE}{RGB}{0,114,178}
\definecolor{colF}{RGB}{213,94,0}
\definecolor{colG}{RGB}{204,121,167}
\definecolor{colH}{RGB}{238,130,238}
\definecolor{colI}{RGB}{64,224,208}
\definecolor{colGrey}{RGB}{211,211,211}
\definecolor{colIvory}{RGB}{255,235,215}
\definecolor{colDeepPink}{RGB}{255,20,147}
\definecolor{colLightRed}{RGB}{255,235,0}
\definecolor{colSkyBlue}{RGB}{135, 206, 235}
\definecolor{colVermillion}{RGB}{227, 66, 52}
\definecolor{colReddishPurple}{RGB}{100, 0, 120}
\definecolor{DarkOrange}{RGB}{255,140,0}
\definecolor{Peru}{RGB}{128,0,0}
\definecolor{deepblue}{rgb}{0,0,0.5}
\definecolor{deepred}{rgb}{0.6,0,0}
\definecolor{lightgreen}{RGB}{34,139,34}
\definecolor{lightyellow}{RGB}{218,112,214}
\definecolor{slateblue}{RGB}{123,104,238}
\definecolor{lightblue}{RGB}{30,144,255}
\definecolor{lightBrown}{RGB}{188,143,143}
\definecolor{deepGreen}{RGB}{0,128,0}
\definecolor{deepBlue}{RGB}{0,0,255}
\definecolor{deepPurple}{RGB}{128,0,128}
\definecolor{fakeGreen}{RGB}{102,205,170}
\definecolor{Maroon}{RGB}{210,105,30}

\pgfplotscreateplotcyclelist{mycyclelist}{
    very thick,solid,black,every mark/.append style={solid},mark=*\\         
    very thick,solid,colG,every mark/.append style={solid},mark=*\\          
    very thick,solid,colB,every mark/.append style={solid},mark=*\\          
    very thick,solid,colE,every mark/.append style={solid},mark=*\\          
    very thick,solid,colD,every mark/.append style={solid},mark=*\\          
}

\pgfplotsset{
    compat=1.16,
    width=195pt,
    height=156pt,
    every axis title shift=0pt,
    max space between ticks=25,
    every axis/.append style={
            cycle list name=mycyclelist,
            ymin=0,
            enlargelimits=0.05,
            scale ticks above exponent=1,
            scaled x ticks=false,
            xtick=data,
            mark size=1pt,
            font=\Large,
            y tick label style={
                },
            ylabel shift={-5pt}
        },
    every axis legend/.append style={
            cells={anchor=west}
        }
}

\tikzset{
    plot/.append style={baseline,scale=0.65}
}

\usepackage{tikz,pgfplots,pgfplotstable}
\usetikzlibrary{arrows.meta}
\tikzset{
    >=Stealth,
    node_text/.append style={font=\strut\bfseries},
    plot/.append style={baseline,scale=0.8},
}
\pgfplotscreateplotcyclelist{mycyclelist}{
    thick,black   ,every mark/.append style={solid,fill=\pgfplotsmarklistfill},mark=*\\
    thick,red     ,every mark/.append style={solid,fill=\pgfplotsmarklistfill},mark=square*\\
    thick,blue    ,every mark/.append style={solid,fill=\pgfplotsmarklistfill},mark=triangle*\\
    thick,teal    ,every mark/.append style={solid,fill=\pgfplotsmarklistfill},mark=pentagon*\\
    thick,brown   ,every mark/.append style={solid,fill=\pgfplotsmarklistfill},mark=halfsquare*\\
    thick,orange  ,every mark/.append style={solid,fill=\pgfplotsmarklistfill},mark=halfcircle*\\
    thick,violet  ,every mark/.append style={solid,fill=\pgfplotsmarklistfill,rotate=180},mark=halfdiamond*\\
}
\pgfplotscreateplotcyclelist{mycyclelistex}{
    black   ,every mark/.append style={solid,fill=\pgfplotsmarklistfill},mark=*\\
    red     ,every mark/.append style={solid,fill=\pgfplotsmarklistfill},mark=square*\\
}

\pgfplotsset{
 tick label style={font=\Large},
 legend style={font=\LARGE,cells={anchor=west}},
 title style={font=\Large},
 label style={font=\LARGE},
 width=262.5pt,
 height=185pt,
every axis/.append style={
        xlabel near ticks,
        mark size=2.5pt,
        cycle list name=mycyclelist,
        y tick label style={
            }
    },
barstyle/.append style={
        ybar,
       bar width={0.5cm},
        enlarge x limits=0.4,
        enlarge y limits={upper=0.025},
        ymin=0,
        xtick=data
    }
every axis/.append style={
ylabel near ticks,
xlabel near ticks,
mark size=1pt,
cycle list name=mycyclelist,
 font=\Large,
enlargelimits=0.1
},
}

\pgfplotsset{select coords index/.style n args={2}{
    x filter/.code={
        \ifnum\coordindex<#1\fi
        \ifnum\coordindex>#2\fi
    }
}}

\definecolor{colAA}{RGB}{255,87,51}
\definecolor{colAB}{RGB}{255,116,51}
\definecolor{colBA}{RGB}{53,142,250}
\definecolor{colBB}{RGB}{147,199,251}
\definecolor{colCA}{RGB}{92,208,48}
\definecolor{colCB}{RGB}{167,252,105}
\definecolor{colDA}{RGB}{214,217,60}
\definecolor{colDB}{RGB}{242,245,52}
\definecolor{colEA}{RGB}{242,84,228}
\definecolor{colEB}{RGB}{223,128,216}
\definecolor{colFA}{RGB}{159,28,249}
\definecolor{colGA}{RGB}{28,71,249}

\pgfplotstableread{
Id Alpha Throughput Latency Retry
1	0	24758.47591	0.012104	0.056528
2	0.2	24417.37432	0.012273	0.060278
3	0.4	23442.57465	0.012784	0.086597
4	0.6	21938.00178	0.01671	0.273889
5	0.65	20594.05261	0.014552	0.382014
6	0.7	17737.1138	0.016897	0.632639
7	0.75	15646.79958	0.019157	0.960139
8	0.8	13323.50108	0.022499	1.386875
9	0.85	11804.91396	0.025395	2.197431
10	0.9	11198.47825	0.02677	3.006181
11	0.95	10160.05922	0.022777	3.664097
}\dataZipfianWriteOCC

\pgfplotstableread{
Id Alpha Throughput Latency Retry
1	0	18125.19982	0.016594	0.002431
2	0.2	18240.75517	0.016489	0.003056
3	0.4	18445.6232	0.016307	0.004167
4	0.6	18736.38684	0.016055	0.017847
5	0.65	18036.63692	0.016676	0.025694
6	0.7	17898.18172	0.016805	0.051944
7	0.75	18088.2026	0.016629	0.095625
8	0.8	18795.22601	0.016005	0.179306
9	0.85	18090.95226	0.016626	0.392847
10	0.9	16350.44453	0.018392	0.793125
11	0.95	13337.65572	0.022536	1.624722
}\dataZipfianWriteBamboo

\pgfplotstableread{
Id Alpha Throughput Latency Retry
1	0	19138.349	0.015719	0.005139
2	0.2	18922.74396	0.015897	0.005972
3	0.4	18884.6515	0.015929	0.01
4	0.6	19127.90057	0.015727	0.034097
5	0.65	19392.1373	0.015513	0.054861
6	0.7	19563.96806	0.015377	0.118264
7	0.75	18705.59856	0.016081	0.232361
8	0.8	18052.44235	0.016662	0.447569
9	0.85	16559.05199	0.018161	1.017014
10	0.9	15659.01622	0.019203	1.776944
11	0.95	13822.82782	0.021748	3.182917
}\dataZipfianWriteTPL

\pgfplotstableread{
Id Alpha Throughput Latency Retry
1	0	18199.07513	0.016526	0.002014
2	0.2	18374.00283	0.016368	0.003194
3	0.4	18738.31295	0.016051	0.003681
4	0.6	18471.29509	0.016283	0.014375
5	0.65	18210.7208	0.016516	0.017292
6	0.7	18647.65642	0.01613	0.029583
7	0.75	19727.5948	0.015249	0.046181
8	0.8	20274.37994	0.014839	0.071389
9	0.85	20468.2394	0.0147	0.108889
10	0.9	20689.14983	0.014543	0.164236
11	0.95	21057.61891	0.014289	0.241667
}\dataZipfianWriteX

\pgfplotstableread{
Id Alpha Throughput Latency Retry
1	0	17943.54512	0.01676	0.002153
2	0.2	17068.2243	0.017617	0.002569
3	0.4	18018.17584	0.01669	0.004028
4	0.6	18702.24596	0.016082	0.011875
5	0.65	18876.92768	0.015934	0.019653
6	0.7	19204.22493	0.015663	0.030903
7	0.75	19555.94426	0.015382	0.046667
8	0.8	19778.69834	0.01521	0.066944
9	0.85	20185.96319	0.014903	0.112014
10	0.9	20468.2685	0.014699	0.166875
11	0.95	20745.81194	0.014502	0.245417
}\dataZipfianWriteXV

\pgfplotstableread{
Id Alpha Throughput Latency Retry
1	0	18295.89029	0.016438	0.002083
2	0.2	17749.20375	0.016943	0.002778
3	0.4	19426.90627	0.015484	0.003889
4	0.6	19686.30421	0.01528	0.008194
5	0.65	19884.64146	0.015128	0.008194
6	0.7	20038.54637	0.015013	0.011667
7	0.75	20226.11107	0.014874	0.016875
8	0.8	20418.98625	0.014734	0.020694
9	0.85	20773.58523	0.014483	0.029861
10	0.9	21200.85157	0.014192	0.042222
11	0.95	21687.00828	0.013875	0.063056
}\dataZipfianWriteXVV

\pgfplotstableread{
Id Alpha Throughput Latency Retry
1	0	18124.12762	0.016594	0.002153
2	0.2	18519.75698	0.016241	0.002708
3	0.4	17897.89252	0.016803	0.003889
4	0.6	17880.53567	0.016819	0.01375
5	0.65	19123.02346	0.01573	0.019306
6	0.7	19635.99769	0.01532	0.029861
7	0.75	20008.42021	0.015035	0.043056
8	0.8	20201.73679	0.014892	0.070764
9	0.85	20652.59326	0.014569	0.109931
10	0.9	21049.55416	0.014295	0.164375
11	0.95	21251.16217	0.014159	0.234653
}\dataZipfianWriteXVVV

\pgfplotstableread{
Id Alpha Throughput Latency Retry
1	0	40149.22878	0.005963	0
2	0.2	44132.11439	0.006782	0
3	0.4	41676.91707	0.007183	0
4	0.6	43579.80042	0.005584	0
5	0.65	45333.48445	0.006603	0
6	0.7	41594.09364	0.007198	0
7	0.75	44144.42585	0.006781	0
8	0.8	45365.04686	0.006598	0
9	0.85	46574.50951	0.006426	0
10	0.9	46858.71593	0.008123	0
11	0.95	47823.07207	0.007917	0
}\dataZipfianReadOCC

\pgfplotstableread{
Id Alpha Throughput Latency Retry
1	0	28319.30807	0.010636	0
2	0.2	28356.05756	0.010622	0
3	0.4	28482.75903	0.010576	0
4	0.6	29170.05127	0.010328	0
5	0.65	29544.94628	0.010197	0
6	0.7	29905.11357	0.010075	0
7	0.75	30594.48917	0.009849	0
8	0.8	31375.48834	0.009605	0
9	0.85	32352.35307	0.009316	0
10	0.9	33488.60573	0.009001	0
11	0.95	35063.19898	0.008599	0
}\dataZipfianReadBamboo

\pgfplotstableread{
Id Alpha Throughput Latency Retry
1	0	28650.21667	0.010515	0.0125
2	0.2	28759.67893	0.010475	0.015417
3	0.4	28897.55816	0.010425	0.02
4	0.6	28899.58798	0.010425	0.058056
5	0.65	28672.9787	0.010507	0.093264
6	0.7	28603.26673	0.010533	0.139722
7	0.75	28188.8654	0.010687	0.235972
8	0.8	27666.66282	0.010888	0.393958
9	0.85	27284.92281	0.01104	0.570903
10	0.9	26753.51651	0.011259	0.867986
11	0.95	26163.21851	0.011513	1.288333
}\dataZipfianReadTPL

\pgfplotstableread{
Id Alpha Throughput Latency Retry
1	0	27723.01626	0.010864	0
2	0.2	27789.73208	0.010838	0
3	0.4	27876.07536	0.010805	0
4	0.6	28242.3824	0.010665	0
5	0.65	28513.32701	0.010564	0
6	0.7	28782.55777	0.010466	0
7	0.75	29371.33113	0.010256	0
8	0.8	30019.76301	0.010036	0
9	0.85	30604.30756	0.009845	0
10	0.9	31636.87893	0.009526	0
11	0.95	32684.8567	0.009221	0
}\dataZipfianReadX

\pgfplotstableread{
Id Alpha Throughput Latency Retry
1	0	27348.29766	0.011013	0
2	0.2	27496.18586	0.010957	0
3	0.4	27678.57486	0.010891	0
4	0.6	28017.30847	0.01076	0
5	0.65	28292.1558	0.010654	0
6	0.7	28660.19557	0.01051	0
7	0.75	29185.12694	0.010322	0
8	0.8	29813.60287	0.010105	0
9	0.85	30493.74455	0.009881	0
10	0.9	31632.22269	0.009527	0
11	0.95	33104.3615	0.009105	0
}\dataZipfianReadXV

\pgfplotstableread{
Id Alpha Throughput Latency Retry
1	0	27171.70821	0.011085	0
2	0.2	27644.09002	0.010896	0
3	0.4	28541.9805	0.010554	0
4	0.6	28905.91126	0.010422	0
5	0.65	29194.00225	0.01032	0
6	0.7	29500.51933	0.010213	0
7	0.75	29935.45168	0.010065	0
8	0.8	30422.49236	0.009905	0
9	0.85	31092.37891	0.009692	0
10	0.9	31883.09532	0.009452	0
11	0.95	32527.37721	0.009266	0
}\dataZipfianReadXVV

\pgfplotstableread{
Id Alpha Throughput Latency Retry
1	0	27454.35237	0.010971	0
2	0.2	27578.75746	0.010921	0
3	0.4	27751.92095	0.010853	0
4	0.6	28266.9975	0.010657	0
5	0.65	28449.10869	0.010589	0
6	0.7	28764.67698	0.010473	0
7	0.75	29282.9937	0.010289	0
8	0.8	29609.5244	0.010175	0
9	0.85	30209.95922	0.009975	0
10	0.9	30611.59414	0.009844	0
11	0.95	31662.33875	0.009519	0
}\dataZipfianReadXVVV

\pgfplotstableread{
Id Alpha Throughput Latency Retry
1	0	24974.8084	0.011999	0.053611
2	0.2	23919.94791	0.012529	0.056944
3	0.4	23755.43777	0.012615	0.077222
4	0.6	22923.25644	0.013073	0.24875
5	0.65	24297.40018	0.012331	0.342292
6	0.7	19237.06341	0.015579	0.570208
7	0.75	15446.93121	0.019405	0.900417
8	0.8	15300.86928	0.01959	1.30875
9	0.85	13659.69896	0.021945	2.006667
10	0.9	12369.53658	0.020859	2.625625
11	0.95	12038.99328	0.024294	3.539583
}\dataZipfianWriteHOCC

\pgfplotstableread{
Id Alpha Throughput Latency Retry
1	0	18402.25019	0.016345	0.002153
2	0.2	18628.45419	0.016148	0.002917
3	0.4	18723.03688	0.016066	0.004167
4	0.6	19048.27415	0.015792	0.015556
5	0.65	19243.77317	0.015632	0.020417
6	0.7	19446.66145	0.015469	0.040486
7	0.75	19457.90809	0.01546	0.076944
8	0.8	19220.11706	0.015652	0.168542
9	0.85	18410.86139	0.016338	0.341875
10	0.9	16208.84192	0.018552	0.74375
11	0.95	12589.42646	0.023873	1.772292
}\dataZipfianWriteHBamboo

\pgfplotstableread{
Id Alpha Throughput Latency Retry
1	0	17681.72888	0.017011	0.005347
2	0.2	18415.38196	0.016335	0.006528
3	0.4	18991.74914	0.01584	0.008542
4	0.6	19349.58519	0.015547	0.041181
5	0.65	19500.70081	0.015427	0.062222
6	0.7	19181.9945	0.015683	0.142361
7	0.75	19365.87518	0.015535	0.262639
8	0.8	18710.82421	0.016077	0.511528
9	0.85	17456.85188	0.017229	1.035069
10	0.9	15517.19122	0.019378	1.882431
11	0.95	11165.24013	0.026919	4.557639
}\dataZipfianWriteHTPL

\pgfplotstableread{
Id Alpha Throughput Latency Retry
1	0	18442.67019	0.016308	0.002431
2	0.2	18571.68283	0.016195	0.002986
3	0.4	18677.81802	0.016103	0.003403
4	0.6	19096.29439	0.015751	0.013403
5	0.65	19404.83746	0.015502	0.016875
6	0.7	19554.66956	0.015383	0.026181
7	0.75	19695.25309	0.015088	0.046458
8	0.8	19914.61608	0.015106	0.069306
9	0.85	20614.30633	0.014595	0.105625
10	0.9	20961.03431	0.014354	0.151389
11	0.95	21233.23828	0.014171	0.242014
}\dataZipfianWriteHX

\pgfplotstableread{
Id Alpha Throughput Latency Retry
1	0	18902.07805	0.015912	0.002708
2	0.2	19006.28791	0.015825	0.003125
3	0.4	17717.42915	0.016973	0.003333
4	0.6	19437.23797	0.015476	0.012639
5	0.65	19935.95574	0.01509	0.017361
6	0.7	20183.92603	0.014905	0.028194
7	0.75	20444.75869	0.014715	0.044792
8	0.8	20833.24291	0.014442	0.067986
9	0.85	21277.79042	0.014141	0.097708
10	0.9	21647.98275	0.0139	0.154583
11	0.95	21782.73537	0.013814	0.234028
}\dataZipfianWriteHXV

\pgfplotstableread{
Id Alpha Throughput Latency Retry
1	0	18701.51729	0.016083	0.002431
2	0.2	18887.07888	0.015926	0.002778
3	0.4	18958.19455	0.015866	0.002708
4	0.6	19477.80545	0.015444	0.007361
5	0.65	19705.7278	0.015266	0.008264
6	0.7	20059.73343	0.014997	0.010694
7	0.75	20335.88097	0.014794	0.015486
8	0.8	20611.0902	0.014597	0.020347
9	0.85	20851.13161	0.01443	0.024167
10	0.9	21384.3707	0.014071	0.045069
11	0.95	20331.94728	0.014797	0.06375
}\dataZipfianWriteHXVV

\pgfplotstableread{
Id Alpha Throughput Latency Retry
1	0	19986.73103	0.015051	0.002083
2	0.2	19862.75389	0.015145	0.002917
3	0.4	21383.89436	0.014071	0.003403
4	0.6	19834.76539	0.015167	0.012083
5	0.65	19525.26488	0.015407	0.018056
6	0.7	19797.37932	0.015195	0.026458
7	0.75	20178.80665	0.014908	0.040069
8	0.8	20521.29797	0.01466	0.061667
9	0.85	20886.75896	0.014405	0.085556
10	0.9	21278.16771	0.014141	0.131111
11	0.95	21790.25068	0.01381	0.184444
}\dataZipfianWriteHXVVV

\pgfplotstableread{
Id Alpha Throughput Latency Retry
1	0	40721.10286	0.007353	0.023333
2	0.2	30217.94694	0.009915	0.022292
3	0.4	28730.5869	0.010429	0.0325
4	0.6	28678.86049	0.010445	0.099861
5	0.65	27203.68761	0.011012	0.145
6	0.7	26238.828	0.011417	0.217292
7	0.75	27543.99388	0.010875	0.341111
8	0.8	23123.83779	0.012956	0.543889
9	0.85	23295.5051	0.012858	0.828889
10	0.9	17219.89568	0.017401	1.165833
11	0.95	19669.36346	0.015234	1.557222
}\dataZipfianRWOCC

\pgfplotstableread{
Id Alpha Throughput Latency Retry
1	0	19892.25031	0.015125	0.001458
2	0.2	20049.9299	0.015007	0.001944
3	0.4	21565.47368	0.013954	0.002222
4	0.6	22150.19989	0.013587	0.005694
5	0.65	22224.48583	0.013542	0.010556
6	0.7	22553.49407	0.013345	0.019653
7	0.75	22647.17123	0.01329	0.036667
8	0.8	22936.07207	0.013123	0.070694
9	0.85	23034.95074	0.013067	0.127847
10	0.9	22331.84816	0.013476	0.292708
11	0.95	21024.53739	0.014312	0.589306
}\dataZipfianRWBamboo

\pgfplotstableread{
Id Alpha Throughput Latency Retry
1	0	21435.11041	0.014038	0.021319
2	0.2	21602.90559	0.01393	0.022014
3	0.4	21411.52505	0.014054	0.030347
4	0.6	21731.224	0.013847	0.100347
5	0.65	20715.63901	0.014524	0.158472
6	0.7	19900.58002	0.015118	0.255556
7	0.75	19474.88146	0.015452	0.440208
8	0.8	18660.94743	0.01612	0.771736
9	0.85	17589.52641	0.017102	1.394653
10	0.9	16433.00981	0.0183	2.073819
11	0.95	15418.7805	0.019501	3.012639
}\dataZipfianRWTPL

\pgfplotstableread{
Id Alpha Throughput Latency Retry
1	0	20995.78335	0.01433	0.001528
2	0.2	21155.96797	0.014222	0.001458
3	0.4	21190.93028	0.014198	0.001736
4	0.6	21845.28984	0.013773	0.006111
5	0.65	21878.84589	0.013753	0.011528
6	0.7	22289.22749	0.0135	0.015694
7	0.75	22594.50859	0.013319	0.025417
8	0.8	23031.56124	0.013066	0.036111
9	0.85	23529.98848	0.012792	0.056042
10	0.9	22599.43753	0.013318	0.085556
11	0.95	23538.91195	0.012789	0.12
}\dataZipfianRWX

\pgfplotstableread{
Id Alpha Throughput Latency Retry
1	0	34287.34702	0.008789	0.001667
2	0.2	34458.25167	0.008746	0.001458
3	0.4	34548.69567	0.008723	0.001389
4	0.6	35559.24318	0.008476	0.005278
5	0.65	35579.01135	0.008472	0.009097
6	0.7	36008.91221	0.008371	0.012917
7	0.75	36688.45893	0.008216	0.019167
8	0.8	34352.78401	0.008774	0.030903
9	0.85	35739.99156	0.008436	0.044722
10	0.9	38536.36736	0.007825	0.066806
11	0.95	38820.18973	0.007768	0.102083
}\dataZipfianRWXV

\pgfplotstableread{
Id Alpha Throughput Latency Retry
1	0	28567.51754	0.010544	0.001181
2	0.2	28502.99875	0.010568	0.001458
3	0.4	28901.79211	0.010422	0.001111
4	0.6	29484.08985	0.010217	0.003542
5	0.65	29609.09822	0.010174	0.005347
6	0.7	29993.81378	0.010044	0.005625
7	0.75	30501.237	0.009878	0.008958
8	0.8	30964.01294	0.00973	0.011736
9	0.85	31723.09265	0.009498	0.013889
10	0.9	32277.87665	0.009337	0.020069
11	0.95	32663.35648	0.009227	0.029722
}\dataZipfianRWXVV

\pgfplotstableread{
Id Alpha Throughput Latency Retry
1	0	28408.5865	0.010601	0.000972
2	0.2	28661.90693	0.010508	0.001319
3	0.4	27167.29963	0.011091	0.000972
4	0.6	28749.05667	0.010479	0.002917
5	0.65	29751.63582	0.010125	0.006458
6	0.7	27897.02897	0.010796	0.009306
7	0.75	30746.76198	0.009799	0.011042
8	0.8	31386.6354	0.0096	0.016042
9	0.85	32169.42564	0.009367	0.01875
10	0.9	32801.2993	0.009188	0.025486
11	0.95	33019.11669	0.009128	0.032639
}\dataZipfianRWXVVV

\pgfplotstableread{
Id Alpha Latency Delay Percentage
1	0	0.013666	0.000048	0.3512366457
2	0.2	0.013609	0.000066	0.4849731795
3	0.4	0.013541	0.000217	1.602540433
4	0.6	0.013356	0.000242	1.811919736
5	0.65	0.013263	0.000399	3.008369147
6	0.7	0.013144	0.000942	7.166768107
7	0.75	0.012864	0.000667	5.185012438
8	0.8	0.012664	0.001284	10.13897663
9	0.85	0.012733	0.001661	13.04484411
10	0.9	0.012992	0.001866	14.36268473
11	0.95	0.013927	0.003598	25.83470956
}\dataZipfianDelayBam

\pgfplotstableread{
Id Alpha Latency Delay Percentage
1	0	0.013976	0.000064	0.4579278764
2	0.2	0.014294	0.000037	0.2588498671
3	0.4	0.014151	0.000259	1.830259346
4	0.6	0.013881	0.000168	1.210287443
5	0.65	0.013691	0.00045	3.286830765
6	0.7	0.013544	0.000282	2.082102776
7	0.75	0.0133	0.000259	2.047368421
8	0.8	0.013093	0.000221	2.087924845
9	0.85	0.012836	0.000575	4.479588657
10	0.9	0.01265	0.000874	6.909090909
11	0.95	0.012341	0.001589	12.87577992
}\dataZipfianDelayXP

\pgfplotstableread{
Id Alpha Throughput Latency Retry
1	0	31616.17958	0.009476	0.001597
2	0.2	31799.80258	0.009422	0.000903
3	0.4	31684.14429	0.009457	0.002292
4	0.6	31933.43653	0.009383	0.005694
5	0.65	31019.30737	0.009659	0.007153
6	0.7	31775.73393	0.009428	0.014306
7	0.75	30965.61097	0.009675	0.024792
8	0.8	31269.95088	0.00958	0.032222
9	0.85	31089.89514	0.009636	0.048194
10	0.9	30621.42351	0.009782	0.070139
11	0.95	30360.27527	0.009866	0.121875
}\dataZipfianReadHOCC

\pgfplotstableread{
Id Alpha Throughput Latency Retry
1	0	26661.18631	0.011295	0.000069
2	0.2	26670.02511	0.011291	0.000069
3	0.4	26789.30214	0.011241	0.000208
4	0.6	27534.51375	0.010938	0.000556
5	0.65	27837.54473	0.01082	0.000556
6	0.7	28289.4879	0.010648	0.000972
7	0.75	28799.8272	0.010459	0.003125
8	0.8	29514.30522	0.010207	0.004583
9	0.85	30417.28716	0.009909	0.008056
10	0.9	31254.47655	0.009643	0.013542
11	0.95	32998.61131	0.009134	0.019097
}\dataZipfianReadHBamboo

\pgfplotstableread{
Id Alpha Throughput Latency Retry
1	0	26210.60219	0.01149	0.015486
2	0.2	26209.17103	0.011491	0.015
3	0.4	26193.91501	0.011497	0.019931
4	0.6	26191.5805	0.011498	0.069028
5	0.65	26095.98611	0.01154	0.102222
6	0.7	25394.58602	0.011858	0.165347
7	0.75	25039.51549	0.012026	0.267569
8	0.8	24802.22807	0.012141	0.434444
9	0.85	24528.08977	0.012276	0.680347
10	0.9	24029.51626	0.01253	0.942153
11	0.95	23543.76115	0.012788	1.390278
}\dataZipfianReadHTPL

\pgfplotstableread{
Id Alpha Throughput Latency Retry
1	0	25856.82144	0.011645	0.000139
2	0.2	25941.78591	0.011607	0.000139
3	0.4	26030.32171	0.011568	0.000347
4	0.6	26771.12363	0.011249	0.001667
5	0.65	27032.96332	0.01114	0.002361
6	0.7	27401.53525	0.010991	0.003403
7	0.75	27967.4102	0.010769	0.004375
8	0.8	28566.38411	0.010544	0.006389
9	0.85	29291.27324	0.010285	0.012847
10	0.9	30041.11878	0.010029	0.01375
11	0.95	31342.09461	0.009614	0.022847
}\dataZipfianReadHX

\pgfplotstableread{
Id Alpha Throughput Latency Retry
1	0	108484.3828	0.002805	0
2	0.2	108959.6622	0.002794	0
3	0.4	107624.9271	0.002837	0.000069
4	0.6	108924.222	0.002801	0.000139
5	0.65	110583.7903	0.002753	0
6	0.7	110510.8055	0.002755	0
7	0.75	111405.0968	0.002733	0.000139
8	0.8	110660.273	0.002751	0.000208
9	0.85	112209.8324	0.002713	0.000069
10	0.9	113406.3649	0.002686	0.000278
11	0.95	111996.889	0.002723	0.000694
}\dataZipfianReadHXV

\pgfplotstableread{
Id Alpha Throughput Latency Retry
1	0	83668.88041	0.003628	0
2	0.2	84599.38665	0.003589	0
3	0.4	84412.42504	0.003597	0
4	0.6	85267.14077	0.003562	0.000139
5	0.65	85171.31856	0.003565	0
6	0.7	86658.24156	0.003505	0.000139
7	0.75	87275.90094	0.003481	0.000069
8	0.8	88239.07888	0.003443	0.000486
9	0.85	86873.63506	0.003496	0.000417
10	0.9	88913.58711	0.003417	0.000486
11	0.95	89319.49708	0.003401	0.000556
}\dataZipfianReadHXVV

\pgfplotstableread{
Id Alpha Throughput Latency Retry
1	0	84306.17191	0.003601	0
2	0.2	84792.17086	0.003581	0
3	0.4	85397.0965	0.003556	0
4	0.6	86603.51827	0.003506	0.000069
5	0.65	86683.28096	0.003503	0
6	0.7	87883.11546	0.003456	0.000069
7	0.75	89086.85969	0.00341	0
8	0.8	91159.43405	0.003333	0.000139
9	0.85	90135.767	0.003371	0.000278
10	0.9	91649.6945	0.003316	0.000069
11	0.95	92312.42628	0.003292	0.000069
}\dataZipfianReadHXVVV

\pgfplotstableread{
Id Batch Throughput Latency
1	100,	9649.269168	0.010359
2	150,	13686.88587	0.010951
3	200,	17479.87538	0.011431
4	250,	20246.74028	0.012335
5	300,	23007.52236	0.01248
6	350,	23180.94745	0.012858
7	400,	22518.13766	0.017741
}\dataOCCBatch

\pgfplotstableread{
Id Batch Throughput Latency
1	100,	18131.06494	0.005511
2	150,	18847.77282	0.007951
3	200,	19588.3588	0.010201
4	250,	19864.49106	0.012575
5	300,	20089.93034	0.01492
6	350,	19630.47959	0.017814
7	400,	19518.70136	0.020475
}\dataBambooBatch

\pgfplotstableread{
Id Batch Throughput Latency
1	100	10856.38808	0.009226
2	150	12285.93465	0.01223
3	200	13727.62112	0.014596
4	250	14281.14092	0.01754
5	300	14913.03835	0.020157
6	350	14485.68073	0.02421
7	400	14568.5877	0.027511
}\dataTPLBatch

\pgfplotstableread{
Id Threads	Protocol	Batch	Throughput	Latency
1	1	occ	100	81.1598	0.533102
2	2	occ	100	54.8628	0.771414
3	4	occ	100	65.0454	0.755049
4	8	occ	100	90.1166	0.439948
5	16	occ	100	56.3193	0.740967
6	32	occ	100	81.1652	0.547239
1	1	occ	300	90.0433	0.377916
2	2	occ	300	84.9988	0.4465
3	4	occ	300	88.4714	0.450773
4	8	occ	300	102.795	0.380749
5	16	occ	300	84.4342	0.471383
6	32	occ	300	93.0561	0.455972
1	1	occ	500	95.9623	0.346051
2	2	occ	500	108.2	0.324163
3	4	occ	500	92.5704	0.415311
4	8	occ	500	86.1235	0.462094
5	16	occ	500	91.6345	0.434726
6	32	occ	500	83.4363	0.520591
1	1	bam	100	119.471	0.269387
2	2	bam	100	114.159	0.350792
3	4	bam	100	120.381	0.305888
4	8	bam	100	120.2902	0.466819
5	16	bam	100	113.204	0.341725
6	32	bam	100	89.8173	0.520782
1	1	bam	300	116.623	0.280427
2	2	bam	300	105.97	0.352624
3	4	bam	300	105.764	0.367738
4	8	bam	300	111.454	0.355268
5	16	bam	300	110.546	0.343837
6	32	bam	300	99.661	0.391814
1	1	bam	500	105.24	0.313931
2	2	bam	500	100.84	0.36104
3	4	bam	500	106.696	0.360295
4	8	bam	500	108.109	0.36954
5	16	bam	500	106.609	0.367823
6	32	bam	500	95.3233	0.425144
1	1	xp	100	88.6425	0.384886
2	2	xp	100	92.5636	0.419997
3	4	xp	100	92.3826	0.436384
4	8	xp	100	96.3269	0.413444
5	16	xp	100	89.5717	0.444016
6	32	xp	100	88.5546	0.467045
1	1	xp	300	111.408	0.298008
2	2	xp	300	107.186	0.339947
3	4	xp	300	109.312	0.347718
4	8	xp	300	108.64	0.359964
5	16	xp	300	108.7	0.352967
6	32	xp	300	98.5256	0.400521
1	1	xp	500	115.307	0.281963
2	2	xp	500	119.148	0.297159
3	4	xp	500	115.331	0.308856
4	8	xp	500	116.642	0.334018
5	16	xp	500	111.074	0.33697
6	32	xp	500	107.241	0.364368
}\dataYCSBThreads

\pgfplotstableread{
Id Threads Protocol	Batch Throughput Latency Theta Retry
1	16	occ	300	77	0.48545	0.8	0.009
2	16	occ	500	80	0.450144	0.8	0.01257
3	16	occ	1000	80	0.444997	0.8	0.03619
4	16	occ	1500	80	0.453995	0.8	0.02853
1	16	bam	300	57	0.745702	0.8	0.00125
2	16	bam	500	59	0.711289	0.8	0.00172
3	16	bam	1000	60	0.70798	0.8	0.00133
4	16	bam	1500	60	0.71439	0.8	0.00133
1	16	xp	300	56	0.772153	0.8	0
2	16	xp	500	57	0.752247	0.8	0
3	16	xp	1000	59	0.707318	0.8	0
4	16	xp	1500	60	0.704412	0.8	0
1	16	2pl	300	76	0.502783	0.8	0.14708
2	16	2pl	500	73	0.552604	0.8	0.23648
3	16	2pl	1000	73	0.523208	0.8	0.36959
4	16	2pl	1500	67	0.594727	0.8	0.4971
1	16	occ	300	73	0.537669	1.6	0.49303
2	16	occ	500	77	0.499348	1.6	0.50598
3	16	occ	1000	89	0.426161	1.6	0.44466
4	16	occ	1500	93	0.409616	1.6	0.42468
1	16	bam	300	94	0.399161	1.6	0.08368
2	16	bam	500	101	0.388678	1.6	0.08215
3	16	bam	1000	101	0.380951	1.6	0.08454
4	16	bam	1500	102	0.385866	1.6	0.08536
1	16	xp	300	98	0.394984	1.6	0
2	16	xp	500	100	0.379789	1.6	0
3	16	xp	1000	103	0.37468	1.6	0
4	16	xp	1500	108	0.339752	1.6	0
1	16	occ	300	76	0.512479	2	0.56693
2	16	occ	500	100	0.396591	2	0.43696
3	16	occ	1000	96	0.402868	2	0.48502
4	16	occ	1500	99	0.388908	2	0.45849
1	16	bam	300	107	0.345084	2	0.11695
2	16	bam	500	105	0.372196	2	0.1132
3	16	bam	1000	103	0.381963	2	0.11691
4	16	bam	1500	95	0.403307	2	0.11735
1	16	xp	300	101	0.37759	2	0
2	16	xp	500	108	0.357807	2	0
3	16	xp	1000	110	0.332384	2	0
4	16	xp	1500	113	0.343225	2	0
}\dataYCSBBatch

\pgfplotstableread{
Id Threads Protocol Batch Throughput Latency Theta
1	16	2pl	500	80	0.437896	0
2	16	2pl	500	79	0.457337	0.2
3	16	2pl	500	80	0.443603	0.4
4	16	2pl	500	80	0.450638	0.6
5	16	2pl	500	61	0.604791	0.8
6	16	2pl	500	37	1.25303	1
7	16	2pl	500	37	1.25303	1
}\dataYCSBTPLTheta

\pgfplotstableread{
Id Threads Protocol Batch Throughput Latency Theta Retry
1	16	occ	500	110	0.244953	0	0.00104
2	16	occ	500	107	0.267073	0.1	0.00098
3	16	occ	500	110	0.261042	0.2	0.00109
4	16	occ	500	111	0.250982	0.3	0.00121
5	16	occ	500	110	0.251878	0.4	0.00127
6	16	occ	500	108	0.266173	0.5	0.00215
7	16	occ	500	109	0.255481	0.6	0.00552
8	16	occ	500	97	0.328612	0.7	0.01667
9	16	occ	500	104	0.297771	0.8	0.04884
10	16	occ	500	76	0.419835	0.9	0.12447
}\dataYCSBOCC

\pgfplotstableread{
Id Threads Protocol Batch Throughput Latency Theta Retry
1	16	bam	500	74	0.495478	0	2.00E-05
2	16	bam	500	75	0.479237	0.1	2.00E-05
3	16	bam	500	74	0.50028	0.2	0
4	16	bam	500	76	0.483175	0.3	3.00E-05
5	16	bam	500	76	0.483837	0.4	2.00E-05
6	16	bam	500	76	0.482143	0.5	6.00E-05
7	16	bam	500	79	0.459633	0.6	9.00E-05
8	16	bam	500	82	0.439593	0.7	0.00031
9	16	bam	500	82	0.444621	0.8	0.00139
10	16	bam	500	86	0.423128	0.9	0.00348
}\dataYCSBBam

\pgfplotstableread{
Id Threads Protocol Batch Throughput Latency Theta Retry
1	16	2pl	500	75	0.488189	0	0.00321
2	16	2pl	500	75	0.490533	0.1	0.00348
3	16	2pl	500	75	0.488658	0.2	0.00376
4	16	2pl	500	75	0.485488	0.3	0.00431
5	16	2pl	500	75	0.485039	0.4	0.00503
6	16	2pl	500	76	0.478965	0.5	0.00775
7	16	2pl	500	75	0.485698	0.6	0.0191
8	16	2pl	500	75	0.494248	0.7	0.06421
9	16	2pl	500	70	0.554037	0.8	0.22966
10	16	2pl	500	64	0.624116	0.9	0.62296
}\dataYCSBTPL

\pgfplotstableread{
Id Threads Protocol Batch Throughput Latency Theta Retry
1	16	xp	500	72	0.514185	0	0
2	16	xp	500	74	0.495518	0.1	0
3	16	xp	500	74	0.49074	0.2	0
4	16	xp	500	74	0.496116	0.3	0
5	16	xp	500	75	0.491315	0.4	0
6	16	xp	500	74	0.506579	0.5	0
7	16	xp	500	77	0.478296	0.6	0
8	16	xp	500	78	0.478276	0.7	0
9	16	xp	500	80	0.464009	0.8	0
10	16	xp	500	83	0.451124	0.9	0
}\dataYCSBXP

\pgfplotstableread{
Id Threads Protocol Batch Throughput Latency Theta Retry
1	16	occ	500	192	0.0548444	0	0
2	16	occ	500	202	0.0322813	0.2	0
3	16	occ	500	201	0.0269267	0.4	0
4	16	occ	500	200	0.0929137	0.6	0
5	16	occ	500	200	0.071107	0.8	0
}\dataYCSBOCCAllRead

\pgfplotstableread{
Id Threads Protocol Batch Throughput Latency Theta Retry
1	16	bam	500	73	0.524643	0	0
2	16	bam	500	73	0.514922	0.2	0
3	16	bam	500	77	0.491614	0.4	0
4	16	bam	500	74	0.516408	0.6	0
5	16	bam	500	79	0.481146	0.8	0
}\dataYCSBBamAllRead

\pgfplotstableread{
Id Threads Protocol Batch Throughput Latency Theta Retry
1	16	2pl	500	76	0.480023	0	0.00366
2	16	2pl	500	76	0.484986	0.2	0.00399
3	16	2pl	500	72	0.528419	0.4	0.00522
4	16	2pl	500	77	0.483145	0.6	0.01836
5	16	2pl	500	66	0.597352	0.8	0.23961
}\dataYCSBTPLAllRead

\pgfplotstableread{
Id Threads Protocol Batch Throughput Latency Theta Retry
1	16	xp	500	73	0.517001	0	0
2	16	xp	500	71	0.527724	0.2	0
3	16	xp	500	71	0.531621	0.4	0
4	16	xp	500	73	0.52453	0.6	0
5	16	xp	500	76	0.51005	0.8	0
}\dataYCSBXAllRead

\pgfplotstableread{
Id Threads Protocol Batch Throughput Latency Theta Retry
1	16	occ	500	87	0.343995	0	0
2	16	occ	500	88	0.348903	0.2	0
3	16	occ	500	89	0.337015	0.4	0
4	16	occ	500	86	0.352593	0.6	0
5	16	occ	500	88	0.350561	0.8	0
}\dataYCSBOCCAllWrite

\pgfplotstableread{
Id Threads Protocol Batch Throughput Latency Theta Retry
1	16	bam	500	84	0.364421	0	0
2	16	bam	500	83	0.366686	0.2	0
3	16	bam	500	86	0.351744	0.4	0
4	16	bam	500	85	0.361635	0.6	0
5	16	bam	500	87	0.372657	0.8	0
}\dataYCSBBamAllWrite

\pgfplotstableread{
Id Threads Protocol Batch Throughput Latency Theta Retry
1	16	2pl	500	86	0.346982	0	0.00381
2	16	2pl	500	85	0.354042	0.2	0.0042
3	16	2pl	500	87	0.352287	0.4	0.00514
4	16	2pl	500	78	0.439572	0.6	0.01979
5	16	2pl	500	67	0.521771	0.8	0.25571
}\dataYCSBTPLAllWrite

\pgfplotstableread{
Id Threads Protocol Batch Throughput Latency Theta Retry
1	16	xp	500	80	0.388691	0	0
2	16	xp	500	83	0.386234	0.2	0
3	16	xp	500	84	0.371319	0.4	0
4	16	xp	500	82	0.394617	0.6	0
5	16	xp	500	81	0.397716	0.8	0
}\dataYCSBXAllWrite

\pgfplotstableread{
Id type	node Throughput Latency
1	none	8	12622.375	32.46136784
2	none	16	12067.63889	41.16006144
3	none	24	11913.7037	43.10124918
4	none	32	11914.51389	43.73126384
5	none	48	11718.50622	44.9113551
6	none	64	11399.72028	45.57460411
1	xp	8	80837.1875	6.351744563
2	xp	16	72927.88732	13.33550963
3	xp	24	71149.53704	20.15010017
4	xp	32	68198.74448	26.32732763
5	xp	48	64632.27273	39.63442035
6	xp	64	62405.97765	49.37663729
}\dataDAG

\definecolor{deepGreen}{RGB}{0,128,0}
\definecolor{deepBlue}{RGB}{0,0,255}
\definecolor{deepPurple}{RGB}{128,0,128}
\definecolor{fakeGreen}{RGB}{102,205,170}
\definecolor{Maroon}{RGB}{210,105,30}

\pgfplotstableread{
Id	name	BATCH	alpha	Throughput Latency exe	Retry	delay
1	smallbank_db:bamboo_update	100	0	14942.28542	0.006708	1.009444	0.002083	0.000037
2	smallbank_db:bamboo_update	100	0.2	14947.31073	0.006706	1.008303	0.003125	0.000038
3	smallbank_db:bamboo_update	100	0.4	14993.86189	0.006685	1.0038	0.003958	0.000051
4	smallbank_db:bamboo_update	100	0.6	14885.42872	0.006735	0.993067	0.017292	0.000039
5	smallbank_db:bamboo_update	100	0.7	15173.97916	0.006606	0.958477	0.039792	0.000187
6	smallbank_db:bamboo_update	100	0.8	15169.80703	0.006607	0.888779	0.137917	0.000103
7	smallbank_db:bamboo_update	100	0.9	14800.06043	0.006772	0.782784	0.380625	0.000501
1	smallbank_db:bamboo_update	300	0	17624.40196	0.017065	0.8848	0.002431	0.000043
2	smallbank_db:bamboo_update	300	0.2	17638.6502	0.017055	0.87911	0.003681	0.000058
3	smallbank_db:bamboo_update	300	0.4	17724.73487	0.016969	0.87866	0.004792	0.00005
4	smallbank_db:bamboo_update	300	0.6	17849.24085	0.016851	0.858334	0.016319	0.000122
5	smallbank_db:bamboo_update	300	0.7	17819.49532	0.017265	0.862382	0.047083	0.000081
6	smallbank_db:bamboo_update	300	0.8	18177.22797	0.016547	0.746332	0.207639	0.000231
7	smallbank_db:bamboo_update	300	0.9	15696.34548	0.019157	0.656767	0.803611	0.000305
1	smallbank_db:bamboo_update	500	0	17968.37566	0.027896	0.875219	0.002417	0.00006
2	smallbank_db:bamboo_update	500	0.2	18032.95071	0.027796	0.872914	0.002583	0.000039
3	smallbank_db:bamboo_update	500	0.4	17858.20585	0.028069	0.878804	0.003708	0.000049
4	smallbank_db:bamboo_update	500	0.6	18292.29254	0.027405	0.845494	0.017083	0.000087
5	smallbank_db:bamboo_update	500	0.7	19095.46379	0.026253	0.787832	0.055542	0.000149
6	smallbank_db:bamboo_update	500	0.8	18571.40978	0.026992	0.731692	0.247958	0.000252
7	smallbank_db:bamboo_update	500	0.9	15171.35093	0.033027	0.640541	1.108083	0.000465
1	smallbank_db:TPL_update	100	0	14944.565	0.006707	1.009189	0.00375	0
2	smallbank_db:TPL_update	100	0.2	14558.64556	0.006887	1.029367	0.006458	0
3	smallbank_db:TPL_update	100	0.4	14855.65296	0.007234	1.067139	0.00625	0
4	smallbank_db:TPL_update	100	0.6	15286.72157	0.006557	0.969929	0.030833	0
5	smallbank_db:TPL_update	100	0.7	15510.83982	0.006462	0.933663	0.071458	0
6	smallbank_db:TPL_update	100	0.8	15488.01776	0.006472	0.854721	0.251042	0
7	smallbank_db:TPL_update	100	0.9	15085.61084	0.006645	0.748126	0.695625	0
1	smallbank_db:TPL_update	300	0	18011.61749	0.0167	0.863627	0.005556	0
2	smallbank_db:TPL_update	300	0.2	18206.11599	0.016521	0.855158	0.006181	0
3	smallbank_db:TPL_update	300	0.4	18302.21533	0.016435	0.850535	0.010069	0
4	smallbank_db:TPL_update	300	0.6	18663.17251	0.016118	0.821678	0.031667	0
5	smallbank_db:TPL_update	300	0.7	18991.74914	0.015839	0.779079	0.1025	0
6	smallbank_db:TPL_update	300	0.8	18094.20294	0.016623	0.713557	0.458056	0
7	smallbank_db:TPL_update	300	0.9	15550.18498	0.019336	0.64865	1.641042	0
1	smallbank_db:TPL_update	500	0	18315.84298	0.027368	0.856815	0.005	0
2	smallbank_db:TPL_update	500	0.2	18478.76636	0.027127	0.850837	0.005958	0
3	smallbank_db:TPL_update	500	0.4	18712.86666	0.026789	0.838481	0.008125	0
4	smallbank_db:TPL_update	500	0.6	19307.30226	0.025966	0.798984	0.039625	0
5	smallbank_db:TPL_update	500	0.7	18803.01569	0.026661	0.777581	0.158083	0
6	smallbank_db:TPL_update	500	0.8	18033.56046	0.027796	0.680912	0.7395	0
7	smallbank_db:TPL_update	500	0.9	14701.98164	0.034079	0.581207	2.642542	0
1	smallbank_db:fx_update	100	0	14305.17102	0.007005	0.911109	0.00125	0.000034
2	smallbank_db:fx_update	100	0.2	14381.37612	0.006968	0.904621	0.003125	0.000039
3	smallbank_db:fx_update	100	0.4	14441.38902	0.006939	0.902467	0.003125	0.000033
4	smallbank_db:fx_update	100	0.6	14568.18975	0.006878	0.880266	0.014375	0.000035
5	smallbank_db:fx_update	100	0.7	14874.31206	0.006737	0.846922	0.030208	0.000034
6	smallbank_db:fx_update	100	0.8	15299.19488	0.006551	0.787414	0.075417	0.000173
7	smallbank_db:fx_update	100	0.9	15643.58694	0.006407	0.71691	0.158958	0.000196
1	smallbank_db:fx_update	300	0	16861.09618	0.017834	0.773818	0.002847	0.000035
2	smallbank_db:fx_update	300	0.2	16483.30893	0.018241	0.794003	0.003125	0.000038
3	smallbank_db:fx_update	300	0.4	17584.17815	0.017102	0.736149	0.004792	0.00004
4	smallbank_db:fx_update	300	0.6	17905.10295	0.016797	0.710061	0.013264	0.000081
5	smallbank_db:fx_update	300	0.7	18592.9525	0.016177	0.674826	0.027708	0.000121
6	smallbank_db:fx_update	300	0.8	19331.27223	0.015561	0.618434	0.071806	0.000155
7	smallbank_db:fx_update	300	0.9	20011.67348	0.015033	0.548676	0.167569	0.000414
1	smallbank_db:fx_update	500	0	17834.14689	0.028104	0.725579	0.001875	0.000037
2	smallbank_db:fx_update	500	0.2	17612.4943	0.028458	0.722883	0.002583	0.000036
3	smallbank_db:fx_update	500	0.4	17875.73682	0.028039	0.712386	0.003958	0.000042
4	smallbank_db:fx_update	500	0.6	18727.04474	0.026768	0.678869	0.013375	0.000077
5	smallbank_db:fx_update	500	0.7	19428.14864	0.025804	0.637113	0.029625	0.00015
6	smallbank_db:fx_update	500	0.8	19639.98275	0.025527	0.595538	0.075208	0.000184
7	smallbank_db:fx_update	500	0.9	21159.33775	0.023699	0.503612	0.166542	0.00052
1	smallbank_db:x_update	100	0	14346.73059	0.006985	0.914227	0.001458	0
2	smallbank_db:x_update	100	0.2	14546.95332	0.006889	0.900798	0.003333	0.000082
3	smallbank_db:x_update	100	0.4	14533.07941	0.006895	0.898527	0.004583	0.00028
4	smallbank_db:x_update	100	0.6	14701.33323	0.006817	0.87709	0.015	0.00028
5	smallbank_db:x_update	100	0.7	14912.11174	0.006721	0.850319	0.03	0.00028
6	smallbank_db:x_update	100	0.8	15343.99315	0.006532	0.789534	0.075833	0.00028
7	smallbank_db:x_update	100	0.9	15737.70492	0.006369	0.712536	0.158542	0.00028
1	smallbank_db:x_update	300	0	17567.85278	0.017118	0.745645	0.002708	5394428022
2	smallbank_db:x_update	300	0.2	17652.53197	0.017036	0.738707	0.002986	936417403
3	smallbank_db:x_update	300	0.4	17541.78661	0.017144	0.740139	0.004583	936417403
4	smallbank_db:x_update	300	0.6	17257.66168	0.017424	0.751825	0.013264	29533861628095
5	smallbank_db:x_update	300	0.7	18777.3592	0.016019	0.670776	0.029792	0.00682
6	smallbank_db:x_update	300	0.8	19323.20475	0.015568	0.621219	0.073819	0
7	smallbank_db:x_update	300	0.9	20131.30082	0.014944	0.549724	0.168472	12787648669110
1	smallbank_db:x_update	500	0	17852.01572	0.028075	0.73459	0.002333	36857459350400
2	smallbank_db:x_update	500	0.2	17800.36446	0.028158	0.730978	0.002292	36857459350400
3	smallbank_db:x_update	500	0.4	17864.2673	0.028057	0.725739	0.003583	36857459350400
4	smallbank_db:x_update	500	0.6	18577.70629	0.026982	0.688498	0.013333	36857459350400
5	smallbank_db:x_update	500	0.7	19179.94545	0.026138	0.644176	0.030292	36857459350400
6	smallbank_db:x_update	500	0.8	20205.16663	0.024814	0.593206	0.071583	36857459350400
7	smallbank_db:x_update	500	0.9	21019.29571	0.023857	0.519147	0.172917	36857459350400
1	smallbank_db:occreset	100	0	12000.06	0.008329	0.278361	0.017083	0
2	smallbank_db:occreset	100	0.2	11638.33863	0.008588	0.266626	0.022083	0
3	smallbank_db:occreset	100	0.4	10614.04459	0.009417	0.27179	0.026875	0
4	smallbank_db:occreset	100	0.6	10125.64235	0.009872	0.279168	0.088125	0
5	smallbank_db:occreset	100	0.7	9819.123561	0.010179	0.262741	0.195208	0
6	smallbank_db:occreset	100	0.8	9072.815011	0.011017	0.229544	0.498333	0
7	smallbank_db:occreset	100	0.9	9109.490381	0.010972	0.246009	1.077292	0
1	smallbank_db:occreset	300	0	28408.65588	0.011795	0.188087	0.053681	0
2	smallbank_db:occreset	300	0.2	27160.89435	0.011034	0.188304	0.059444	0
3	smallbank_db:occreset	300	0.4	23858.12369	0.012562	0.185264	0.089722	0
4	smallbank_db:occreset	300	0.6	22567.80921	0.013279	0.183965	0.260347	0
5	smallbank_db:occreset	300	0.7	18242.71938	0.016429	0.178368	0.569722	0
6	smallbank_db:occreset	300	0.8	14855.74453	0.020177	0.168968	1.478611	0
7	smallbank_db:occreset	300	0.9	10949.34635	0.02738	0.159066	3.072708	0
1	smallbank_db:occreset	500	0	31388.09395	0.016432	0.169419	0.09	0
2	smallbank_db:occreset	500	0.2	30900.1603	0.01616	0.167372	0.097958	0
3	smallbank_db:occreset	500	0.4	27562.50948	0.018119	0.16404	0.145	0
4	smallbank_db:occreset	500	0.6	22143.9393	0.022555	0.162495	0.442125	0
5	smallbank_db:occreset	500	0.7	16621.7533	0.030055	0.156424	0.992292	0
6	smallbank_db:occreset	500	0.8	12974.7193	0.038509	0.155245	2.294917	0
7	smallbank_db:occreset	500	0.9	10513.75678	0.042655	0.147464	3.861167	0
}\dataSendPayment

\pgfplotstableread{
Id	name	BATCH	alpha	Throughput Latency exe	retry	delay
1	smallbank_db:bamboo_get_banlance	100	0	22491.91697	0.004462	0.658804	0	0.000033
2	smallbank_db:bamboo_get_banlance	100	0.2	22709.96068	0.00442	0.65356	0	0.000033
3	smallbank_db:bamboo_get_banlance	100	0.4	22861.0619	0.004391	0.649777	0	0.000038
4	smallbank_db:bamboo_get_banlance	100	0.6	23072.81878	0.00435	0.644426	0	0.001367
5	smallbank_db:bamboo_get_banlance	100	0.7	23319.53594	0.004305	0.636056	0	0.001782
6	smallbank_db:bamboo_get_banlance	100	0.8	23734.17722	0.004229	0.62475	0	0.003709
7	smallbank_db:bamboo_get_banlance	100	0.9	25015.1133	0.004014	0.592209	0	0.009084
1	smallbank_db:bamboo_get_banlance	300	0	27763.15713	0.01085	0.553947	0	0.000211
2	smallbank_db:bamboo_get_banlance	300	0.2	27757.32389	0.010852	0.553717	0	0.00019
3	smallbank_db:bamboo_get_banlance	300	0.4	27970.34367	0.01077	0.548265	0	0.000211
4	smallbank_db:bamboo_get_banlance	300	0.6	28502.99875	0.010574	0.537176	0	0.000905
5	smallbank_db:bamboo_get_banlance	300	0.7	29380.49992	0.010254	0.522934	0	0.001397
6	smallbank_db:bamboo_get_banlance	300	0.8	30849.12209	0.009768	0.497677	0	0.003051
7	smallbank_db:bamboo_get_banlance	300	0.9	32846.56548	0.009183	0.464973	0	0.007809
1	smallbank_db:bamboo_get_banlance	500	0	30419.17625	0.016509	0.509169	0	0.000125
2	smallbank_db:bamboo_get_banlance	500	0.2	30603.22788	0.016411	0.506751	0	0.000131
3	smallbank_db:bamboo_get_banlance	500	0.4	30958.63411	0.016223	0.5005	0	0.000173
4	smallbank_db:bamboo_get_banlance	500	0.6	32121.50495	0.015639	0.481856	0	0.000388
5	smallbank_db:bamboo_get_banlance	500	0.7	33390.14295	0.015048	0.463515	0	0.001317
6	smallbank_db:bamboo_get_banlance	500	0.8	35216.27931	0.014271	0.438775	0	0.002792
7	smallbank_db:bamboo_get_banlance	500	0.9	38613.26881	0.013022	0.399438	0	0.006919
1	smallbank_db:TPL_get_banlance	100	0	22963.5406	0.004372	0.639924	0.003125	0
2	smallbank_db:TPL_get_banlance	100	0.2	23056.85918	0.004354	0.633731	0.00375	0
3	smallbank_db:TPL_get_banlance	100	0.4	23034.8402	0.004358	0.631366	0.007083	0
4	smallbank_db:TPL_get_banlance	100	0.6	23022.57651	0.00436	0.62111	0.020208	0
5	smallbank_db:TPL_get_banlance	100	0.7	22805.66721	0.004401	0.606089	0.054167	0
6	smallbank_db:TPL_get_banlance	100	0.8	22460.86896	0.004468	0.58907	0.11625	0
7	smallbank_db:TPL_get_banlance	100	0.9	22005.83155	0.004561	0.532653	0.320417	0
1	smallbank_db:TPL_get_banlance	300	0	28303.44437	0.010644	0.529022	0.012917	0
2	smallbank_db:TPL_get_banlance	300	0.2	28238.94857	0.010669	0.53051	0.012014	0
3	smallbank_db:TPL_get_banlance	300	0.4	28317.02483	0.010638	0.524481	0.018611	0
4	smallbank_db:TPL_get_banlance	300	0.6	27859.84178	0.010816	0.509215	0.060556	0
5	smallbank_db:TPL_get_banlance	300	0.7	27331.58335	0.011026	0.487405	0.156528	0
6	smallbank_db:TPL_get_banlance	300	0.8	26617.52283	0.011317	0.454051	0.382361	0
7	smallbank_db:TPL_get_banlance	300	0.9	25582.99015	0.011773	0.411436	0.854306	0
1	smallbank_db:TPL_get_banlance	500	0	30799.47744	0.016304	0.483471	0.020333	0
2	smallbank_db:TPL_get_banlance	500	0.2	30938.40036	0.016232	0.479947	0.023417	0
3	smallbank_db:TPL_get_banlance	500	0.4	30961.82923	0.01622	0.47504	0.032208	0
4	smallbank_db:TPL_get_banlance	500	0.6	30604.6718	0.016408	0.45246	0.101625	0
5	smallbank_db:TPL_get_banlance	500	0.7	29903.41198	0.016793	0.427845	0.245958	0
6	smallbank_db:TPL_get_banlance	500	0.8	28904.55355	0.017371	0.394689	0.576875	0
7	smallbank_db:TPL_get_banlance	500	0.9	27611.40627	0.018184	0.345533	1.249958	0
1	smallbank_db:fx_get_banlance	100	0	70579.93177	0.001434	0.184929	0	0.000044
2	smallbank_db:fx_get_banlance	100	0.2	71768.18875	0.00141	0.182399	0	0.000039
3	smallbank_db:fx_get_banlance	100	0.4	72465.9561	0.001396	0.179184	0	0.00004
4	smallbank_db:fx_get_banlance	100	0.6	73379.90919	0.00138	0.178115	0	0.000044
5	smallbank_db:fx_get_banlance	100	0.7	73267.90103	0.001382	0.175974	0	0.00004
6	smallbank_db:fx_get_banlance	100	0.8	71559.55096	0.001414	0.183027	0	0.000042
7	smallbank_db:fx_get_banlance	100	0.9	72553.58385	0.001395	0.180029	0	0.000042
1	smallbank_db:fx_get_banlance	300	0	117560.6172	0.002597	0.1208	0	0.000037
2	smallbank_db:fx_get_banlance	300	0.2	117793.3381	0.002592	0.120554	0	0.000038
3	smallbank_db:fx_get_banlance	300	0.4	118782.4796	0.002571	0.11963	0	0.000042
4	smallbank_db:fx_get_banlance	300	0.6	118319.858	0.002581	0.120032	0	0.000039
5	smallbank_db:fx_get_banlance	300	0.7	119834.2293	0.002549	0.118834	0	0.000038
6	smallbank_db:fx_get_banlance	300	0.8	117228.5224	0.002605	0.121538	0	0.000038
7	smallbank_db:fx_get_banlance	300	0.9	119985.0019	0.002546	0.117602	0	0.00004
1	smallbank_db:fx_get_banlance	500	0	140404.2472	0.003642	0.103715	0	0.000038
2	smallbank_db:fx_get_banlance	500	0.2	140804.5808	0.003632	0.103891	0	0.000037
3	smallbank_db:fx_get_banlance	500	0.4	140848.3767	0.00363	0.103519	0	0.000035
4	smallbank_db:fx_get_banlance	500	0.6	141411.0466	0.003617	0.102853	0	0.000036
5	smallbank_db:fx_get_banlance	500	0.7	143006.9597	0.003579	0.102123	0	0.000037
6	smallbank_db:fx_get_banlance	500	0.8	143685.0424	0.003574	0.101253	0	0.00004
7	smallbank_db:fx_get_banlance	500	0.9	144129.4282	0.003562	0.101133	0	0.000042
1	smallbank_db:x_get_banlance	100	0	21474.49233	0.004673	0.675485	0	0
2	smallbank_db:x_get_banlance	100	0.2	21697.0727	0.004625	0.669587	0	0
3	smallbank_db:x_get_banlance	100	0.4	21719.75185	0.00462	0.666993	0	0
4	smallbank_db:x_get_banlance	100	0.6	21887.52548	0.004585	0.664831	0	0
5	smallbank_db:x_get_banlance	100	0.7	22059.12765	0.004549	0.658532	0	0
6	smallbank_db:x_get_banlance	100	0.8	22357.19343	0.004489	0.649403	0	0
7	smallbank_db:x_get_banlance	100	0.9	23262.57633	0.004314	0.625199	0	0
1	smallbank_db:x_get_banlance	300	0	26963.16158	0.01117	0.565517	0	0
2	smallbank_db:x_get_banlance	300	0.2	26879.16137	0.011204	0.567062	0	0
3	smallbank_db:x_get_banlance	300	0.4	27087.8323	0.011118	0.562402	0	0
4	smallbank_db:x_get_banlance	300	0.6	27525.25059	0.010942	0.553199	0	0
5	smallbank_db:x_get_banlance	300	0.7	28137.91486	0.010706	0.540649	0	0
6	smallbank_db:x_get_banlance	300	0.8	29066.71619	0.010369	0.521616	0	0
7	smallbank_db:x_get_banlance	300	0.9	30542.44658	0.009871	0.494657	0	0
1	smallbank_db:x_get_banlance	500	0	29779.22428	0.016862	0.516584	0	0
2	smallbank_db:x_get_banlance	500	0.2	30012.45517	0.01673	0.512268	0	0
3	smallbank_db:x_get_banlance	500	0.4	30124.60289	0.016669	0.510633	0	0
4	smallbank_db:x_get_banlance	500	0.6	31077.13345	0.016161	0.494669	0	0
5	smallbank_db:x_get_banlance	500	0.7	31860.07057	0.015765	0.482727	0	0
6	smallbank_db:x_get_banlance	500	0.8	33273.02597	0.015098	0.461607	0	0
7	smallbank_db:x_get_banlance	500	0.9	35521.50753	0.014147	0.431519	0	0
1	smallbank_db:occget_banlance	100	0	86967.55023	0.001146	0.15399	0	0
2	smallbank_db:occget_banlance	100	0.2	87090.62869	0.001145	0.153577	0	0
3	smallbank_db:occget_banlance	100	0.4	89144.76739	0.001118	0.151325	0	0
4	smallbank_db:occget_banlance	100	0.6	91512.23976	0.00109	0.145566	0	0
5	smallbank_db:occget_banlance	100	0.7	89460.44171	0.001114	0.149276	0	0
6	smallbank_db:occget_banlance	100	0.8	89265.78889	0.001117	0.149441	0	0
7	smallbank_db:occget_banlance	100	0.9	88305.0941	0.001129	0.151994	0	0
1	smallbank_db:occget_banlance	300	0	136457.9681	0.002184	0.103577	0	0
2	smallbank_db:occget_banlance	300	0.2	136055.7073	0.002191	0.104138	0	0
3	smallbank_db:occget_banlance	300	0.4	135738.9288	0.002196	0.103997	0	0
4	smallbank_db:occget_banlance	300	0.6	135608.5433	0.002198	0.104723	0	0
5	smallbank_db:occget_banlance	300	0.7	135133.867	0.002205	0.104128	0	0
6	smallbank_db:occget_banlance	300	0.8	135556.2041	0.002199	0.104267	0	0
7	smallbank_db:occget_banlance	300	0.9	134967.9451	0.002208	0.1054	0	0
1	smallbank_db:occget_banlance	500	0	151753.7037	0.003265	0.094741	0	0
2	smallbank_db:occget_banlance	500	0.2	151640.5613	0.003269	0.094745	0	0
3	smallbank_db:occget_banlance	500	0.4	154863.6877	0.0032	0.092296	0	0
4	smallbank_db:occget_banlance	500	0.6	154397.4319	0.00321	0.092886	0	0
5	smallbank_db:occget_banlance	500	0.7	152773.7993	0.003244	0.094283	0	0
6	smallbank_db:occget_banlance	500	0.8	152164.5406	0.003256	0.094104	0	0
7	smallbank_db:occget_banlance	500	0.9	154056.5002	0.003217	0.092971	0	0
}\dataGetBalance

\pgfplotstableread{
Id	name	BATCH	alpha	Throughput Latency exe	Retry	delay
1	smallbank_db:bamboo__mix_5	100	0	15167.79372	0.006608	0.996555	0.00125	0.000034
2	smallbank_db:bamboo__mix_5	100	0.2	15174.26697	0.006606	0.991483	0.003125	0.000044
3	smallbank_db:bamboo__mix_5	100	0.4	15304.65836	0.006549	0.982424	0.003542	0.000039
4	smallbank_db:bamboo__mix_5	100	0.6	15327.82383	0.006539	0.973517	0.0125	0.000041
5	smallbank_db:bamboo__mix_5	100	0.7	15556.73672	0.006444	0.936499	0.03875	0.000174
6	smallbank_db:bamboo__mix_5	100	0.8	15534.88553	0.006452	0.872973	0.120833	0.000445
7	smallbank_db:bamboo__mix_5	100	0.9	14816.32389	0.006765	0.77791	0.402292	0.000436
1	smallbank_db:bamboo__mix_5	300	0	18085.79449	0.016631	0.861174	0.002569	0.000039
2	smallbank_db:bamboo__mix_5	300	0.2	18239.9465	0.01649	0.854016	0.002431	0.000124
3	smallbank_db:bamboo__mix_5	300	0.4	18419.50421	0.01633	0.844884	0.004653	0.000075
4	smallbank_db:bamboo__mix_5	300	0.6	18766.07823	0.016029	0.819703	0.016042	0.000129
5	smallbank_db:bamboo__mix_5	300	0.7	19078.9125	0.015768	0.789699	0.039444	0.000091
6	smallbank_db:bamboo__mix_5	300	0.8	18969.73274	0.015858	0.727408	0.165764	0.00015
7	smallbank_db:bamboo__mix_5	300	0.9	16143.71494	0.018626	0.640879	0.79	0.000496
1	smallbank_db:bamboo__mix_5	500	0	18296.83876	0.027397	0.858478	0.00275	0.000044
2	smallbank_db:bamboo__mix_5	500	0.2	17939.94156	0.027944	0.875283	0.003625	0.000094
3	smallbank_db:bamboo__mix_5	500	0.4	18284.44656	0.027414	0.858186	0.004292	0.000042
4	smallbank_db:bamboo__mix_5	500	0.6	19336.81181	0.025926	0.803089	0.015208	0.000124
5	smallbank_db:bamboo__mix_5	500	0.7	19404.20991	0.025839	0.774986	0.055292	0.000211
6	smallbank_db:bamboo__mix_5	500	0.8	19446.89782	0.025779	0.700647	0.229917	0.000333
7	smallbank_db:bamboo__mix_5	500	0.9	15584.36499	0.032152	0.615852	1.146208	0.000723
1	smallbank_db:TPL__mix_5	100	0	15483.97086	0.006474	0.970356	0.005417	0
2	smallbank_db:TPL__mix_5	100	0.2	15455.55241	0.006486	0.970857	0.005625	0
3	smallbank_db:TPL__mix_5	100	0.4	14898.13401	0.006729	1.00495	0.009375	0
4	smallbank_db:TPL__mix_5	100	0.6	15166.45181	0.006609	0.964435	0.029583	0
5	smallbank_db:TPL__mix_5	100	0.7	15775.72182	0.006355	0.920382	0.073542	0
6	smallbank_db:TPL__mix_5	100	0.8	15692.06929	0.006388	0.847651	0.232708	0
7	smallbank_db:TPL__mix_5	100	0.9	15070.31243	0.006651	0.745334	0.738542	0
1	smallbank_db:TPL__mix_5	300	0	17735.73751	0.016959	0.874235	0.006042	0
2	smallbank_db:TPL__mix_5	300	0.2	18399.99284	0.016348	0.845016	0.005972	0
3	smallbank_db:TPL__mix_5	300	0.4	18679.44133	0.016104	0.830915	0.010347	0
4	smallbank_db:TPL__mix_5	300	0.6	19074.71726	0.015771	0.799577	0.0375	0
5	smallbank_db:TPL__mix_5	300	0.7	19153.31698	0.015706	0.763554	0.124306	0
6	smallbank_db:TPL__mix_5	300	0.8	18490.53067	0.016268	0.69149	0.485069	0
7	smallbank_db:TPL__mix_5	300	0.9	14925.86819	0.020143	0.580949	2.078472	0
1	smallbank_db:TPL__mix_5	500	0	18984.00756	0.026407	0.826512	0.006208	0
2	smallbank_db:TPL__mix_5	500	0.2	18601.65183	0.026949	0.84388	0.006958	0
3	smallbank_db:TPL__mix_5	500	0.4	19026.34276	0.026348	0.821708	0.0095	0
4	smallbank_db:TPL__mix_5	500	0.6	19158.35739	0.026167	0.802492	0.037875	0
5	smallbank_db:TPL__mix_5	500	0.7	19565.99365	0.025626	0.737251	0.184333	0
6	smallbank_db:TPL__mix_5	500	0.8	18060.03608	0.027756	0.65972	0.829042	0
7	smallbank_db:TPL__mix_5	500	0.9	13640.44802	0.036726	0.528331	3.453125	0
1	smallbank_db:fx__mix_5	100	0	15271.5472	0.006563	0.851925	0.001667	0.000036
2	smallbank_db:fx__mix_5	100	0.2	15261.78735	0.006567	0.851054	0.002917	0.000034
3	smallbank_db:fx__mix_5	100	0.4	15367.17955	0.006522	0.844038	0.003125	0.000035
4	smallbank_db:fx__mix_5	100	0.6	15376.48359	0.006518	0.832282	0.014792	0.000037
5	smallbank_db:fx__mix_5	100	0.7	15725.48544	0.006374	0.803645	0.026042	0.000034
6	smallbank_db:fx__mix_5	100	0.8	16033.72427	0.006251	0.758786	0.059583	0.000076
7	smallbank_db:fx__mix_5	100	0.9	16420.13656	0.006105	0.681661	0.14	0.000934
1	smallbank_db:fx__mix_5	300	0	18516.18495	0.016245	0.697561	0.002639	0.000036
2	smallbank_db:fx__mix_5	300	0.2	17263.57888	0.017419	0.759889	0.002222	0.00006
3	smallbank_db:fx__mix_5	300	0.4	18329.19229	0.016409	0.703991	0.003611	0.000033
4	smallbank_db:fx__mix_5	300	0.6	19264.6264	0.015615	0.661769	0.012708	0.000105
5	smallbank_db:fx__mix_5	300	0.7	19651.53999	0.015308	0.636309	0.027708	0.000049
6	smallbank_db:fx__mix_5	300	0.8	20288.17664	0.01483	0.589988	0.066181	0.000229
7	smallbank_db:fx__mix_5	300	0.9	21067.2308	0.014283	0.526589	0.149931	0.000361
1	smallbank_db:fx__mix_5	500	0	18855.48762	0.026587	0.684573	0.002458	0.000043
2	smallbank_db:fx__mix_5	500	0.2	18809.20491	0.026654	0.684576	0.003125	0.000037
3	smallbank_db:fx__mix_5	500	0.4	19410.55012	0.025828	0.663132	0.003375	0.000036
4	smallbank_db:fx__mix_5	500	0.6	19367.61509	0.025886	0.643825	0.012167	0.000065
5	smallbank_db:fx__mix_5	500	0.7	20418.51142	0.024558	0.608023	0.028875	0.000099
6	smallbank_db:fx__mix_5	500	0.8	21258.52223	0.02359	0.547948	0.066083	0.000203
7	smallbank_db:fx__mix_5	500	0.9	22770.9818	0.022028	0.47274	0.161542	0.000536
1	smallbank_db:x__mix_5	100	0	14586.07455	0.006871	0.911139	0.001875	0.000047
2	smallbank_db:x__mix_5	100	0.2	14705.74719	0.006815	0.898104	0.003125	0.00028
3	smallbank_db:x__mix_5	100	0.4	14812.3918	0.006766	0.89264	0.003958	0.00028
4	smallbank_db:x__mix_5	100	0.6	14885.79802	0.006733	0.877081	0.013333	0.00028
5	smallbank_db:x__mix_5	100	0.7	15178.29756	0.006603	0.849192	0.025833	0.00028
6	smallbank_db:x__mix_5	100	0.8	15456.6971	0.006485	0.799822	0.062917	0.00028
7	smallbank_db:x__mix_5	100	0.9	14971.46065	0.006698	0.719291	0.156667	0.00028
1	smallbank_db:x__mix_5	300	0	17905.54823	0.016797	0.738814	0.002639	0.006813
2	smallbank_db:x__mix_5	300	0.2	18093.99832	0.016623	0.730603	0.002222	0.006813
3	smallbank_db:x__mix_5	300	0.4	18118.2214	0.0166	0.727756	0.003889	0.006813
4	smallbank_db:x__mix_5	300	0.6	17772.33637	0.016923	0.736385	0.013403	0.006813
5	smallbank_db:x__mix_5	300	0.7	18030.60696	0.016681	0.716214	0.026389	0.006813
6	smallbank_db:x__mix_5	300	0.8	19754.49552	0.01523	0.618475	0.070417	0.006813
7	smallbank_db:x__mix_5	300	0.9	20443.91694	0.014717	0.552957	0.161319	0.006813
1	smallbank_db:x__mix_5	500	0	18225.75329	0.027504	0.719173	0.00225	0.262136
2	smallbank_db:x__mix_5	500	0.2	18437.55546	0.027188	0.718023	0.002833	0.262136
3	smallbank_db:x__mix_5	500	0.4	18589.47696	0.026966	0.709536	0.003708	0.262136
4	smallbank_db:x__mix_5	500	0.6	19367.72449	0.025885	0.669828	0.011083	0.262136
5	smallbank_db:x__mix_5	500	0.7	19576.28757	0.02561	0.646576	0.029125	0.262136
6	smallbank_db:x__mix_5	500	0.8	20492.13572	0.024469	0.593558	0.067208	0.262136
7	smallbank_db:x__mix_5	500	0.9	21351.09745	0.023488	0.515811	0.167542	0.262136
1	smallbank_db:occ_mix_5	100	0	11293.58618	0.00885	0.281059	0.02	 0
2	smallbank_db:occ_mix_5	100	0.2	11174.58899	0.008945	0.284166	0.019167	0
3	smallbank_db:occ_mix_5	100	0.4	11865.30865	0.0092	0.303354	0.029792	0
4	smallbank_db:occ_mix_5	100	0.6	12452.91529	0.006467	0.266747	0.082917	0
5	smallbank_db:occ_mix_5	100	0.7	12965.78923	0.007708	0.260637	0.182292	0
6	smallbank_db:occ_mix_5	100	0.8	10263.44993	0.009738	0.274411	0.465	0
7	smallbank_db:occ_mix_5	100	0.9	11525.68908	0.008671	0.208315	0.998125	0
1	smallbank_db:occ_mix_5	300	0	28737.63932	0.010427	0.175543	0.054236	0
2	smallbank_db:occ_mix_5	300	0.2	26929.67987	0.011128	0.183068	0.052222	0
3	smallbank_db:occ_mix_5	300	0.4	25884.80185	0.011577	0.184664	0.079653	0
4	smallbank_db:occ_mix_5	300	0.6	22621.51998	0.013247	0.177357	0.247431	0
5	smallbank_db:occ_mix_5	300	0.7	19705.53904	0.015208	0.174283	0.5625	0
6	smallbank_db:occ_mix_5	300	0.8	13885.87546	0.021587	0.166669	1.344653	0
7	smallbank_db:occ_mix_5	300	0.9	11663.75343	0.025701	0.164462	2.882083	0
1	smallbank_db:occ_mix_5	500	0	32299.30691	0.015459	0.161798	0.086792	0
2	smallbank_db:occ_mix_5	500	0.2	31548.47818	0.015824	0.160505	0.09025	0
3	smallbank_db:occ_mix_5	500	0.4	27487.77224	0.019595	0.158604	0.131333	0
4	smallbank_db:occ_mix_5	500	0.6	26258.55044	0.019015	0.157392	0.385917	0
5	smallbank_db:occ_mix_5	500	0.7	22106.42956	0.022591	0.157484	0.948708	0
6	smallbank_db:occ_mix_5	500	0.8	17570.38219	0.028428	0.162581	1.919958	0
7	smallbank_db:occ_mix_5	500	0.9	14197.865	0.035186	0.147666	3.382625	0
}\dataMIXA

\pgfplotstableread{
Id	name	BATCH	alpha	Throughput Latency exe	Retry	delay
1	smallbank_db:bamboo__mix_50	100	0	17185.26793	0.005834	0.86863	0.002083	0.000035
2	smallbank_db:bamboo__mix_50	100	0.2	17235.68361	0.005818	0.865679	0.002083	0.000078
3	smallbank_db:bamboo__mix_50	100	0.4	17205.28776	0.005828	0.867436	0.002917	0.000174
4	smallbank_db:bamboo__mix_50	100	0.6	17640.314	0.005685	0.842897	0.005833	0.000433
5	smallbank_db:bamboo__mix_50	100	0.7	17664.20349	0.005677	0.826915	0.01875	0.000705
6	smallbank_db:bamboo__mix_50	100	0.8	17794.05604	0.005635	0.79751	0.053958	0.001706
7	smallbank_db:bamboo__mix_50	100	0.9	17831.66904	0.005623	0.728332	0.155208	0.002712
1	smallbank_db:bamboo__mix_50	300	0	20655.17053	0.014567	0.753806	0.001597	0.000037
2	smallbank_db:bamboo__mix_50	300	0.2	20102.12999	0.014967	0.774757	0.000764	0.000073
3	smallbank_db:bamboo__mix_50	300	0.4	19756.47402	0.015228	0.787596	0.001944	0.000052
4	smallbank_db:bamboo__mix_50	300	0.6	20661.42383	0.01457	0.740892	0.007917	0.000214
5	smallbank_db:bamboo__mix_50	300	0.7	21761.10837	0.013829	0.696374	0.026181	0.000836
6	smallbank_db:bamboo__mix_50	300	0.8	22391.54097	0.01344	0.650662	0.067569	0.000785
7	smallbank_db:bamboo__mix_50	300	0.9	22057.945	0.013644	0.584639	0.279444	0.002278
1	smallbank_db:bamboo__mix_50	500	0	21648.26264	0.023167	0.724496	0.001125	0.000033
2	smallbank_db:bamboo__mix_50	500	0.2	21521.63328	0.023305	0.723939	0.001792	0.000057
3	smallbank_db:bamboo__mix_50	500	0.4	21979.57182	0.022817	0.713917	0.002042	0.000088
4	smallbank_db:bamboo__mix_50	500	0.6	22391.65239	0.022398	0.697533	0.006917	0.000308
5	smallbank_db:bamboo__mix_50	500	0.7	23060.2709	0.021751	0.66789	0.020667	0.000577
6	smallbank_db:bamboo__mix_50	500	0.8	23386.95768	0.021449	0.624656	0.086792	0.001125
7	smallbank_db:bamboo__mix_50	500	0.9	22367.48662	0.022423	0.548874	0.396167	0.00368
1	smallbank_db:TPL__mix_50	100	0	17014.69999	0.005896	0.809521	0.007083	0
2	smallbank_db:TPL__mix_50	100	0.2	16947.59662	0.005919	0.813677	0.007917	0
3	smallbank_db:TPL__mix_50	100	0.4	17543.47492	0.005716	0.813474	0.011875	0
4	smallbank_db:TPL__mix_50	100	0.6	17654.78279	0.005679	0.785619	0.035208	0
5	smallbank_db:TPL__mix_50	100	0.7	17314.51822	0.005791	0.762574	0.093542	0
6	smallbank_db:TPL__mix_50	100	0.8	17237.7881	0.005817	0.712471	0.229375	0
7	smallbank_db:TPL__mix_50	100	0.9	16404.47979	0.006111	0.635135	0.601458	0
1	smallbank_db:TPL__mix_50	300	0	21328.72001	0.014108	0.669003	0.018125	0
2	smallbank_db:TPL__mix_50	300	0.2	21460.05896	0.014022	0.663245	0.018542	0
3	smallbank_db:TPL__mix_50	300	0.4	21377.86263	0.014075	0.657071	0.028333	0
4	smallbank_db:TPL__mix_50	300	0.6	21046.69297	0.014296	0.631568	0.097292	0
5	smallbank_db:TPL__mix_50	300	0.7	20309.95244	0.014814	0.601186	0.273611	0
6	smallbank_db:TPL__mix_50	300	0.8	18324.57402	0.016414	0.568698	0.752153	0
7	smallbank_db:TPL__mix_50	300	0.9	16106.15747	0.01867	0.495611	1.936875	0
1	smallbank_db:TPL__mix_50	500	0	21947.07098	0.02285	0.638968	0.0305	0
2	smallbank_db:TPL__mix_50	500	0.2	21606.36091	0.023209	0.636798	0.035542	0
3	smallbank_db:TPL__mix_50	500	0.4	21972.24722	0.022824	0.619513	0.054125	0
4	smallbank_db:TPL__mix_50	500	0.6	21518.93173	0.023304	0.583536	0.175083	0
5	smallbank_db:TPL__mix_50	500	0.7	20510.9272	0.024446	0.536055	0.468875	0
6	smallbank_db:TPL__mix_50	500	0.8	18346.43575	0.027323	0.475205	1.364833	0
7	smallbank_db:TPL__mix_50	500	0.9	15794.92917	0.031726	0.413389	2.932417	0
1	smallbank_db:fx__mix_50	100	0	26770.32732	0.00375	0.469906	0.001042	0.000033
2	smallbank_db:fx__mix_50	100	0.2	26820.13745	0.003744	0.46962	0.002083	0.00004
3	smallbank_db:fx__mix_50	100	0.4	26313.62552	0.003816	0.467984	0.002083	0.000035
4	smallbank_db:fx__mix_50	100	0.6	26082.27872	0.003851	0.463107	0.006875	0.000216
5	smallbank_db:fx__mix_50	100	0.7	27249.50326	0.003684	0.446672	0.016458	0.000047
6	smallbank_db:fx__mix_50	100	0.8	27668.73606	0.003629	0.423696	0.032708	0.000032
7	smallbank_db:fx__mix_50	100	0.9	27905.67881	0.003599	0.398099	0.069583	0.000043
1	smallbank_db:fx__mix_50	300	0	33736.13656	0.008933	0.377818	0.00125	0.000031
2	smallbank_db:fx__mix_50	300	0.2	33931.05398	0.008882	0.375755	0.001458	0.000031
3	smallbank_db:fx__mix_50	300	0.4	34214.92198	0.00881	0.372066	0.001736	0.000029
4	smallbank_db:fx__mix_50	300	0.6	34381.4913	0.008767	0.367214	0.005694	0.00003
5	smallbank_db:fx__mix_50	300	0.7	33560.4099	0.008986	0.353052	0.012778	0.000032
6	smallbank_db:fx__mix_50	300	0.8	36440.93532	0.008274	0.326798	0.031389	0.000111
7	smallbank_db:fx__mix_50	300	0.9	37669.90698	0.008005	0.299053	0.063611	0.000253
1	smallbank_db:fx__mix_50	500	0	36787.1907	0.013659	0.34596	0.001083	0.000028
2	smallbank_db:fx__mix_50	500	0.2	36575.21737	0.013775	0.350411	0.001333	0.000031
3	smallbank_db:fx__mix_50	500	0.4	36826.26771	0.014424	0.369933	0.001708	0.000028
4	smallbank_db:fx__mix_50	500	0.6	37685.66076	0.013335	0.33488	0.004667	0.000055
5	smallbank_db:fx__mix_50	500	0.7	39118.2741	0.012849	0.31619	0.012292	0.000167
6	smallbank_db:fx__mix_50	500	0.8	40400.63968	0.012444	0.296221	0.028583	0.00012
7	smallbank_db:fx__mix_50	500	0.9	41809.00601	0.012027	0.269425	0.06375	0.000264
1	smallbank_db:x__mix_50	100	0	16556.86248	0.006055	0.839976	0.000833	0.00028
2	smallbank_db:x__mix_50	100	0.2	16531.37345	0.006064	0.84	0.001667	0.00028
3	smallbank_db:x__mix_50	100	0.4	16546.87419	0.006059	0.839499	0.002083	0.00028
4	smallbank_db:x__mix_50	100	0.6	16855.8265	0.005948	0.821524	0.006458	0.00028
5	smallbank_db:x__mix_50	100	0.7	16920.29484	0.005925	0.804314	0.0175	0.00028
6	smallbank_db:x__mix_50	100	0.8	17233.08465	0.005818	0.776061	0.035833	0.00028
7	smallbank_db:x__mix_50	100	0.9	17592.5349	0.0057	0.724	0.079375	0.00028
1	smallbank_db:x__mix_50	300	0	20383.37736	0.014759	0.69291	0.001181	0.006813
2	smallbank_db:x__mix_50	300	0.2	20494.25307	0.01468	0.690843	0.001181	0.006813
3	smallbank_db:x__mix_50	300	0.4	20584.9264	0.014617	0.685196	0.002361	0.006813
4	smallbank_db:x__mix_50	300	0.6	20918.7698	0.014384	0.671996	0.00625	0.006813
5	smallbank_db:x__mix_50	300	0.7	21580.08159	0.013944	0.639896	0.017292	0.006813
6	smallbank_db:x__mix_50	300	0.8	22312.22865	0.013488	0.606526	0.036667	0.006813
7	smallbank_db:x__mix_50	300	0.9	23129.6691	0.013013	0.556923	0.085139	0.006813
1	smallbank_db:x__mix_50	500	0	21080.97995	0.023788	0.670964	0.001375	0.262136
2	smallbank_db:x__mix_50	500	0.2	21424.51607	0.023408	0.658796	0.001667	0.262136
3	smallbank_db:x__mix_50	500	0.4	21314.53865	0.023527	0.665523	0.002125	0.262136
4	smallbank_db:x__mix_50	500	0.6	21371.25068	0.023466	0.654386	0.006333	0.262136
5	smallbank_db:x__mix_50	500	0.7	22055.19496	0.022741	0.625114	0.015583	0.262136
6	smallbank_db:x__mix_50	500	0.8	22874.44029	0.021928	0.594168	0.036542	0.262136
7	smallbank_db:x__mix_50	500	0.9	24665.49985	0.020341	0.519551	0.082375	0.262136
1	smallbank_db:occ_mix_50	100	0	14171.37324	0.007051	0.258704	0.006667	0
2	smallbank_db:occ_mix_50	100	0.2	11982.23634	0.008342	0.275622	0.005625	0
3	smallbank_db:occ_mix_50	100	0.4	11634.92514	0.008591	0.257346	0.014792	0
4	smallbank_db:occ_mix_50	100	0.6	10887.31628	0.009181	0.268063	0.033958	0
5	smallbank_db:occ_mix_50	100	0.7	10365.76017	0.009643	0.298422	0.074583	0
6	smallbank_db:occ_mix_50	100	0.8	9995.127375	0.01	0.253654	0.188958	0
7	smallbank_db:occ_mix_50	100	0.9	9848.153779	0.010149	0.248254	0.429375	0
1	smallbank_db:occ_mix_50	300	0	28768.75713	0.010415	0.168341	0.024583	0
2	smallbank_db:occ_mix_50	300	0.2	28793.72297	0.010407	0.166985	0.024097	0
3	smallbank_db:occ_mix_50	300	0.4	29021.43354	0.010325	0.167727	0.033611	0
4	smallbank_db:occ_mix_50	300	0.6	28184.18055	0.010264	0.164592	0.102083	0
5	smallbank_db:occ_mix_50	300	0.7	25899.10675	0.01203	0.169718	0.227014	0
6	smallbank_db:occ_mix_50	300	0.8	23614.45617	0.010846	0.160767	0.515208	0
7	smallbank_db:occ_mix_50	300	0.9	21715.26463	0.013797	0.155493	1.154653	0
1	smallbank_db:occ_mix_50	500	0	42200.62176	0.011825	0.147559	0.037792	0
2	smallbank_db:occ_mix_50	500	0.2	42027.37047	0.01179	0.14461	0.036875	0
3	smallbank_db:occ_mix_50	500	0.4	41261.95111	0.011534	0.151902	0.054583	0
4	smallbank_db:occ_mix_50	500	0.6	38958.69891	0.012808	0.146062	0.158583	0
5	smallbank_db:occ_mix_50	500	0.7	36005.23205	0.013128	0.143156	0.376042	0
6	smallbank_db:occ_mix_50	500	0.8	24143.36329	0.020681	0.140657	0.826625	0
7	smallbank_db:occ_mix_50	500	0.9	19061.11311	0.026201	0.132801	1.728333	0
}\dataMIXB

\pgfplotstableread{
Id	name	BATCH	alpha	Throughput Latency exe	Retry	delay
1	smallbank_db:bamboo__mix_95	100	0	20360.55143	0.004927	0.728853	0	0.000214
2	smallbank_db:bamboo__mix_95	100	0.2	20407.20885	0.004916	0.724649	0	0.000474
3	smallbank_db:bamboo__mix_95	100	0.4	20249.23433	0.004954	0.72923	0.000208	0.000199
4	smallbank_db:bamboo__mix_95	100	0.6	20434.22733	0.004909	0.723904	0.000625	0.001431
5	smallbank_db:bamboo__mix_95	100	0.7	20660.53451	0.004855	0.71424	0.001875	0.000484
6	smallbank_db:bamboo__mix_95	100	0.8	21156.27879	0.004742	0.696726	0.0025	0.004816
7	smallbank_db:bamboo__mix_95	100	0.9	22105.75763	0.004539	0.660694	0.012083	0.009664
1	smallbank_db:bamboo__mix_95	300	0	25889.08326	0.01163	0.591671	0.000139	0.000054
2	smallbank_db:bamboo__mix_95	300	0.2	26015.88414	0.011574	0.588296	0.000139	0.000113
3	smallbank_db:bamboo__mix_95	300	0.4	26078.9722	0.011546	0.586506	0.000069	0.000107
4	smallbank_db:bamboo__mix_95	300	0.6	26706.23145	0.011276	0.572398	0.000625	0.000753
5	smallbank_db:bamboo__mix_95	300	0.7	27338.0694	0.011016	0.558041	0.001528	0.001526
6	smallbank_db:bamboo__mix_95	300	0.8	28516.20671	0.010563	0.533559	0.005278	0.004013
7	smallbank_db:bamboo__mix_95	300	0.9	30497.81326	0.00988	0.494335	0.010486	0.00944
1	smallbank_db:bamboo__mix_95	500	0	27497.99207	0.018257	0.555156	0.000042	0.000078
2	smallbank_db:bamboo__mix_95	500	0.2	27277.0048	0.0184	0.559402	0.000125	0.000083
3	smallbank_db:bamboo__mix_95	500	0.4	27586.36544	0.0182	0.5534	0.00025	0.000189
4	smallbank_db:bamboo__mix_95	500	0.6	28583.57659	0.017564	0.5308	0.000375	0.000749
5	smallbank_db:bamboo__mix_95	500	0.7	29952.86168	0.016764	0.509898	0.001583	0.001291
6	smallbank_db:bamboo__mix_95	500	0.8	32115.44432	0.015639	0.475989	0.005333	0.004774
7	smallbank_db:bamboo__mix_95	500	0.9	32585.27215	0.014526	0.436141	0.016125	0.011761
1	smallbank_db:TPL__mix_95	100	0	20441.5372	0.00491	0.707638	0.003542	0
2	smallbank_db:TPL__mix_95	100	0.2	20458.00356	0.004906	0.706088	0.004792	0
3	smallbank_db:TPL__mix_95	100	0.4	20443.19135	0.004908	0.70336	0.010208	0
4	smallbank_db:TPL__mix_95	100	0.6	20373.16854	0.004926	0.695812	0.021667	0
5	smallbank_db:TPL__mix_95	100	0.7	20162.0525	0.004976	0.680834	0.055625	0
6	smallbank_db:TPL__mix_95	100	0.8	19946.80851	0.005029	0.648838	0.136667	0
7	smallbank_db:TPL__mix_95	100	0.9	19616.9785	0.005114	0.597994	0.3575	0
1	smallbank_db:TPL__mix_95	300	0	25800.48698	0.011671	0.573209	0.013056	0
2	smallbank_db:TPL__mix_95	300	0.2	26195.10625	0.011496	0.563491	0.013681	0
3	smallbank_db:TPL__mix_95	300	0.4	26107.19906	0.011534	0.56249	0.019514	0
4	smallbank_db:TPL__mix_95	300	0.6	25821.0743	0.011662	0.542038	0.067014	0
5	smallbank_db:TPL__mix_95	300	0.7	25318.94842	0.011893	0.520082	0.165347	0
6	smallbank_db:TPL__mix_95	300	0.8	24466.20351	0.012307	0.482961	0.418889	0
7	smallbank_db:TPL__mix_95	300	0.9	23527.56647	0.012796	0.431335	0.957222	0
1	smallbank_db:TPL__mix_95	500	0	27080.79216	0.018535	0.527858	0.023583	0
2	smallbank_db:TPL__mix_95	500	0.2	27369.54705	0.018339	0.523466	0.025292	0
3	smallbank_db:TPL__mix_95	500	0.4	27191.24623	0.018458	0.51624	0.035958	0
4	smallbank_db:TPL__mix_95	500	0.6	26734.92985	0.018771	0.496312	0.107	0
5	smallbank_db:TPL__mix_95	500	0.7	26098.44551	0.019229	0.468931	0.260333	0
6	smallbank_db:TPL__mix_95	500	0.8	24941.07671	0.020119	0.426381	0.704625	0
7	smallbank_db:TPL__mix_95	500	0.9	24243.10991	0.020699	0.379414	1.362917	0
1	smallbank_db:fx__mix_95	100	0	65608.72596	0.00154	0.195536	0	0.000028
2	smallbank_db:fx__mix_95	100	0.2	65854.46164	0.001534	0.193655	0	0.00003
3	smallbank_db:fx__mix_95	100	0.4	65886.09941	0.001534	0.19459	0	0.000029
4	smallbank_db:fx__mix_95	100	0.6	65270.60103	0.001548	0.197464	0.000625	0.000027
5	smallbank_db:fx__mix_95	100	0.7	65462.87709	0.001543	0.193951	0	0.000029
6	smallbank_db:fx__mix_95	100	0.8	64964.06675	0.001555	0.198385	0.000208	0.000031
7	smallbank_db:fx__mix_95	100	0.9	65975.75391	0.001531	0.193866	0.000417	0.000035
1	smallbank_db:fx__mix_95	300	0	108082.9537	0.002817	0.127819	0	0.000028
2	smallbank_db:fx__mix_95	300	0.2	107114.9626	0.002842	0.12916	0	0.000031
3	smallbank_db:fx__mix_95	300	0.4	108056.1892	0.002818	0.127839	0.000069	0.000027
4	smallbank_db:fx__mix_95	300	0.6	108138.9586	0.002815	0.127063	0.000069	0.000029
5	smallbank_db:fx__mix_95	300	0.7	108710.4226	0.002801	0.126081	0.000139	0.000029
6	smallbank_db:fx__mix_95	300	0.8	109076.8613	0.002793	0.126047	0.000347	0.000028
7	smallbank_db:fx__mix_95	300	0.9	109038.8678	0.0028	0.124715	0.000417	0.000029
1	smallbank_db:fx__mix_95	500	0	123601.7551	0.004126	0.11379	0	0.000035
2	smallbank_db:fx__mix_95	500	0.2	123081.9726	0.004145	0.115349	0	0.000031
3	smallbank_db:fx__mix_95	500	0.4	123950.9361	0.004117	0.114252	0	0.000032
4	smallbank_db:fx__mix_95	500	0.6	124889.4208	0.004086	0.112439	0.000042	0.000033
5	smallbank_db:fx__mix_95	500	0.7	125772.3207	0.004053	0.112013	0.000083	0.00003
6	smallbank_db:fx__mix_95	500	0.8	127779.8779	0.003985	0.109957	0.000333	0.000028
7	smallbank_db:fx__mix_95	500	0.9	128474.1553	0.003958	0.109399	0.000583	0.000025
1	smallbank_db:x__mix_95	100	0	19684.47428	0.005096	0.736295	0.000208	0.00028
2	smallbank_db:x__mix_95	100	0.2	19752.51742	0.005079	0.733316	0	0.00028
3	smallbank_db:x__mix_95	100	0.4	19748.94158	0.00508	0.733554	0.000417	0.00028
4	smallbank_db:x__mix_95	100	0.6	19847.83328	0.005054	0.729706	0.001042	0.00028
5	smallbank_db:x__mix_95	100	0.7	19971.78985	0.005023	0.726673	0.003125	0.00028
6	smallbank_db:x__mix_95	100	0.8	20307.57515	0.00494	0.708177	0.005833	0.00028
7	smallbank_db:x__mix_95	100	0.9	21007.6678	0.004776	0.674346	0.016458	0.00028
1	smallbank_db:x__mix_95	300	0	24999.47918	0.012044	0.604064	0.000556	0.006813
2	smallbank_db:x__mix_95	300	0.2	25006.9898	0.01204	0.603177	0.000347	0.006813
3	smallbank_db:x__mix_95	300	0.4	25094.84632	0.011998	0.603246	0.000208	0.006813
4	smallbank_db:x__mix_95	300	0.6	25616.12108	0.011755	0.590562	0.001458	0.006813
5	smallbank_db:x__mix_95	300	0.7	26196.29759	0.011496	0.574955	0.003194	0.006813
6	smallbank_db:x__mix_95	300	0.8	27207.79957	0.01107	0.549455	0.008194	0.006813
7	smallbank_db:x__mix_95	300	0.9	28703.15555	0.010496	0.515392	0.016319	0.006813
1	smallbank_db:x__mix_95	500	0	26449.3121	0.018981	0.569319	0.000208	0.262136
2	smallbank_db:x__mix_95	500	0.2	26600.81487	0.018867	0.57216	0.000375	0.262136
3	smallbank_db:x__mix_95	500	0.4	27367.39358	0.018342	0.556382	0.000417	0.262136
4	smallbank_db:x__mix_95	500	0.6	27987.06531	0.017936	0.541855	0.0015	0.262136
5	smallbank_db:x__mix_95	500	0.7	28667.97823	0.017516	0.51954	0.003208	0.262136
6	smallbank_db:x__mix_95	500	0.8	30216.42515	0.016617	0.488792	0.010125	0.262136
7	smallbank_db:x__mix_95	500	0.9	32775.73612	0.015326	0.453038	0.0185	0.262136
1	smallbank_db:occ_mix_95	100	0	12421.58872	0.008047	0.264696	0	0
2	smallbank_db:occ_mix_95	100	0.2	12399.48026	0.008061	0.282095	0.000417	0
3	smallbank_db:occ_mix_95	100	0.4	12217.43989	0.008182	0.262474	0.001042	0
4	smallbank_db:occ_mix_95	100	0.6	14083.80451	0.007097	0.249298	0.003125	0
5	smallbank_db:occ_mix_95	100	0.7	12812.50501	0.007802	0.290402	0.004583	0
6	smallbank_db:occ_mix_95	100	0.8	12564.75429	0.007956	0.25931	0.009167	0
7	smallbank_db:occ_mix_95	100	0.9	12437.81095	0.008036	0.269563	0.028125	0
1	smallbank_db:occ_mix_95	300	0	31591.20975	0.009485	0.154459	0.001458	0
2	smallbank_db:occ_mix_95	300	0.2	32770.32101	0.009141	0.150214	0.001458	0
3	smallbank_db:occ_mix_95	300	0.4	33008.29333	0.009076	0.150942	0.001528	0
4	smallbank_db:occ_mix_95	300	0.6	31953.63173	0.009375	0.151671	0.005972	0
5	smallbank_db:occ_mix_95	300	0.7	31418.34177	0.009536	0.167273	0.012986	0
6	smallbank_db:occ_mix_95	300	0.8	31828.90197	0.009411	0.148085	0.035139	0
7	smallbank_db:occ_mix_95	300	0.9	30979.13426	0.009669	0.1492	0.080972	0
1	smallbank_db:occ_mix_95	500	0	50239.68517	0.009932	0.12954	0.002333	0
2	smallbank_db:occ_mix_95	500	0.2	49669.69652	0.010046	0.128389	0.002083	0
3	smallbank_db:occ_mix_95	500	0.4	50400.57961	0.009899	0.129403	0.003125	0
4	smallbank_db:occ_mix_95	500	0.6	50414.66058	0.009895	0.130271	0.009042	0
5	smallbank_db:occ_mix_95	500	0.7	51568.32158	0.009672	0.127821	0.019125	0
6	smallbank_db:occ_mix_95	500	0.8	50508.66434	0.009875	0.126465	0.051292	0
7	smallbank_db:occ_mix_95	500	0.9	50407.88379	0.009892	0.122699	0.120292	0
}\dataMIXC

\pgfplotstableread{
Id type Node	MThroughput Latency Throughput
1	commit	8	1919	0.00260552371	115140
2	commit	16	3679	0.001359064963	220740
3	commit	24	5403	0.0009254118083	324180
4	commit	32	7296	0.0006853070175	437760
5	commit	48	10956	0.0004563709383	657360
6	commit	64	13834	0.0003614283649	830040
7	commit	128	20208	0.0002474267617	1212480
1	seq	8	193	0.02590673575	11580
2	seq	16	187	0.02673796791	11220
3	seq	24	188	0.02659574468	11280
4	seq	32	188	0.02659574468	11280
5	seq	48	178	0.02808988764	10680
6	seq	64	170	0.02941176471	10200
7	seq	128	150	0.03333333333	9000
1	aacc	8	961	0.005202913632	57660
2	aacc	16	946	0.005285412262	56760
3	aacc	24	954	0.005241090147	57240
4	aacc	32	955	0.005235602094	57300
5	aacc	48	950	0.005263157895	57000
6	aacc	64	955	0.005235602094	57300
7	aacc	128	942	0.005307855626	56520
}\dataExecuteTime

\pgfplotstableread{
Id type Node	Throughput Latency 
1	8	commit	114386.087	1.702131508
2	16	commit	222878.8235	1.808451181
3	24	commit	325521.8182	1.947817287
4	32	commit	427423.6364	1.959339617
5	48	commit	630671.25	1.983215486
6	64	commit	806468.2982	2.047286125
7	128	commit	978141.4513	5.601759624
1	8	aacc	113440	1.742665639
2	16	aacc	217167.6923	1.867144567
3	24	aacc	312650.1266	1.987747766
4	32	aacc	397890.0561	2.18515
5	48	aacc	546245.3431	2.705963461
6	64	aacc	712088.6497	3.259097801
7	128	aacc	803658.3134	7.164293893
}\dataXPExecute

\pgfplotstableread{
Id Node	Batch Throughput Latency 
1	16	1	94728.46154	4.047645937
2	16	2	118423.8462	3.327337969
3	16	4	116255.3846	3.404689375
4	16	8	115038.8235	3.411107813
5	16	16	115200	3.382413984
1	64	1	93107.02194	16.10532296
2	64	2	112422.2652	12.8187106
3	64	4	109380.9492	12.85935963
4	64	8	106731.0508	12.75315562
5	64	16	105939.4576	12.84538422
}\dataXPValidate

\pgfplotstableread{
Id Node	Block1 Block2
1 1	104	108
2 2	196	204
3 4	425	415
4 8	649	612
5 16	1644 1599
}\dataXPValidatePal

\pgfplotstableread{
Id Workers Throughput Latency
1	1	11978.82353	16.66579336
2	2	35452.94118	9.978098712
3	4	58583.52941	6.416423022
4	8	90249.41176	4.19636089
5	16	112461.1765	3.423968922
6	24	108073.5849	3.524736625
7	32	99636.22642	3.822109688
}\dataValidateWorkers

\pgfplotstableread{
Id Type Node Throughput Latency
1	8	aacc 113827.5	1.803057321
2	16	aacc 112461.1765	3.423968922
3	24	aacc 116832.9114	5.062022308
4	32	aacc 114456.3158	6.941720603
5	48	aacc 113939.1892	9.988809226
6	64	aacc 112514.8571	12.80750824
7	128	aacc 110389.5678	25.7711051
1	8	seq 11627.30769	8.952478611
2	16	seq 11356.57895	17.26493241
3	24	seq 11128.17391	25.92402472
4	32	seq 10870.39666	33.9271253
5	48	seq 10400.26281	48.2997843
6	64	seq 10135.26923	62.15454516
7	128	seq 9017.76	92.16592373
}\dataPerformance

\newcommand{\EvalEEWorker}[6]{\begin{tikzpicture}
    \begin{axis}[title={\normalsize \vspace{4cm} #1},ylabel={\footnotesize #2},xlabel={\footnotesize Number of Executors},
            xtick={1,2,3,4,5},
            xticklabels={1, 4, 8, 12, 16},
            width=100,
            height=80,
            grid=both,
            legend to name={workermain},
            legend style={font=\footnotesize, only marks, draw=none, legend columns=3},
            legend pos={outer north east},
            y tick label style={
                /pgf/number format/.cd,
                scaled y ticks=base 10:#6,
                fixed,
                /tikz/.cd,
                font=\scriptsize
            },
            x tick label style={
                font=\scriptsize
            }
        ]

        \addplot[color=slateblue, mark=*, mark size=1pt, select coords index={50}{54}] table[x={Id}, y={#3}] {\dataEEWorkerFX};
        \addplot[color=slateblue, mark=triangle, mark size=1pt, select coords index={30}{34}] table[x={Id}, y={#3}] {\dataEEWorkerFX};

        \addplot[color=lightyellow, mark=*, mark size=1pt, select coords index={50}{54}] table[x={Id}, y={#3}] {\dataEEWorkerOCC};
        \addplot[color=lightyellow, mark=triangle, mark size=1pt, select coords index={30}{34}] table[x={Id}, y={#3}] {\dataEEWorkerOCC};

        \addplot[color=lightblue, mark=*, mark size=1pt, select coords index={50}{54}] table[x={Id}, y={#3}] {\dataEEWorkerTPL};
        \addplot[color=lightblue, mark=triangle, mark size=1pt, select coords index={30}{34}] table[x={Id}, y={#3}] {\dataEEWorkerTPL};

        \legend{
                \sysname{}-b500, \sysname{}-b300, 
                \OCC{}-b500, \OCC{}-b300, 
                \TPL{}-b500, \TPL{}-b300, 
                };
    \end{axis}
\end{tikzpicture}}

\newcommand{\EvalEEWorkerTPS}{\EvalEEWorker{}{Throughput (tps)}{Throughput}{8000}{150}{0}}
\newcommand{\EvalEEWorkerLat}{\EvalEEWorker{}{Latency (s)}{Latency}{8000}{150}{0}}

\newcommand{\EvalEEWorkerQ}[6]{\begin{tikzpicture}
    \begin{axis}[title={\normalsize \vspace{4cm} #1},ylabel={\footnotesize #2},xlabel={\footnotesize Number of Executors},
            xtick={1,2,3,4,5},
            xticklabels={1, 4, 8, 12, 16},
            width=100,
            height=80,
            grid=both,
            legend to name={workermain},
            legend style={font=\footnotesize, only marks, draw=none, legend columns=6},
            y tick label style={
                /pgf/number format/.cd,
                scaled y ticks=base 10:#6,
                fixed,
                /tikz/.cd,
                font=\scriptsize
            },
            x tick label style={
                font=\scriptsize
            }
        ]

        \addplot[color=slateblue, mark=*, mark size=1pt, select coords index={50}{54}] table[x={Id}, y={#3}] {\dataEEWorkerFX};
        \addplot[color=slateblue, mark=triangle, mark size=1pt, select coords index={30}{34}] table[x={Id}, y={#3}] {\dataEEWorkerFX};

        \addplot[color=lightyellow, mark=*, mark size=1pt, select coords index={50}{54}] table[x={Id}, y={#3}] {\dataEEWorkerOCC};
        \addplot[color=lightyellow, mark=triangle, mark size=1pt, select coords index={30}{34}] table[x={Id}, y={#3}] {\dataEEWorkerOCC};

        \addplot[color=lightblue, mark=*, mark size=1pt, select coords index={50}{54}] table[x={Id}, y={#3}] {\dataEEWorkerTPL};
        \addplot[color=lightblue, mark=triangle, mark size=1pt, select coords index={30}{34}] table[x={Id}, y={#3}] {\dataEEWorkerTPL};

        \legend{
                \sysname{}-b500, \sysname{}-b300, 
                \OCC{}-b500, \OCC{}-b300, 
                \TPL{}-b500, \TPL{}-b300, 
                };
    \end{axis}
\end{tikzpicture}}

\newcommand{\EvalEEWorkerRetry}{\EvalEEWorkerQ{}{\# of Re-executions}{Retry}{8000}{150}{0}}

\pgfplotstableread{
Id idx name worker	batch	alpha	Throughput	Latency Retry
1	0	smallbank_db:x__mix_50	1	100	0.75	8170.713439	0.012252	0
2	1	smallbank_db:x__mix_50	4	100	0.75	17743.21137	0.005651	0.001667
3	2	smallbank_db:x__mix_50	8	100	0.75	17338.09649	0.005786	0.010625
4	3	smallbank_db:x__mix_50	12	100	0.75	17062.90524	0.005877	0.013125
5	4	smallbank_db:x__mix_50	16	100	0.75	17034.14281	0.005886	0.022917
1	5	smallbank_db:x__mix_50	1	100	0.8	8012.125016	0.012494	0
2	6	smallbank_db:x__mix_50	4	100	0.8	18406.95476	0.005447	0.004792
3	7	smallbank_db:x__mix_50	8	100	0.8	18455.78877	0.005433	0.018125
4	8	smallbank_db:x__mix_50	12	100	0.8	17914.52532	0.005597	0.028958
5	9	smallbank_db:x__mix_50	16	100	0.8	17505.66199	0.005727	0.036875
1	10	smallbank_db:x__mix_50	1	100	0.85	8549.078837	0.01171	0
2	11	smallbank_db:x__mix_50	4	100	0.85	20708.39985	0.004843	0.006458
3	12	smallbank_db:x__mix_50	8	100	0.85	21659.77014	0.004632	0.032917
4	13	smallbank_db:x__mix_50	12	100	0.85	20327.69946	0.004934	0.043542
5	14	smallbank_db:x__mix_50	16	100	0.85	19613.77214	0.005114	0.061667
1	15	smallbank_db:x__mix_50	1	100	0.9	8432.192001	0.011873	0
2	16	smallbank_db:x__mix_50	4	100	0.9	18975.25706	0.005284	0.017292
3	17	smallbank_db:x__mix_50	8	100	0.9	19057.90427	0.005262	0.044167
4	18	smallbank_db:x__mix_50	12	100	0.9	18334.25641	0.005469	0.064583
5	19	smallbank_db:x__mix_50	16	100	0.9	17839.15769	0.00562	0.083125
1	20	smallbank_db:x__mix_50	1	300	0.75	8434.696122	0.035602	0
2	21	smallbank_db:x__mix_50	4	300	0.75	21774.76426	0.013817	0.002917
3	22	smallbank_db:x__mix_50	8	300	0.75	22851.30308	0.013169	0.010347
4	23	smallbank_db:x__mix_50	12	300	0.75	22538.70318	0.013352	0.015833
5	24	smallbank_db:x__mix_50	16	300	0.75	22119.95158	0.013604	0.027292
1	25	smallbank_db:x__mix_50	1	300	0.8	8736.623562	0.034373	0
2	26	smallbank_db:x__mix_50	4	300	0.8	22235.7421	0.013532	0.006319
3	27	smallbank_db:x__mix_50	8	300	0.8	23452.31038	0.012833	0.015
4	28	smallbank_db:x__mix_50	12	300	0.8	22955.70665	0.01311	0.033819
5	29	smallbank_db:x__mix_50	16	300	0.8	22532.49608	0.013356	0.039306
1	30	smallbank_db:x__mix_50	1	300	0.85	9204.489234	0.032628	0
2	31	smallbank_db:x__mix_50	4	300	0.85	24498.12862	0.012286	0.008333
3	32	smallbank_db:x__mix_50	8	300	0.85	27607.52075	0.010909	0.030417
4	33	smallbank_db:x__mix_50	12	300	0.85	26486.75295	0.011373	0.053472
5	34	smallbank_db:x__mix_50	16	300	0.85	25329.19154	0.011895	0.07125
1	35	smallbank_db:x__mix_50	1	300	0.9	8898.946896	0.033746	0
2	36	smallbank_db:x__mix_50	4	300	0.9	23361.15041	0.012882	0.012569
3	37	smallbank_db:x__mix_50	8	300	0.9	24836.79295	0.012119	0.043472
4	38	smallbank_db:x__mix_50	12	300	0.9	24113.29229	0.012483	0.070556
5	39	smallbank_db:x__mix_50	16	300	0.9	23444.78834	0.012838	0.093056
1	40	smallbank_db:x__mix_50	1	500	0.75	8750.271623	0.057197	0
2	41	smallbank_db:x__mix_50	4	500	0.75	22833.33096	0.021959	0.004083
3	42	smallbank_db:x__mix_50	8	500	0.75	24653.79396	0.020346	0.010792
4	43	smallbank_db:x__mix_50	12	500	0.75	23781.18928	0.021091	0.018542
5	44	smallbank_db:x__mix_50	16	500	0.75	22987.20666	0.021818	0.024208
1	45	smallbank_db:x__mix_50	1	500	0.8	8858.841013	0.056497	0
2	46	smallbank_db:x__mix_50	4	500	0.8	23022.64277	0.021779	0.005208
3	47	smallbank_db:x__mix_50	8	500	0.8	25078.21268	0.020002	0.017417
4	48	smallbank_db:x__mix_50	12	500	0.8	24565.87493	0.020419	0.029583
5	49	smallbank_db:x__mix_50	16	500	0.8	24144.52915	0.020775	0.039042
1	50	smallbank_db:x__mix_50	1	500	0.85	9557.290381	0.052384	0
2	51	smallbank_db:x__mix_50	4	500	0.85	26273.38331	0.0191	0.007833
3	52	smallbank_db:x__mix_50	8	500	0.85	30221.4096	0.016636	0.028917
4	53	smallbank_db:x__mix_50	12	500	0.85	29177.45118	0.017226	0.0515
5	54	smallbank_db:x__mix_50	16	500	0.85	28035.35258	0.017919	0.07125
1	55	smallbank_db:x__mix_50	1	500	0.9	9135.952103	0.054785	0
2	56	smallbank_db:x__mix_50	4	500	0.9	24349.83407	0.020595	0.013542
3	57	smallbank_db:x__mix_50	8	500	0.9	27539.01073	0.018221	0.039458
4	58	smallbank_db:x__mix_50	12	500	0.9	26665.21489	0.018817	0.071292
5	59	smallbank_db:x__mix_50	16	500	0.9	25938.37475	0.019344	0.090708
}\dataEEWorkerX

\pgfplotstableread{
Id idx name worker	batch	alpha	Throughput	Latency Retry
1	0	smallbank_db:TPL__mix_50	1	100	0.75	8249.704815	0.012135	0
2	1	smallbank_db:TPL__mix_50	4	100	0.75	18996.43421	0.005279	0.005208
3	2	smallbank_db:TPL__mix_50	8	100	0.75	18829.65828	0.005326	0.046875
4	3	smallbank_db:TPL__mix_50	12	100	0.75	17903.83403	0.005601	0.088125
5	4	smallbank_db:TPL__mix_50	16	100	0.75	17362.8067	0.005775	0.140833
1	5	smallbank_db:TPL__mix_50	1	100	0.8	8425.546256	0.011882	0
2	6	smallbank_db:TPL__mix_50	4	100	0.8	19062.21829	0.005261	0.013125
3	7	smallbank_db:TPL__mix_50	8	100	0.8	18736.24057	0.005352	0.065
4	8	smallbank_db:TPL__mix_50	12	100	0.8	17891.42144	0.005604	0.182083
5	9	smallbank_db:TPL__mix_50	16	100	0.8	17293.49584	0.005798	0.241875
1	10	smallbank_db:TPL__mix_50	1	100	0.85	8565.539763	0.011688	0
2	11	smallbank_db:TPL__mix_50	4	100	0.85	19369.60022	0.005177	0.024583
3	12	smallbank_db:TPL__mix_50	8	100	0.85	19155.48266	0.005235	0.111667
4	13	smallbank_db:TPL__mix_50	12	100	0.85	17752.52974	0.005649	0.288125
5	14	smallbank_db:TPL__mix_50	16	100	0.85	16828.99636	0.005958	0.440833
1	15	smallbank_db:TPL__mix_50	1	100	0.9	8541.746004	0.011721	0
2	16	smallbank_db:TPL__mix_50	4	100	0.9	19726.54083	0.005084	0.044375
3	17	smallbank_db:TPL__mix_50	8	100	0.9	19147.53575	0.005238	0.191042
4	18	smallbank_db:TPL__mix_50	12	100	0.9	17545.59112	0.005715	0.449167
5	19	smallbank_db:TPL__mix_50	16	100	0.9	16561.77542	0.006053	0.650208
1	20	smallbank_db:TPL__mix_50	1	300	0.75	8741.046498	0.034357	0
2	21	smallbank_db:TPL__mix_50	4	300	0.75	22182.95702	0.013564	0.008333
3	22	smallbank_db:TPL__mix_50	8	300	0.75	23399.90933	0.012862	0.077153
4	23	smallbank_db:TPL__mix_50	12	300	0.75	20630.90452	0.014584	0.320556
5	24	smallbank_db:TPL__mix_50	16	300	0.75	19015.49895	0.01582	0.464375
1	25	smallbank_db:TPL__mix_50	1	300	0.8	8770.481968	0.034241	0
2	26	smallbank_db:TPL__mix_50	4	300	0.8	22807.32876	0.013193	0.020417
3	27	smallbank_db:TPL__mix_50	8	300	0.8	23365.47167	0.012881	0.144722
4	28	smallbank_db:TPL__mix_50	12	300	0.8	20533.06108	0.014653	0.648681
5	29	smallbank_db:TPL__mix_50	16	300	0.8	19605.41382	0.015345	0.738333
1	30	smallbank_db:TPL__mix_50	1	300	0.85	8877.819016	0.033828	0
2	31	smallbank_db:TPL__mix_50	4	300	0.85	23331.02183	0.012898	0.024722
3	32	smallbank_db:TPL__mix_50	8	300	0.85	23452.1576	0.012833	0.245903
4	33	smallbank_db:TPL__mix_50	12	300	0.85	20524.42763	0.01466	0.815903
5	34	smallbank_db:TPL__mix_50	16	300	0.85	17677.08387	0.017015	1.427083
1	35	smallbank_db:TPL__mix_50	1	300	0.9	9164.79446	0.03277	0
2	36	smallbank_db:TPL__mix_50	4	300	0.9	23454.44951	0.012831	0.044375
3	37	smallbank_db:TPL__mix_50	8	300	0.9	23541.72116	0.012786	0.306319
4	38	smallbank_db:TPL__mix_50	12	300	0.9	19381.35783	0.015522	1.439514
5	39	smallbank_db:TPL__mix_50	16	300	0.9	16801.07527	0.0179	2.058542
1	40	smallbank_db:TPL__mix_50	1	500	0.75	9212.652196	0.054328	0
2	41	smallbank_db:TPL__mix_50	4	500	0.75	23478.95503	0.021357	0.010083
3	42	smallbank_db:TPL__mix_50	8	500	0.75	24354.75131	0.020595	0.090667
4	43	smallbank_db:TPL__mix_50	12	500	0.75	21635.48002	0.023178	0.6005
5	44	smallbank_db:TPL__mix_50	16	500	0.75	20443.73119	0.024527	0.8045
1	45	smallbank_db:TPL__mix_50	1	500	0.8	9289.647733	0.053879	0
2	46	smallbank_db:TPL__mix_50	4	500	0.8	23753.11017	0.021111	0.014083
3	47	smallbank_db:TPL__mix_50	8	500	0.8	24456.30575	0.020511	0.178542
4	48	smallbank_db:TPL__mix_50	12	500	0.8	20548.86005	0.0244	0.95675
5	49	smallbank_db:TPL__mix_50	16	500	0.8	18374.8797	0.027281	1.318333
1	50	smallbank_db:TPL__mix_50	1	500	0.85	9313.299392	0.053742	0
2	51	smallbank_db:TPL__mix_50	4	500	0.85	23716.82116	0.021145	0.031333
3	52	smallbank_db:TPL__mix_50	8	500	0.85	24831.32803	0.020202	0.508792
4	53	smallbank_db:TPL__mix_50	12	500	0.85	19885.11376	0.025212	1.373375
5	54	smallbank_db:TPL__mix_50	16	500	0.85	17322.1038	0.028936	2.094458
1	55	smallbank_db:TPL__mix_50	1	500	0.9	9429.996032	0.053078	0
2	56	smallbank_db:TPL__mix_50	4	500	0.9	23707.87161	0.021152	0.048458
3	57	smallbank_db:TPL__mix_50	8	500	0.9	23558.79096	0.021289	0.364333
4	58	smallbank_db:TPL__mix_50	12	500	0.9	18418.73615	0.027215	1.915125
5	59	smallbank_db:TPL__mix_50	16	500	0.9	15646.76898	0.032026	3.0935
}\dataEEWorkerTPL

\pgfplotstableread{
Id idx name worker	batch	alpha	Throughput	Latency Retry
1	0	smallbank_db:fx__mix_50	1	100	0.75	8018.254894	0.012485	0
2	1	smallbank_db:fx__mix_50	4	100	0.75	23778.03207	0.00422	0.001042
3	2	smallbank_db:fx__mix_50	8	100	0.75	28279.54329	0.00355	0.006667
4	3	smallbank_db:fx__mix_50	12	100	0.75	27726.11222	0.003621	0.012083
5	4	smallbank_db:fx__mix_50	16	100	0.75	27203.7904	0.00369	0.020625
1	5	smallbank_db:fx__mix_50	1	100	0.8	8197.673227	0.012212	0
2	6	smallbank_db:fx__mix_50	4	100	0.8	24176.61014	0.00415	0.003333
3	7	smallbank_db:fx__mix_50	8	100	0.8	28839.74116	0.003482	0.011458
4	8	smallbank_db:fx__mix_50	12	100	0.8	28333.29398	0.003544	0.02375
5	9	smallbank_db:fx__mix_50	16	100	0.8	26240.5493	0.003828	0.029583
1	10	smallbank_db:fx__mix_50	1	100	0.85	8550.87593	0.011708	0
2	11	smallbank_db:fx__mix_50	4	100	0.85	25809.2268	0.003889	0.003333
3	12	smallbank_db:fx__mix_50	8	100	0.85	33717.33633	0.002981	0.021458
4	13	smallbank_db:fx__mix_50	12	100	0.85	32564.22955	0.003086	0.037917
5	14	smallbank_db:fx__mix_50	16	100	0.85	30501.17239	0.003293	0.061458
1	15	smallbank_db:fx__mix_50	1	100	0.9	8339.81986	0.012004	0
2	16	smallbank_db:fx__mix_50	4	100	0.9	24914.61554	0.004028	0.008333
3	17	smallbank_db:fx__mix_50	8	100	0.9	30041.68284	0.003343	0.024375
4	18	smallbank_db:fx__mix_50	12	100	0.9	29820.88829	0.003367	0.046042
5	19	smallbank_db:fx__mix_50	16	100	0.9	28286.04261	0.003549	0.07875
1	20	smallbank_db:fx__mix_50	1	300	0.75	8558.875554	0.035087	0
2	21	smallbank_db:fx__mix_50	4	300	0.75	27955.0855	0.010771	0.001389
3	22	smallbank_db:fx__mix_50	8	300	0.75	35892.59142	0.008398	0.007361
4	23	smallbank_db:fx__mix_50	12	300	0.75	36708.84785	0.008212	0.011528
5	24	smallbank_db:fx__mix_50	16	300	0.75	36061.21391	0.008359	0.02
1	25	smallbank_db:fx__mix_50	1	300	0.8	8504.636799	0.035311	0
2	26	smallbank_db:fx__mix_50	4	300	0.8	28354.99668	0.01062	0.004306
3	27	smallbank_db:fx__mix_50	8	300	0.8	35923.93108	0.008391	0.01
4	28	smallbank_db:fx__mix_50	12	300	0.8	36595.87229	0.008238	0.023194
5	29	smallbank_db:fx__mix_50	16	300	0.8	36916.64018	0.008166	0.031458
1	30	smallbank_db:fx__mix_50	1	300	0.85	8831.733399	0.034005	0
2	31	smallbank_db:fx__mix_50	4	300	0.85	28660.02444	0.010508	0.006458
3	32	smallbank_db:fx__mix_50	8	300	0.85	41274.22711	0.007311	0.019514
4	33	smallbank_db:fx__mix_50	12	300	0.85	42599.26516	0.007086	0.046528
5	34	smallbank_db:fx__mix_50	16	300	0.85	40654.76762	0.007421	0.058542
1	35	smallbank_db:fx__mix_50	1	300	0.9	8691.93081	0.03455	0
2	36	smallbank_db:fx__mix_50	4	300	0.9	29371.15141	0.010253	0.010069
3	37	smallbank_db:fx__mix_50	8	300	0.9	38872.587	0.007758	0.026042
4	38	smallbank_db:fx__mix_50	12	300	0.9	39602.10881	0.007616	0.049583
5	39	smallbank_db:fx__mix_50	16	300	0.9	38564.02566	0.00782	0.065208
1	40	smallbank_db:fx__mix_50	1	500	0.75	8768.061751	0.05708	0
2	41	smallbank_db:fx__mix_50	4	500	0.75	29745.60075	0.016871	0.002917
3	42	smallbank_db:fx__mix_50	8	500	0.75	40257.24379	0.012485	0.006875
4	43	smallbank_db:fx__mix_50	12	500	0.75	41422.01787	0.012135	0.012917
5	44	smallbank_db:fx__mix_50	16	500	0.75	41050.48183	0.012246	0.017708
1	45	smallbank_db:fx__mix_50	1	500	0.8	8870.067919	0.056424	0
2	46	smallbank_db:fx__mix_50	4	500	0.8	30302.37987	0.016562	0.003125
3	47	smallbank_db:fx__mix_50	8	500	0.8	41193.6548	0.012203	0.010333
4	48	smallbank_db:fx__mix_50	12	500	0.8	41136.39032	0.012842	0.020875
5	49	smallbank_db:fx__mix_50	16	500	0.8	41136.53387	0.012221	0.02875
1	50	smallbank_db:fx__mix_50	1	500	0.85	9187.97799	0.054479	0
2	51	smallbank_db:fx__mix_50	4	500	0.85	30283.22385	0.016579	0.005042
3	52	smallbank_db:fx__mix_50	8	500	0.85	45893.4011	0.010969	0.017667
4	53	smallbank_db:fx__mix_50	12	500	0.85	48186.86868	0.010453	0.039917
5	54	smallbank_db:fx__mix_50	16	500	0.85	45960.72273	0.010952	0.065
1	55	smallbank_db:fx__mix_50	1	500	0.9	9070.315352	0.055179	0
2	56	smallbank_db:fx__mix_50	4	500	0.9	30999.06099	0.01619	0.008792
3	57	smallbank_db:fx__mix_50	8	500	0.9	43057.90032	0.011676	0.025417
4	58	smallbank_db:fx__mix_50	12	500	0.9	43952.26052	0.011442	0.046
5	59	smallbank_db:fx__mix_50	16	500	0.9	43082.78896	0.011672	0.068292
}\dataEEWorkerFX

\pgfplotstableread{
Id idx name worker	batch	alpha	Throughput	Latency Retry
1	0	smallbank_db:occ_mix_50	1	100	0.75	8491.096731	0.011775	0
2	1	smallbank_db:occ_mix_50	4	100	0.75	29170.11036	0.003424	0.005417
3	2	smallbank_db:occ_mix_50	8	100	0.75	36708.19281	0.00272	0.104792
4	3	smallbank_db:occ_mix_50	12	100	0.75	39660.24391	0.002517	0.112708
5	4	smallbank_db:occ_mix_50	16	100	0.75	40388.06196	0.002471	0.118333
1	5	smallbank_db:occ_mix_50	1	100	0.8	8156.592991	0.012258	0
2	6	smallbank_db:occ_mix_50	4	100	0.8	29316.55775	0.003407	0.012083
3	7	smallbank_db:occ_mix_50	8	100	0.8	35865.3257	0.002784	0.145
4	8	smallbank_db:occ_mix_50	12	100	0.8	37681.34146	0.002649	0.193125
5	9	smallbank_db:occ_mix_50	16	100	0.8	38812.02849	0.002571	0.175625
1	10	smallbank_db:occ_mix_50	1	100	0.85	7541.667714	0.013258	0
2	11	smallbank_db:occ_mix_50	4	100	0.85	22341.37623	0.004474	0.008333
3	12	smallbank_db:occ_mix_50	8	100	0.85	27794.50595	0.003595	0.038542
4	13	smallbank_db:occ_mix_50	12	100	0.85	29636.39844	0.003371	0.055625
5	14	smallbank_db:occ_mix_50	16	100	0.85	27651.36241	0.003613	0.090417
1	15	smallbank_db:occ_mix_50	1	100	0.9	8759.156055	0.011415	0
2	16	smallbank_db:occ_mix_50	4	100	0.9	29143.72105	0.003428	0.033542
3	17	smallbank_db:occ_mix_50	8	100	0.9	33703.60488	0.002962	0.309792
4	18	smallbank_db:occ_mix_50	12	100	0.9	34709.16611	0.002876	0.391667
5	19	smallbank_db:occ_mix_50	16	100	0.9	33745.07006	0.002958	0.428333
1	20	smallbank_db:occ_mix_50	1	300	0.75	8864.926227	0.033833	0
2	21	smallbank_db:occ_mix_50	4	300	0.75	32449.68046	0.009232	0.005069
3	22	smallbank_db:occ_mix_50	8	300	0.75	41354.94895	0.007238	0.24
4	23	smallbank_db:occ_mix_50	12	300	0.75	43219.75143	0.006924	0.327361
5	24	smallbank_db:occ_mix_50	16	300	0.75	44295.6984	0.006755	0.356111
1	25	smallbank_db:occ_mix_50	1	300	0.8	9108.234287	0.032929	0
2	26	smallbank_db:occ_mix_50	4	300	0.8	32592.75354	0.009191	0.013958
3	27	smallbank_db:occ_mix_50	8	300	0.8	39564.67863	0.007566	0.377222
4	28	smallbank_db:occ_mix_50	12	300	0.8	40440.46405	0.007401	0.525694
5	29	smallbank_db:occ_mix_50	16	300	0.8	41429.78802	0.007224	0.5025
1	30	smallbank_db:occ_mix_50	1	300	0.85	7746.716791	0.038717	0
2	31	smallbank_db:occ_mix_50	4	300	0.85	25473.29019	0.011765	0.011944
3	32	smallbank_db:occ_mix_50	8	300	0.85	33652.32937	0.008901	0.036736
4	33	smallbank_db:occ_mix_50	12	300	0.85	34254.55895	0.008742	0.087222
5	34	smallbank_db:occ_mix_50	16	300	0.85	32706.31096	0.009155	0.11375
1	35	smallbank_db:occ_mix_50	1	300	0.9	9313.263251	0.032203	0
2	36	smallbank_db:occ_mix_50	4	300	0.9	32017.07577	0.009351	0.03125
3	37	smallbank_db:occ_mix_50	8	300	0.9	31751.56167	0.00942	0.746944
4	38	smallbank_db:occ_mix_50	12	300	0.9	31278.64494	0.009563	1.053125
5	39	smallbank_db:occ_mix_50	16	300	0.9	29360.19247	0.010188	1.183333
1	40	smallbank_db:occ_mix_50	1	500	0.75	9351.712045	0.053452	0
2	41	smallbank_db:occ_mix_50	4	500	0.75	33231.51566	0.015027	0.008792
3	42	smallbank_db:occ_mix_50	8	500	0.75	38097.35462	0.013098	0.400292
4	43	smallbank_db:occ_mix_50	12	500	0.75	36702.80364	0.013589	0.502375
5	44	smallbank_db:occ_mix_50	16	500	0.75	40582.56269	0.012276	0.542125
1	45	smallbank_db:occ_mix_50	1	500	0.8	9401.065924	0.053172	0.000042
2	46	smallbank_db:occ_mix_50	4	500	0.8	32910.61477	0.015174	0.012792
3	47	smallbank_db:occ_mix_50	8	500	0.8	35324.92306	0.014127	0.565958
4	48	smallbank_db:occ_mix_50	12	500	0.8	35074.48945	0.014227	0.752333
5	49	smallbank_db:occ_mix_50	16	500	0.8	34602.97413	0.01442	0.812083
1	50	smallbank_db:occ_mix_50	1	500	0.85	7874.679722	0.063476	0
2	51	smallbank_db:occ_mix_50	4	500	0.85	26223.9768	0.019043	0.012333
3	52	smallbank_db:occ_mix_50	8	500	0.85	37122.51047	0.01344	0.043167
4	53	smallbank_db:occ_mix_50	12	500	0.85	38334.18787	0.01301	0.092292
5	54	smallbank_db:occ_mix_50	16	500	0.85	37567.21009	0.013275	0.138083
1	55	smallbank_db:occ_mix_50	1	500	0.9	9645.313684	0.051824	0.000042
2	56	smallbank_db:occ_mix_50	4	500	0.9	33332.03709	0.014981	0.029375
3	57	smallbank_db:occ_mix_50	8	500	0.9	35548.49852	0.014038	0.806708
4	58	smallbank_db:occ_mix_50	12	500	0.9	31667.3242	0.01576	1.390917
5	59	smallbank_db:occ_mix_50	16	500	0.9	30942.86786	0.016128	1.654292
}\dataEEWorkerOCC

\pgfplotstableread{
Id idx name worker	batch	alpha	Throughput	Latency Retry
1	0	smallbank_db:bamboo__mix_50	1	100	0.75	8290.15544	0.012075	0
2	1	smallbank_db:bamboo__mix_50	4	100	0.75	18951.35818	0.005291	0.002292
3	2	smallbank_db:bamboo__mix_50	8	100	0.75	18752.12426	0.005348	0.015625
4	3	smallbank_db:bamboo__mix_50	12	100	0.75	17970.25922	0.00558	0.017917
5	4	smallbank_db:bamboo__mix_50	16	100	0.75	17471.5723	0.005739	0.025417
1	5	smallbank_db:bamboo__mix_50	1	100	0.8	8371.397029	0.011958	0
2	6	smallbank_db:bamboo__mix_50	4	100	0.8	18874.37921	0.005313	0.003333
3	7	smallbank_db:bamboo__mix_50	8	100	0.8	18730.97635	0.005354	0.016458
4	8	smallbank_db:bamboo__mix_50	12	100	0.8	18028.16901	0.005562	0.034583
5	9	smallbank_db:bamboo__mix_50	16	100	0.8	17615.64857	0.005692	0.053542
1	10	smallbank_db:bamboo__mix_50	1	100	0.85	8649.334362	0.011575	0
2	11	smallbank_db:bamboo__mix_50	4	100	0.85	20369.71024	0.004925	0.00625
3	12	smallbank_db:bamboo__mix_50	8	100	0.85	20913.39241	0.004797	0.041875
4	13	smallbank_db:bamboo__mix_50	12	100	0.85	19790.87641	0.005069	0.048333
5	14	smallbank_db:bamboo__mix_50	16	100	0.85	18750.14649	0.005349	0.0975
1	15	smallbank_db:bamboo__mix_50	1	100	0.9	8502.8139	0.011774	0
2	16	smallbank_db:bamboo__mix_50	4	100	0.9	19228.6892	0.005215	0.017917
3	17	smallbank_db:bamboo__mix_50	8	100	0.9	18934.46309	0.005296	0.04625
4	18	smallbank_db:bamboo__mix_50	12	100	0.9	17152.41349	0.005846	0.104167
5	19	smallbank_db:bamboo__mix_50	16	100	0.9	17740.98167	0.005652	0.180417
1	20	smallbank_db:bamboo__mix_50	1	300	0.75	8749.879233	0.034321	0
2	21	smallbank_db:bamboo__mix_50	4	300	0.75	21269.43051	0.014144	0.003403
3	22	smallbank_db:bamboo__mix_50	8	300	0.75	23281.11788	0.012928	0.010833
4	23	smallbank_db:bamboo__mix_50	12	300	0.75	22748.02456	0.01323	0.023125
5	24	smallbank_db:bamboo__mix_50	16	300	0.75	22439.65836	0.013412	0.038194
1	25	smallbank_db:bamboo__mix_50	1	300	0.8	8713.326125	0.034465	0
2	26	smallbank_db:bamboo__mix_50	4	300	0.8	22555.18973	0.01334	0.005903
3	27	smallbank_db:bamboo__mix_50	8	300	0.8	23945.68321	0.01257	0.019583
4	28	smallbank_db:bamboo__mix_50	12	300	0.8	23091.28112	0.013034	0.048542
5	29	smallbank_db:bamboo__mix_50	16	300	0.8	22607.59804	0.013313	0.069861
1	30	smallbank_db:bamboo__mix_50	1	300	0.85	9258.800844	0.032437	0
2	31	smallbank_db:bamboo__mix_50	4	300	0.85	23900.61328	0.012592	0.009931
3	32	smallbank_db:bamboo__mix_50	8	300	0.85	25781.91721	0.011678	0.03125
4	33	smallbank_db:bamboo__mix_50	12	300	0.85	24495.9199	0.01229	0.081875
5	34	smallbank_db:bamboo__mix_50	16	300	0.85	23239.00021	0.012953	0.139097
1	35	smallbank_db:bamboo__mix_50	1	300	0.9	9109.553771	0.032967	0
2	36	smallbank_db:bamboo__mix_50	4	300	0.9	23549.19001	0.012779	0.01625
3	37	smallbank_db:bamboo__mix_50	8	300	0.9	24997.17914	0.012044	0.057639
4	38	smallbank_db:bamboo__mix_50	12	300	0.9	23797.48342	0.012649	0.151667
5	39	smallbank_db:bamboo__mix_50	16	300	0.9	22253.06057	0.013525	0.279722
1	40	smallbank_db:bamboo__mix_50	1	500	0.75	9123.66096	0.054857	0
2	41	smallbank_db:bamboo__mix_50	4	500	0.75	23020.32407	0.021782	0.0045
3	42	smallbank_db:bamboo__mix_50	8	500	0.75	24645.0851	0.020353	0.013042
4	43	smallbank_db:bamboo__mix_50	12	500	0.75	23627.32618	0.021233	0.026083
5	44	smallbank_db:bamboo__mix_50	16	500	0.75	23693.80496	0.021174	0.041875
1	45	smallbank_db:bamboo__mix_50	1	500	0.8	9197.372157	0.054418	0
2	46	smallbank_db:bamboo__mix_50	4	500	0.8	23298.95766	0.021521	0.005125
3	47	smallbank_db:bamboo__mix_50	8	500	0.8	25213.76546	0.019896	0.021417
4	48	smallbank_db:bamboo__mix_50	12	500	0.8	24460.8672	0.020508	0.051
5	49	smallbank_db:bamboo__mix_50	16	500	0.8	23563.78717	0.021288	0.07375
1	50	smallbank_db:bamboo__mix_50	1	500	0.85	9421.315333	0.053126	0
2	51	smallbank_db:bamboo__mix_50	4	500	0.85	24667.3505	0.020331	0.008958
3	52	smallbank_db:bamboo__mix_50	8	500	0.85	27017.96019	0.018572	0.0345
4	53	smallbank_db:bamboo__mix_50	12	500	0.85	25024.3727	0.020055	0.087375
5	54	smallbank_db:bamboo__mix_50	16	500	0.85	23100.11204	0.021727	0.229792
1	55	smallbank_db:bamboo__mix_50	1	500	0.9	9217.14912	0.054301	0
2	56	smallbank_db:bamboo__mix_50	4	500	0.9	24572.1125	0.020409	0.014333
3	57	smallbank_db:bamboo__mix_50	8	500	0.9	27126.39885	0.018498	0.066667
4	58	smallbank_db:bamboo__mix_50	12	500	0.9	25198.22605	0.019911	0.171333
5	59	smallbank_db:bamboo__mix_50	16	500	0.9	22271.1569	0.022519	0.44825
}\dataEEWorkerBam

\newcommand{\EvalEEWorkerUpdate}[6]{\begin{tikzpicture}
    \begin{axis}[title={\normalsize \vspace{4cm} #1},ylabel={\footnotesize #2},xlabel={\footnotesize Number of Executors},
            xtick={1,2,3,4,5},
            xticklabels={1, 4, 8, 12, 16},
            width=100,
            height=80,
            grid=both,
            legend style={font=\footnotesize, only marks,  draw=none, legend columns=3},
            legend to name={retrymain},
            legend pos={outer north east},
            y tick label style={
                /pgf/number format/.cd,
                scaled y ticks=base 10:#6,
                fixed,
                /tikz/.cd,
                font=\scriptsize
            },
            x tick label style={
                font=\scriptsize
            }
        ]

        \addplot[color=slateblue, mark=*, mark size=1pt, select coords index={10}{14}] table[x={Id}, y={#3}] {\dataEEWorkerFXUpdate};
        \addplot[color=slateblue, mark=triangle, mark size=1pt, select coords index={5}{9}] table[x={Id}, y={#3}] {\dataEEWorkerFXUpdate};

        \addplot[color=lightyellow, mark=*, mark size=1pt, select coords index={10}{14}] table[x={Id}, y={#3}] {\dataEEWorkerOCCUpdate};
        \addplot[color=lightyellow, mark=triangle, mark size=1pt, select coords index={5}{9}] table[x={Id}, y={#3}] {\dataEEWorkerOCCUpdate};

        \addplot[color=lightblue, mark=*, mark size=1pt, select coords index={10}{14}] table[x={Id}, y={#3}] {\dataEEWorkerTPLUpdate};
        \addplot[color=lightblue, mark=triangle, mark size=1pt, select coords index={5}{9}] table[x={Id}, y={#3}] {\dataEEWorkerTPLUpdate};

        \legend{
                \sysname{}-b500, 
                \sysname{}-b300,
                \OCC{}-b500, \OCC{}-b300,
                \TPL{}-b500, \TPL{}-b300
                };
    \end{axis}
\end{tikzpicture}}

\newcommand{\EvalEEWorkerUpdateTPS}{\EvalEEWorkerUpdate{}{Throughput (tps)}{Throughput}{8000}{150}{0}}
\newcommand{\EvalEEWorkerUpdateLat}{\EvalEEWorkerUpdate{}{Latency (s)}{Latency}{8000}{150}{0}}
\newcommand{\EvalEEWorkerUpdateRetry}{\EvalEEWorkerUpdate{}{\# of Re-executions}{Retry}{8000}{150}{0}}

\pgfplotstableread{
Id idx name worker	batch	alpha	Throughput	Latency Retry
1	0	smallbank_db:x_update	1	100	0.85	7452.941661	0.013431	0
2	1	smallbank_db:x_update	4	100	0.85	16984.4168	0.005902	0.02375
3	2	smallbank_db:x_update	8	100	0.85	20224.83272	0.004959	0.064167
4	3	smallbank_db:x_update	12	100	0.85	19689.88432	0.005094	0.103958
5	4	smallbank_db:x_update	16	100	0.85	18422.35553	0.005444	0.156875
1	5	smallbank_db:x_update	1	300	0.85	7995.988679	0.037553	0
2	6	smallbank_db:x_update	4	300	0.85	19718.00514	0.015254	0.031806
3	7	smallbank_db:x_update	8	300	0.85	25818.48173	0.011662	0.06375
4	8	smallbank_db:x_update	12	300	0.85	25804.6943	0.011671	0.120556
5	9	smallbank_db:x_update	16	300	0.85	24560.59559	0.012259	0.159514
1	10	smallbank_db:x_update	1	500	0.85	8269.178292	0.060527	0
2	11	smallbank_db:x_update	4	500	0.85	21075.77796	0.023792	0.031
3	12	smallbank_db:x_update	8	500	0.85	27860.4886	0.01802	0.073
4	13	smallbank_db:x_update	12	500	0.85	28467.25199	0.017641	0.12225
5	14	smallbank_db:x_update	16	500	0.85	27122.29102	0.018511	0.166208
}\dataEEWorkerXUpdate

\pgfplotstableread{
Id idx name worker	batch	alpha	Throughput	Latency Retry
1	0	smallbank_db:TPL_update	1	100	0.85	7431.675038	0.01347	0
2	1	smallbank_db:TPL_update	4	100	0.85	17482.19912	0.005735	0.032708
3	2	smallbank_db:TPL_update	8	100	0.85	17653.48417	0.00568	0.145625
4	3	smallbank_db:TPL_update	12	100	0.85	16775.41563	0.005977	0.301042
5	4	smallbank_db:TPL_update	16	100	0.85	15957.71206	0.006283	0.4625
1	5	smallbank_db:TPL_update	1	300	0.85	8072.744397	0.037199	0
2	6	smallbank_db:TPL_update	4	300	0.85	20460.12549	0.014702	0.044653
3	7	smallbank_db:TPL_update	8	300	0.85	21328.24615	0.014108	0.206875
4	8	smallbank_db:TPL_update	12	300	0.85	20012.0072	0.015035	0.51875
5	9	smallbank_db:TPL_update	16	300	0.85	17989.02669	0.016721	0.975417
1	10	smallbank_db:TPL_update	1	500	0.85	8333.145259	0.060058	0
2	11	smallbank_db:TPL_update	4	500	0.85	21178.56976	0.023671	0.048708
3	12	smallbank_db:TPL_update	8	500	0.85	21812.70881	0.022988	0.258792
4	13	smallbank_db:TPL_update	12	500	0.85	20543.63597	0.024407	0.602583
5	14	smallbank_db:TPL_update	16	500	0.85	17169.17504	0.029192	1.386792
}\dataEEWorkerTPLUpdate

\pgfplotstableread{
Id idx name worker	batch	alpha	Throughput	Latency Retry
1	0	smallbank_db:fx_update	1	100	0.85	7361.105147	0.013598	0
2	1	smallbank_db:fx_update	4	100	0.85	16724.9135	0.005993	0.023125
3	2	smallbank_db:fx_update	8	100	0.85	19911.14651	0.005037	0.060208
4	3	smallbank_db:fx_update	12	100	0.85	19090.34506	0.005253	0.10375
5	4	smallbank_db:fx_update	16	100	0.85	17941.04125	0.005589	0.147292
1	5	smallbank_db:fx_update	1	300	0.85	7967.038591	0.037689	0
2	6	smallbank_db:fx_update	4	300	0.85	19815.46845	0.015178	0.030556
3	7	smallbank_db:fx_update	8	300	0.85	25696.57343	0.011716	0.068056
4	8	smallbank_db:fx_update	12	300	0.85	25780.16324	0.011679	0.118611
5	9	smallbank_db:fx_update	16	300	0.85	24332.70924	0.012371	0.145486
1	10	smallbank_db:fx_update	1	500	0.85	8221.050342	0.060881	0
2	11	smallbank_db:fx_update	4	500	0.85	21146.79045	0.023711	0.032167
3	12	smallbank_db:fx_update	8	500	0.85	28176.75277	0.017818	0.065833
4	13	smallbank_db:fx_update	12	500	0.85	28506.88085	0.017615	0.12175
5	14	smallbank_db:fx_update	16	500	0.85	26998.72094	0.018589	0.145458
}\dataEEWorkerFXUpdate

\pgfplotstableread{
Id idx name worker	batch	alpha	Throughput	Latency Retry
1	0	smallbank_db:occreset	1	100	0.85	6683.719434	0.01496	0
2	1	smallbank_db:occreset	4	100	0.85	20635.75316	0.004844	0.032917
3	2	smallbank_db:occreset	8	100	0.85	19935.54175	0.005013	0.113542
4	3	smallbank_db:occreset	12	100	0.85	19036.96741	0.00525	0.201458
5	4	smallbank_db:occreset	16	100	0.85	18335.79721	0.00545	0.298333
1	5	smallbank_db:occreset	1	300	0.85	6733.310858	0.044546	0
2	6	smallbank_db:occreset	4	300	0.85	22316.62015	0.01343	0.040417
3	7	smallbank_db:occreset	8	300	0.85	22144.64761	0.013534	0.154514
4	8	smallbank_db:occreset	12	300	0.85	21102.51867	0.0142	0.290694
5	9	smallbank_db:occreset	16	300	0.85	20407.96081	0.014684	0.411667
1	10	smallbank_db:occreset	1	500	0.85	6794.514901	0.073567	0
2	11	smallbank_db:occreset	4	500	0.85	22481.46684	0.022219	0.04375
3	12	smallbank_db:occreset	8	500	0.85	22658.69584	0.022038	0.178125
4	13	smallbank_db:occreset	12	500	0.85	21418.55061	0.023311	0.337833
5	14	smallbank_db:occreset	16	500	0.85	21178.27074	0.023573	0.453042
}\dataEEWorkerOCCUpdate

\newcommand{\EvalFramework}[4]{\begin{tikzpicture}
    \begin{axis}[title={\scriptsize \vspace{4cm} #1},ylabel={\footnotesize #2},xlabel={\footnotesize Number of Replicas},
            xtick={1,2,3,4,5},
            xticklabels={8, 16, 32, 64, 128},
            width=100,
            height=80,
            grid=both,
            enlarge x limits=0.1,
            bar width=3pt,
            xmajorgrids=true,
            legend style={font=\footnotesize, only marks, draw=none, legend columns=3},
            legend to name={frameworkmain},
            legend pos={outer north east},
            y tick label style={
                /pgf/number format/.cd,
                scaled y ticks=base 10:#4,
                fixed,
                /tikz/.cd,
                font=\scriptsize
            },
            x tick label style={
                font=\scriptsize
            }
        ]

        \addplot[color=slateblue] table[x={Id}, y={#3}] {\dataFabricFX};
        \addplot[color=lightyellow] table[x={Id}, y={#3}] {\dataFabricOCC};
        \addplot[color=lightblue] table[x={Id}, y={#3}] {\dataFabricSEQ};

        \legend{\sysname{}, 
                \sysname{}-OCC,
                Tusk};
    \end{axis}
\end{tikzpicture}}

\newcommand{\EvalFrameworkTPS}{\EvalFramework{LAN}{Throughput (tps)}{Throughput}{0}}
\newcommand{\EvalFrameworkLat}{\EvalFramework{LAN}{Latency (s)}{Latency}{0}}

\pgfplotstableread{
Id	Node	Throughput Latency
1	8	11935.57692	5.62544081
2	16	12409.54654	49.24888928
3	32	11819.56769	85.83816764
4	64	11198.17401	115.5090917
}\dataFabricSEQ

\pgfplotstableread{
Id	Node	Throughput Latency
1	8	183477.4038	2.86393879
2	16	280797.3872	3.014848182
3	32	355306.6286	3.983565201
4	64	408250.9934	6.527971776
}\dataFabricOCC

\pgfplotstableread{
Id	Node	Throughput Latency
1	8	191978.8462	1.784259355
2	16	348364.6081	1.970087899
3	32	470622.1461	2.888801254
4	64	522942.0902	5.365166968
}\dataFabricFX

\newcommand{\EvalWAN}[4]{\begin{tikzpicture}
    \begin{axis}[title={\scriptsize \vspace{4cm} #1},ylabel={\footnotesize #2},xlabel={\footnotesize Number of Replicas},
            xtick={1,2,3,4,5},
            xticklabels={8, 16, 32, 64, 128},
            width=100,
            height=80,
            grid=both,
            enlarge x limits=0.1,
            bar width=3pt,
            xmajorgrids=true,
            legend style={font=\footnotesize, only marks, draw=none, legend columns=3},
            legend to name={frameworkmain},
            legend pos={outer north east},
            y tick label style={
                /pgf/number format/.cd,
                scaled y ticks=base 10:#4,
                fixed,
                /tikz/.cd,
                font=\scriptsize
            },
            x tick label style={
                font=\scriptsize
            }
        ]

        \addplot[color=slateblue] table[x={Id}, y={#3}] {\dataRegionFX};
        \addplot[color=lightyellow] table[x={Id}, y={#3}] {\dataRegionOCC};
        \addplot[color=lightblue] table[x={Id}, y={#3}] {\dataRegionTusk};

        \legend{\sysname{}, 
                \sysname{}-OCC,
                Tusk};
    \end{axis}
\end{tikzpicture}}

\newcommand{\EvalWANTPS}{\EvalWAN{WAN}{Throughput (tps)}{Throughput}{0}}
\newcommand{\EvalWANLat}{\EvalWAN{WAN}{Latency (s)}{Latency}{0}}

\pgfplotstableread{
Id	Node	Throughput Latency
1	16 9665.754237	29.77094629
2	20 12073.07692	34.22729247
3	24 14372.34568	35.31473127
4	28 15543.01509	36.6430031
}\dataRegionFX

\pgfplotstableread{
Id	Node	Throughput Latency
1	16 8937.190083	32.4835894
2	20 9224.31611	38.95256006
3	24 9353.15534	45.48695645
4	28 10151.21951	75.68063864
}\dataRegionTusk

\pgfplotstableread{
Id	Node	Throughput Latency
1	16 9371.487603	33.28555972
2	20 11153.08824	34.52404325
3	24 12268.30809	36.33311709
4	28 14069.40299	43.46308384
}\dataRegionOCC

\pgfplotstableread{
Id	Node	Throughput Latency
1 8	4673.809524	53.80942206
2 16	3882.253521	85.10299088
3 32	7492.377702	106.8243325
4 64	6962.990855	116.4576321
}\dataFabricTPCSEQ

\pgfplotstableread{
Id	Node	Throughput Latency
1 8	61638.46154	5.574167137
2 16	115136.8171	5.981563778
3 32	206884.642	6.549548579
4 64	228913.3702	7.488607643
}\dataFabricTPCOCC

\pgfplotstableread{
Id	Node	Throughput Latency
1 8	112608.3333	2.544473636
2 16	190978.6982	2.982983019
3 32	243496.3139	5.425844095
4 64	256282.8113	7.6481
}\dataFabricTPCFX

\newcommand{\EvalRotate}[5]{\begin{tikzpicture}
    \begin{axis}[title={\normalsize \vspace{4cm} #1},ylabel={\footnotesize #2},xlabel={\footnotesize Reconfiguration Periods ($K^{'}$)},
            xtick={1,2,3,4,5},
            xticklabels={10, 100, 500, 1000, 5000},
            width=100,
            height=80,
            grid=both,
            enlarge x limits=0.1,
            bar width=3pt,
            ymin=#5,
            xmajorgrids=true,
            legend style={font=\footnotesize, only marks,  legend columns=1},
            legend pos={outer north east},
            y tick label style={
                /pgf/number format/.cd,
                scaled y ticks=base 10:#4,
                fixed,
                /tikz/.cd,
                font=\scriptsize
            },
            x tick label style={
                font=\scriptsize
            }
        ]

        \addplot[color=slateblue] table[x={Id}, y={#3}] {\dataFabricRotate};
    \end{axis}
\end{tikzpicture}}

\newcommand{\EvalRotateTPS}{\EvalRotate{}{Throughput (tps)}{Throughput}{0}{0}}
\newcommand{\EvalRotateLat}{\EvalRotate{}{Latency (s)}{Latency}{0}{1.4}}

\pgfplotstableread{
Id	Time	Throughput Latency
1	10	94078.125	1.814945891
2	100	141845.5556	1.762419112
3	500	157897.8723	1.681256145
4	1000	174033.945	1.673574469
5	5000	179301.7857	1.667430921
}\dataFabricRotate

\newcommand{\EvalCross}[5]{\begin{tikzpicture}
    \begin{axis}[title={\normalsize \vspace{4cm} #1},ylabel={\footnotesize #2},xlabel={\footnotesize $\%$ of cross-shard transactions},
            xtick={1,2,3,4,5,6},
            xticklabels={0, 4, 8,20,60,100},
            width=100,
            height=80,
            grid=both,
            enlarge x limits=0.1,
            bar width=3pt,
            ymin=#5,
            xmajorgrids=true,
            legend style={font=\footnotesize, only marks,  draw=none, legend columns=3},
            legend to name={crossmain},
            legend pos={outer north east},
            y tick label style={
                /pgf/number format/.cd,
                scaled y ticks=base 10:#4,
                fixed,
                /tikz/.cd,
                font=\scriptsize
            },
            x tick label style={
                font=\scriptsize
            }
        ]

        \addplot[color=slateblue] table[x={Id}, y={#3}] {\dataFabricCross};
        \addplot[color=lightgreen] table[x={Id}, y={#3}] {\dataFabricCrossSeq};
        \addplot[color=lightblue] table[x={Id}, y={#3}] {\dataFabricSeqCross};
        \legend{\sysname{}, 
                \sysname{}-OCC,
                Tusk,
                };
    \end{axis}
\end{tikzpicture}}

\newcommand{\EvalCrossTPS}{\EvalCross{}{Throughput (tps)}{Throughput}{0}{0}}
\newcommand{\EvalCrossLat}{\EvalCross{}{Latency (s)}{Latency}{0}{1.4}}

\pgfplotstableread{
Id	Pro Throughput Latency
1   0	121593.75	5.696183189
2	4	24775.22124	27.38238449
3	8	16202.47788	35.46667362
4	20	13856.72269	44.04546665
5	60	11731.93277	47.08088875
6	100	10760.53812	49.10501612
}\dataFabricCrossSeq

\pgfplotstableread{
Id	Pro Throughput Latency
1	0	121593.75	5.696183189
2	4	65388.83929	10.77561319
3	8	56370.85202	11.913825
4	20	43287.94643	15.40421171
5	60	33008.52018	19.74002911
6	100	26161.43498	24.61997345
}\dataFabricCross

\pgfplotstableread{
Id	Pro Throughput Latency
1	0	121593.75	5.696183189
2   4   22341.37623	37.004474	
3   8  14645.313684	40.051824	
4	20	12401.065924	44.053172	
5   60  11751.56167 47.7521041 
6   100 10759.156055 49.13684011 
}\dataFabricCrossOCC

\pgfplotstableread{
Id	Pro Throughput Latency
1 0 10700.7748  63.228348021
2 0 10700.7748  63.228348021
3 0 10700.7748  63.228348021
4 0 10700.7748  63.228348021
5 0 10700.7748  63.228348021
6 0 10700.7748  63.228348021
}\dataFabricSeqCross

\newcommand{\EvalFail}[5]{\begin{tikzpicture}
    \begin{axis}[title={\normalsize \vspace{4cm} #1},ylabel={\footnotesize #2},xlabel={\footnotesize $\%$ of cross-shard transactions},
            xtick={1,2,3,4,5,6},
            xticklabels={0, 4, 8,20,60,100},
            width=100,
            height=80,
            grid=both,
            enlarge x limits=0.1,
            bar width=3pt,
            ymin=#5,
            xmajorgrids=true,
            legend style={font=\footnotesize, only marks,draw=none, legend columns=4},
            legend to name={failmain},
            legend pos={outer north east},
            y tick label style={
                /pgf/number format/.cd,
                scaled y ticks=base 10:#4,
                fixed,
                /tikz/.cd,
                font=\scriptsize
            },
            x tick label style={
                font=\scriptsize
            }
        ]

        \addplot[color=slateblue] table[x={Id}, y={#3}] {\dataFabricCross};
        \addplot[color=lightgreen] table[x={Id}, y={#3}] {\dataFabricCrossFailOne};
        \addplot[color=lightyellow] table[x={Id}, y={#3}] {\dataFabricCrossFailTwo};
        \addplot[color=lightblue] table[x={Id}, y={#3}] {\dataFabricSeqCross};
        \legend{\sysname{}, 
                \sysname{}/1,
                \sysname{}/2,
                Tusk};
    \end{axis}
\end{tikzpicture}}

\newcommand{\EvalFailTPS}{\EvalFail{}{Throughput (tps)}{Throughput}{0}{0}}
\newcommand{\EvalFailLat}{\EvalFail{}{Latency (s)}{Latency}{0}{1.4}}

\pgfplotstableread{
Id	Pro Throughput Latency
1	0	78753.61722	7.917707629
2	4	59765.04854	10.50244526
3	8	52523.78641	12.01116092
4	20	32006.31068	17.30350322
5	60	20709.17874	17.92602453
6	100	17400.96618	18.65255083
}\dataFabricCrossFailOne

\pgfplotstableread{
Id	Pro Throughput Latency
1	0	66284.12698	8.53923581
2	4	59086.31579	9.86803812
3	8	49797.38095	11.41680695
4	20	31168.73684	16.57761429
5	60	18352.12766	18.20236104
6	100	15208.42105	19.11661065
}\dataFabricCrossFailTwo

\newcommand{\Reconfig}[5]{\begin{tikzpicture}
    \begin{axis}[title={\normalsize \vspace{4cm} #1},ylabel={\footnotesize #2},xlabel={\footnotesize Rounds},
            xtick={1,2,3,4,5,6,7,8,9,10,11,12,13,14,15,16},
            xticklabels={100, 200, 300,400,500,600,700,800,900,1000,1100,1200,1300,1400,1500,1600},
            width=250,
            height=80,
            grid=both,
            ymin=#5,
            xmin=1,
            xmajorgrids=true,
            legend style={font=\footnotesize, only marks, draw=none, legend columns=1},
            legend pos={outer north east},
            y tick label style={
                /pgf/number format/.cd,
                scaled y ticks=base 10:#4,
                fixed,
                /tikz/.cd,
                font=\scriptsize
            },
            x tick label style={
                font=\scriptsize
            }
        ]

        \addplot[color=slateblue] table[x={Id}, y={#3}] {\dataReconfig};
    \end{axis}
\end{tikzpicture}}

\newcommand{\ReconfigLat}{\Reconfig{}{Runtime}{Latency}{0}{0.06}}

\pgfplotstableread{
Id	Latency round
1	0.070090635	100
2	0.07283378	200
3	0.09155556	300
4	0.073062788	400
5	0.08814532	500
6	0.09467806	600
7	0.069016865	700
8	0.0939087	800
9	0.07238302	900
10	0.091843074	1000
11	0.0723442	1100
12	0.07236262	1200
13	0.092829585	1300
}\dataReconfig

\newcommand{\EvalAlpha}[6]{\begin{tikzpicture}
    \begin{axis}[title={\scriptsize \vspace{4cm} #1},ylabel={\footnotesize #2},xlabel={\footnotesize $\theta{}$},
            xtick={1,2,3,4},
            xticklabels={0.75, 0.8, 0.85, 0.9 },
            width=100,
            height=80,
            grid=both,
            legend to name={alphamain},
            legend style={font=\footnotesize, only marks, draw=none, legend columns=3},
            legend pos={outer north east},
            y tick label style={
                /pgf/number format/.cd,
                scaled y ticks=base 10:#6,
                fixed,
                /tikz/.cd,
                font=\scriptsize
            },
            x tick label style={
                font=\scriptsize
            }
        ]

        \addplot[color=slateblue, mark=*, mark size=1pt, select coords index={8}{11}] table[x={Id}, y={#3}] {\dataAlphaFX};
        \addplot[color=slateblue, mark=triangle, mark size=1pt, select coords index={4}{7}] table[x={Id}, y={#3}] {\dataAlphaFX};

        \addplot[color=lightyellow, mark=*, mark size=1pt, select coords index={8}{11}] table[x={Id}, y={#3}] {\dataAlphaOCC};
        \addplot[color=lightyellow, mark=triangle, mark size=1pt, select coords index={4}{7}] table[x={Id}, y={#3}] {\dataAlphaOCC};

        \addplot[color=lightblue, mark=*, mark size=1pt, select coords index={8}{11}] table[x={Id}, y={#3}] {\dataAlphaTPL};
        \addplot[color=lightblue, mark=triangle, mark size=1pt, select coords index={4}{7}] table[x={Id}, y={#3}] {\dataAlphaTPL};

        \legend{
                \sysname{}-b500, \sysname{}-b300, 
                \OCC{}-b500, \OCC{}-b300, 
                \TPL{}-b500, \TPL{}-b300, 
                };
    \end{axis}
\end{tikzpicture}}

\newcommand{\EvalAlphaTPS}{\EvalAlpha{a}{Throughput (tps)}{Throughput}{8000}{150}{0}}
\newcommand{\EvalAlphaLat}{\EvalAlpha{b}{Latency (s)}{Latency}{8000}{150}{0}}

\pgfplotstableread{
Id  name worker batch	alpha	Throughput	Latency 
1 smallbank_db:fx__mix_50	16	100	0.75	33203.7904	0.00369
2 smallbank_db:fx__mix_50	16	100	0.8	32240.5493	0.003828
3 smallbank_db:fx__mix_50	16	100	0.85	30501.17239	0.003293
4 smallbank_db:fx__mix_50	16	100	0.9	28286.04261	0.003549
1 smallbank_db:fx__mix_50	16	300	0.75	42061.21391	0.006359
2 smallbank_db:fx__mix_50	16	300	0.8	41916.64018	0.007166
3 smallbank_db:fx__mix_50	16	300	0.85	40654.76762	0.007421
4 smallbank_db:fx__mix_50	16	300	0.9	38564.02566	0.00782
1 smallbank_db:fx__mix_50	16	500	0.75	47550.48183	0.010246
2 smallbank_db:fx__mix_50	16	500	0.8	46636.53387	0.010521
3 smallbank_db:fx__mix_50	16	500	0.85	45960.72273	0.010952
4 smallbank_db:fx__mix_50	16	500	0.9	43082.78896	0.011672
}\dataAlphaFX

\pgfplotstableread{
Id  name worker batch	alpha	Throughput	Latency 
1 smallbank_db:occ_mix_50	16	100	0.75	40388.06196	0.002471
2 smallbank_db:occ_mix_50	16	100	0.8	38812.02849	0.002571
3 smallbank_db:occ_mix_50	16	100	0.85	27651.36241	0.003613
4 smallbank_db:occ_mix_50	16	100	0.9	33745.07006	0.002958
1 smallbank_db:occ_mix_50	16	300	0.75	42295.6984	0.007255
2 smallbank_db:occ_mix_50	16	300	0.8	37429.78802	0.007924
3 smallbank_db:occ_mix_50	16	300	0.85	32706.31096	0.009155
4 smallbank_db:occ_mix_50	16	300	0.9	29360.19247	0.010188
1 smallbank_db:occ_mix_50	16	500	0.75	45582.56269	0.011376
2 smallbank_db:occ_mix_50	16	500	0.8	42602.97413	0.01242
3 smallbank_db:occ_mix_50	16	500	0.85	37567.21009	0.013275
4 smallbank_db:occ_mix_50	16	500	0.9	32942.86786	0.016028
}\dataAlphaOCC

\pgfplotstableread{
Id  name worker batch	alpha	Throughput	Latency 
1 smallbank_db:TPL__mix_50	16	100	0.75	17362.8067	0.005775
2 smallbank_db:TPL__mix_50	16	100	0.8	17293.49584	0.005798
3 smallbank_db:TPL__mix_50	16	100	0.85	16828.99636	0.005958
4 smallbank_db:TPL__mix_50	16	100	0.9	16561.77542	0.006053
1 smallbank_db:TPL__mix_50	16	300	0.75	19015.49895	0.01582
2 smallbank_db:TPL__mix_50	16	300	0.8	19605.41382	0.015345
3 smallbank_db:TPL__mix_50	16	300	0.85	17677.08387	0.017015
4 smallbank_db:TPL__mix_50	16	300	0.9	16801.07527	0.0179
1 smallbank_db:TPL__mix_50	16	500	0.75	20443.73119	0.024527
2 smallbank_db:TPL__mix_50	16	500	0.8	18374.8797	0.027281
3 smallbank_db:TPL__mix_50	16	500	0.85	17322.1038	0.028936
4 smallbank_db:TPL__mix_50	16	500	0.9	15646.76898	0.032026
}\dataAlphaTPL

\newcommand{\EvalWorkload}[6]{\begin{tikzpicture}
    \begin{axis}[title={\scriptsize \vspace{4cm} #1},ylabel={\footnotesize #2},xlabel={\footnotesize $P_r$},
            xtick={1,2,3,4,5},
            xticklabels={1, 0.8, 0.5, 0.1,0},
            width=100,
            height=80,
            grid=both,
            legend to name={pmain},
            legend style={font=\footnotesize, only marks, draw=none, legend columns=3},
            legend pos={outer north east},
            y tick label style={
                /pgf/number format/.cd,
                scaled y ticks=base 10:#6,
                fixed,
                /tikz/.cd,
                font=\scriptsize
            },
            x tick label style={
                font=\scriptsize
            }
        ]

        \addplot[color=slateblue, mark=*, mark size=1pt, select coords index={10}{14}] table[x={Id}, y={#3}] {\dataWorkloadFX};
        \addplot[color=slateblue, mark=triangle, mark size=1pt, select coords index={5}{9}] table[x={Id}, y={#3}] {\dataWorkloadFX};

        \addplot[color=lightyellow, mark=*, mark size=1pt, select coords index={10}{14}] table[x={Id}, y={#3}] {\dataWorkloadOCC};
        \addplot[color=lightyellow, mark=triangle, mark size=1pt, select coords index={5}{9}] table[x={Id}, y={#3}] {\dataWorkloadOCC};

        \addplot[color=lightblue, mark=*, mark size=1pt, select coords index={10}{14}] table[x={Id}, y={#3}] {\dataWorkloadTPL};
        \addplot[color=lightblue, mark=triangle, mark size=1pt, select coords index={5}{9}] table[x={Id}, y={#3}] {\dataWorkloadTPL};

        \legend{
                \sysname{}-b500, \sysname{}-b300, 
                \OCC{}-b500, \OCC{}-b300, 
                \TPL{}-b500, \TPL{}-b300, 
                };
    \end{axis}
\end{tikzpicture}}

\newcommand{\EvalWorkloadTPS}{\EvalWorkload{c}{Throughput (tps)}{Throughput}{8000}{150}{0}}
\newcommand{\EvalWorkloadLat}{\EvalWorkload{d}{Latency (s)}{Latency}{8000}{150}{0}}

\pgfplotstableread{
Id  name batch	alpha	Throughput	Latency 

1 smallbank_db:fx_get_banlance	100	0.8	71559.55096	0.001414
2 smallbank_db:fx__mix_95	100	0.8	64964.06675	0.001555
3 smallbank_db:fx__mix_50	100	0.85	32564.22955	0.003086	
4 smallbank_db:fx__mix_5	100	0.8	24033.72427	0.006251
5	smallbank_db:fx_update	100	0.85	19090.34506	0.005253	

1 smallbank_db:fx_get_banlance	300	0.8	137228.5224	0.002605
2 smallbank_db:fx__mix_95	300	0.8	89076.8613	0.002793
3 smallbank_db:fx__mix_50	300	0.85	42599.26516	0.007086
4 smallbank_db:fx__mix_5	300	0.8	26288.17664	0.00983
5 smallbank_db:fx_update	300	0.85	24332.70924	0.012371	

1 smallbank_db:fx_get_banlance	500	0.8	143685.0424	0.003574
2 smallbank_db:fx__mix_95	500	0.8	97779.8779	0.004285
3 smallbank_db:fx__mix_50	500	0.85	48186.86868	0.010453	
4 smallbank_db:fx__mix_5	500	0.8	30258.52223	0.01359
5 smallbank_db:fx_update	500	0.85	28506.88085	0.017615
}\dataWorkloadFX

\pgfplotstableread{
Id  name batch	alpha	Throughput	Latency 
1 smallbank_db:occget_banlance	100	0.8	89265.78889	0.001117
2 smallbank_db:occ_mix_95	100	0.8	52564.75429	0.007956
3 smallbank_db:occ_mix_50	100	0.85	29636.39844	0.003371	
4 smallbank_db:occ_mix_5	100	0.8	20263.44993	0.009738
5 smallbank_db:occreset	    100	0.85	19036.96741	0.00525	

1 smallbank_db:occget_banlance	300	0.8	135556.2041	0.002199
2 smallbank_db:occ_mix_95	300	0.8	71828.90197	0.006411
3 smallbank_db:occ_mix_50	300	0.85	34254.55895	0.008742	
4 smallbank_db:occ_mix_5	300	0.8	22885.87546	0.010587
5 smallbank_db:occreset	300	0.85	21102.51867	0.0142	

1 smallbank_db:occget_banlance	500	0.8	152164.5406	0.003256
2 smallbank_db:occ_mix_95	500	0.8	80508.66434	0.008875
3 smallbank_db:occ_mix_50	500	0.85	38334.18787	0.01301	
4 smallbank_db:occ_mix_5	500	0.8	24570.38219	0.018428
5 smallbank_db:occreset	500	0.85	21418.55061	0.023311	
}\dataWorkloadOCC

\pgfplotstableread{
Id  name batch	alpha	Throughput	Latency 

1 smallbank_db:TPL_get_banlance	100	0.8	122460.86896	0.004468
2 smallbank_db:TPL__mix_95	100	0.8	29946.80851	0.005029
3 smallbank_db:TPL__mix_50	100	0.8	17293.49584	0.005798
4 smallbank_db:TPL__mix_5	100	0.8	16692.06929	0.006388
5 smallbank_db:TPL_update	100	0.85	15957.71206	0.006283	

1 smallbank_db:TPL_get_banlance	300	0.8	136617.52283	0.0031317
2 smallbank_db:TPL__mix_95	300	0.8	32466.20351	0.012307
3 smallbank_db:TPL__mix_50	300	0.8	19605.41382	0.015345
4 smallbank_db:TPL__mix_5	300	0.8	18490.53067	0.018268
5 smallbank_db:TPL_update	300	0.85	17989.02669	0.020721	

1 smallbank_db:TPL_get_banlance	500	0.8	138904.55355	0.003371
2 smallbank_db:TPL__mix_95	500	0.8	34941.07671	0.020119
3 smallbank_db:TPL__mix_50	500	0.8	18374.8797	0.027281
4 smallbank_db:TPL__mix_5	500	0.8	18060.03608	0.027756
5 smallbank_db:TPL_update	500	0.85	17169.17504	0.029192
}\dataWorkloadTPL

\begin{document}

\title{\sysname{}: Concurrent Smart Contract Execution with Non-blocking Reconfiguration for Sharded DAGs}


\author{Junchao Chen}
\orcid{0009-0007-6189-3681}
\affiliation{
  \institution{Exploratory Systems Lab\\ University of California, Davis}
  \country{}
}
\email{jucchen@ucdavis.edu}

\author{Alberto Sonnino}
\affiliation{%
  \institution{Mysten Labs\\ University College London (UCL)}
  \country{}
}
\email{alberto@mystenlabs.com}

\author{Lefteris Kokoris-Kogias}
\affiliation{%
  \institution{Mysten Labs}
  \country{}
}
\email{lefteris@mystnelabs.com}

\author{Mohammad Sadoghi}
\affiliation{
  \institution{Exploratory Systems Lab\\ University of California, Davis}
  \country{}
}
\email{msadoghi@ucdavis.edu}


\renewcommand{\shortauthors}{J. Chen, et. al}

\begin{abstract}
    Sharding has emerged as a critical technique for enhancing blockchain system scalability. 
However, existing sharding approaches face unique challenges 
when applied to Directed Acyclic Graph (DAG)-based protocols 
that integrate expressive smart contract processing. 
Current solutions predominantly rely on coordination mechanisms 
like two-phase commit (2PC) and require transaction read/write sets to optimize parallel execution. 
These requirements introduce two fundamental limitations: 
(1) additional coordination phases incur latency overhead, 
and (2) pre-declaration of read/write sets proves impractical 
for Turing-complete smart contracts with dynamic access patterns.

This paper presents \sysname{}, a novel sharding architecture for both single-shard transactions ($\SST{}s$) 
as well as cross-shard transactions ($\CST{}s$), and it enables non-blocking reconfiguration to ensure system liveness. 
Our design introduces four key innovations: 
First, each replica serves dual roles as a full-shard representative and transaction proposer 
(shard proposer), employing differentiated execution models: the Execution-Order-Validation (EOV) model for $\SST{}s$ 
and Order-Execution (OE) model for $\CST{}s$. 
Second, we develop a DAG-based coordination protocol 
that establishes deterministic ordering between two transaction types 
while preserving concurrent execution capabilities. 
Third, we implement a dynamic concurrency controller that 
schedules $\SST{}s$ without requiring prior knowledge of read/write sets, 
enabling runtime dependency resolution.
Fourth, 
\sysname{} introduces a non-blocking shard reconfiguration mechanism to address censorship attacks by featuring frequent shard re-assignment without impeding the construction of DAG nor blocking consensus.
This approach maintains continuous DAG construction 
and consensus progress while preventing persistent adversarial control through periodic shard reassignment. 
\sysname{} achieves a $50 \times$ improvement with 64 replicas 
over serial Tusk execution. 

\end{abstract}
\vspace{-4mm}




\maketitle

\vspace{-2mm}
\section{Introduction}
The emergence of blockchain technology has spurred significant interest 
in developing resilient systems capable of processing data and transactions 
under Byzantine conditions, 
including software errors, hardware failures, network disruptions, 
and coordinated malicious attacks
~\cite{amiri2019caper,hyperledger-fabric, gorenflo2020fastfabric, el2019blockchaindb, gupta2020permissioned, nathan2019blockchain}. 
These systems enhance reliability and security by enabling collaboration 
among multiple independent participants
~\cite{casey2018impact, ge2017blockchain, herlihy2019blockchains, kamel2018geospatial, bitcoin, 
narayanan2017bitcoin, nawab2019blockplane, pisa2017blockchain, wood2014ethereum}. 

Smart contracts~\cite{smartcontract, szabo1997formalizing}, 
as programmable transaction frameworks embedded in blockchain platforms, 
empower developers to address real-world challenges through decentralized solutions~\cite{buterin2014next, kushwaha2022ethereum}. 
However, their adoption is often hindered by execution delays caused 
by runtime contract code analysis~\cite{aldweesh2019opbench}. 
To overcome this limitation, 
researchers are actively exploring performance optimization strategies 
for contract-based blockchain systems.

Several strategies have emerged to improve execution within blockchain systems. 

{\bf Transaction Processing Models:} One practical approach involves enhancing transaction processing. 
Most blockchain systems adopt the Order-Execute ($OE$) model, 
where transactions are ordered through consensus before execution~\cite{pbftj,hotstuff, rcc, poe}. 
$OE$-based systems often employ deterministic concurrency controls 
by constructing transaction dependency graphs to optimize parallelism
~\cite{qadah2018quecc, faleiro2017high, yao2016exploiting, wang2016scaling}. 
However, platforms like Hyperledger Fabric~\cite{hyperledger-fabric} 
utilize the Execute-Order-Validate ($EOV$) model, 
executing transactions before consensus, to enhance flexibility and 
Optimistic Concurrency Control (OCC)~\cite{kung1981optimistic} to improve the concurrency.

{\bf Scalability via DAG and Sharding:} Another critical area for advancement in blockchain technology is scalability, 
particularly in supporting parallel execution. 
Recent advancements leverage Directed Acyclic Graph (DAG)-based consensus protocols 
to improve scalability. 
These protocols enable replicas to submit proposals concurrently 
by building a DAG that links new proposals to historical ones. 
This architecture has gained significant recognition in the industry 
due to its robust security features, exceptional scalability, and capability to support smart contracts
~\cite{narwhat-tusk, spiegelman2022bullshark,keidar2021all,keidar2022cordial,stathakopoulou2023bbca,mysticeti,arun2024shoal,spiegelman2023shoal,shrestha2024sailfish,malkhi2023bbca, raikwar2024sok}.
Complementary to DAG-based approaches, 
sharding techniques allow parallel transaction processing across each shard,
reducing consensus overhead
~\cite{geobft,ext_byshard,kokoris2018omniledger,dang2019towards,zamani2018rapidchain, al2017chainspace, sonnino2020replay, zhang2023sharding, yang2020review}. 

However, the above approaches do not effectively improve the execution of smart contracts:
\begin{enumerate}[wide,nosep,label=\textbf{Challenge\arabic*}:,ref={Challenge\arabic*}]
       \item \textbf{Enhancing Transaction Parallelism without advanced knowledge.} 
       While $OE$-based solutions leverage dependency graphs to optimize parallelism,
       they require prior knowledge of transaction read/write 
       sets,a constraint incompatible with dynamic smart contracts. 
       Conversely, $EOV$-based approaches face high transaction conflict rates, 
       necessitating advanced conflict resolution algorithms.

       \item \textbf{Efficient Cross-Shard Transaction Processing.} 
       Existing solutions for cross-shard atomicity, such as relay-based protocols (Sharper~\cite{amiri2021sharper}, BrokerChain~\cite{huang2022brokerchain}, and SharDAG~\cite{cheng2024shardag}) and traditional Two-Phase Commit (2PC)~\cite{dang2019towards, hong2022scaling}, introduce significant delays due to inter-shard coordination. 
        While multi-shard consensus~\cite{amiri2021sharper, qi2022dag} mitigates 2PC limitations, 
        it sacrifices scalability in large-scale networks with high contention.
\end{enumerate}

The challenges described above raise the question of whether it is possible to design a sharding system that does not depend on
understanding the read and write sets of transactions,
nor requires additional coordinators to manage cross-shard transactions.

We propose \sysname{}, an innovative sharding architecture 
that seamlessly processes both single-shard ($\SST{}s$) and cross-shard transactions ($\CST{}s$) 
without centralized coordinators. 
\sysname{} effectively integrates the $OE$ and $EOV$ models, 
where the $EOV$ model enhances parallelism in the execution of $\SST{}s$, 
while the $OE$ model minimizes the abort rate when handling $\CST{}s$ across different shards. 
Furthermore, the $OE$ model ensures a coordinated execution order between these two types of transactions, 
maintaining the overall correctness of the execution process.

Similar to conventional sharding systems, 
\sysname{} organizes transactions into distinct shards to mitigate potential conflicts. 
In particular, each replica within \sysname{} corresponds to a single shard and acts as a shard proposer, proposing transactions within that shard. 
\sysname{} employs a DAG-based consensus protocol 
~\cite{narwhat-tusk, spiegelman2022bullshark,keidar2021all,keidar2022cordial,stathakopoulou2023bbca,mysticeti,arun2024shoal,spiegelman2023shoal,shrestha2024sailfish,malkhi2023bbca} 
to reach an agreement on the execution results provided by each shard proposer.

Inspired by Sui's epoch switching~\cite{sui-lutris}, 
\sysname{} employs round-robin scheduling~\cite{rasmussen2008round} to rotate shard proposers periodically, 
enhancing the system's security and liveness.
Proposer rotation is triggered on-demand if a shard fails to propose transactions within a timeout, 
with seamless DAG transitions preserving protocol continuity.

\begin{figure}[t]
        \centering
        \includegraphics[width=0.35\textwidth]{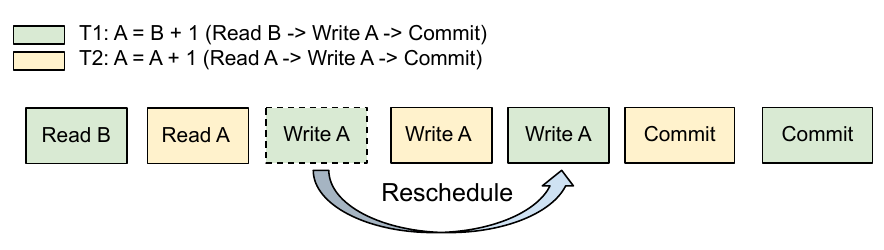}
        \vspace{-4mm}
        \caption{A transaction rescheduling to avoid abortion by moving the $Write A$ on $T_1$ after the $Write A$ on $T_2$.
        They all obtain the correct result based on their operations.}
        \label{fig:ce}
        \vspace{-4mm}
\end{figure}

\sysname{} also introduces a concurrent executor ($CE$) designed 
to improve the execution of $\SST{}s$ before reaching consensus. 
Unlike traditional concurrency protocols that primarily manage conflicts 
based on the order of arrival
~\cite{agrawal1987concurrency,franaszek1985limitations,tay1985locking},
the $CE$ utilizes a non-deterministic ordering system based on the execution run-time states 
of each transaction. 
This innovative approach minimizes the abort rate, 
thereby reducing transaction latency. 
As illustrated in \Cref{fig:ce}, the $CE$ effectively reschedules transactions 
based on their run-time executions to prevent aborts. For instance, transaction 
$T_2$, which would ordinarily conflict with the write operation of transaction 
$T_1$, can be successfully committed without cancellation.

In summary, this paper makes the following contributions.
\vspace{-1mm}
\begin{itemize}
      \item To our knowledge, \sysname{} is the first sharding consensus mechanism 
                that combines the $OE$ and $EOV$ models based on DAG-based protocols 
                without requiring any additional coordinators 
                to determine the order between $\SST{}s$ and $\CST{}s$.
      \item We introduce a new concurrency paradigm that implements a parallel preplay  
        for $\SST{}s$ (concurrent consensus execution). 
        \sysname{} preplays $\SST{}s$ followed by parallel verification without needing 
        prior knowledge of the read/write sets.
    \item   \sysname{} features a non-blocking shard reconfiguration protocol that allows 
                for the rotation of shard assignments without pausing 
                either DAG dissemination or the consensus layer. 
    \item We have implemented a concurrent executor to enhance the parallelism 
        of executing smart contracts without prior knowledge of the read/write sets on $\SST{}s$. 
        The execution engine dynamically arranges transactions based on current assessments 
        to reduce abort rates due to conflicts.
    \item Our evaluation of \sysname{} demonstrates an impressive $50 \times$ speedup over Tusk~\cite{narwhat-tusk} 
        with sequential execution using the SmallBank workload on $64$ replicas built 
        on Apache ResilientDB (Incubating)~\cite{apache-resdb, rcc}.
\end{itemize}

\vspace{-2mm}
\section{Background}
\para{Smart Contract}
A smart contract is a digital protocol, introduced by Nick Szabo in the mid-1990s~\cite{szabo1997formalizing},  designed to facilitate, verify, or enforce the negotiation or performance of a contract automatically. Unlike traditional contracts, which rely on legal systems for enforcement, smart contracts are self-executing and operate on blockchain technology. They are written in code and run on decentralized platforms like Ethereum, ensuring transparency, security, and immutability. Smart contracts eliminate the need for intermediaries, reduce the risk of fraud, and enable trustless transactions between parties. They are widely used in various applications, including decentralized finance (DeFi), supply chain management, and digital identity verification. 

These contracts consist of custom functions that operate on user accounts with associated balances. Once deployed to the network, transactions invoking the functions specified in the contract are proposed to execute predefined operations that interact with the user accounts. However, the execution of the contract code occurs within the Ethereum Virtual Machine (EVM)~\cite{eevm}, which results in the read and write sets of the contract being indeterminate prior to execution.

\para{DAG-based protocols}
\label{ss:preliminaries}
\sysname{} leverages a Directed Acyclic Graph (DAG) structure to address scalability challenges in blockchain systems. This architecture enables efficient transaction proposal mechanisms while maintaining robust security, high throughput, and native support for smart contracts. DAG-based protocols have gained significant traction in the industry due to their ability to decouple transaction dissemination from consensus processes~\cite{narwhat-tusk, spiegelman2022bullshark,keidar2021all,keidar2022cordial,stathakopoulou2023bbca,mysticeti,arun2024shoal,spiegelman2023shoal,shrestha2024sailfish,malkhi2023bbca, raikwar2024sok}.

\begin{figure}[t]
      \centering
      \includegraphics[width=0.48\textwidth]{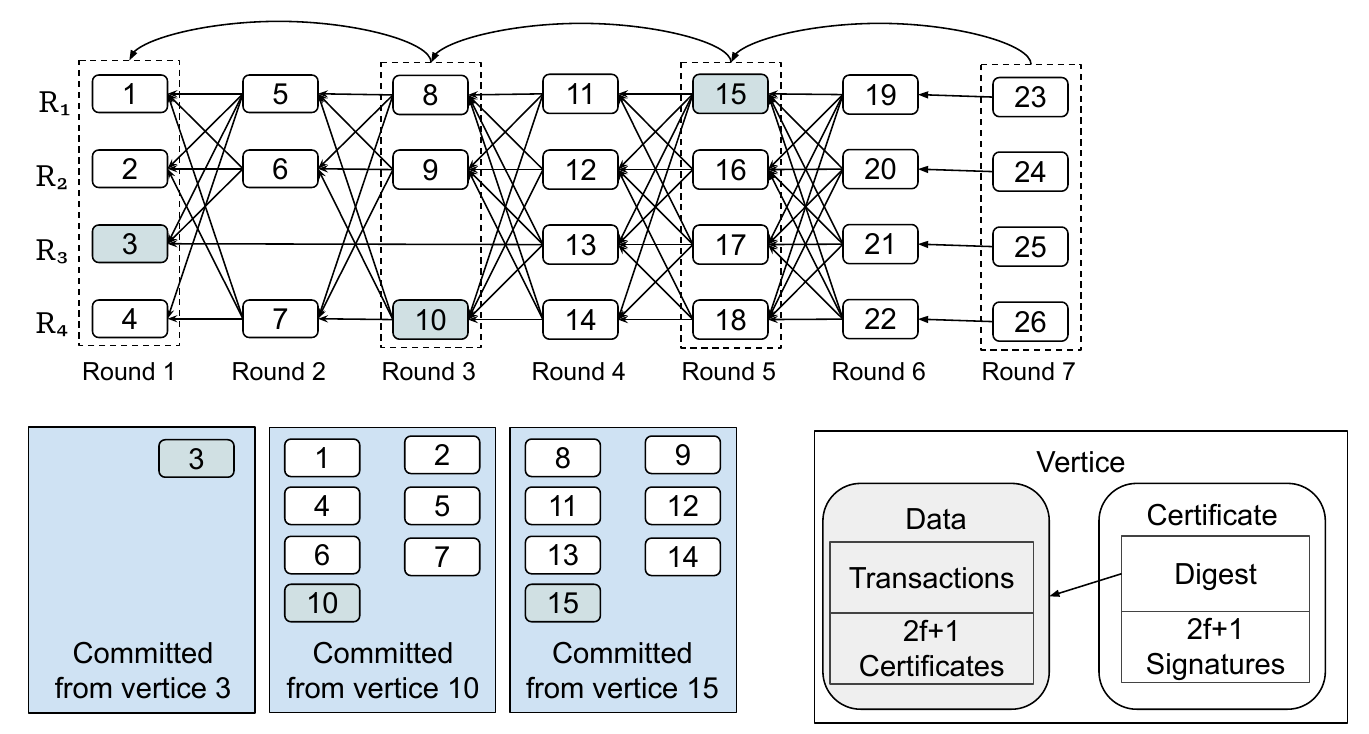}
      \vspace{-4mm}
      \caption{Overview of a DAG-based protocol on Tusk generated by 4 replicas. 
      Each vertice contains both data and a certificate. 
      Each data refers to $2f + 1$ certificates from the previous round. 
      Each certificate is generated from its corresponding data, which has received signatures 
      from $2f + 1$ replicas. 
      The leader vertice for round $r$, the solid vertice used to commit the data in its causal history, 
      will be determined before processing round $r + 2$.}
      \label{fig:tusk_dag}
      \vspace{-4mm}
\end{figure}

In contrast to traditional linear blockchains, DAG-based protocols allow multiple replicas to propose transactions concurrently. These transactions are constructed into a deterministic DAG structure, ensuring a consistent topological ordering across all honest replicas. Recent advancements in this domain, including Tusk~\cite{narwhat-tusk,spiegelman2022bullshark}, BBCA-Chain~\cite{malkhi2023bbca}, Shoal/Shoal++~\cite{spiegelman2023shoal,arun2024shoal}, Mysticeti~\cite{babel2023mysticeti}, and Cordial Miners~\cite{keidar2022cordial}, demonstrate how DAGs streamline consensus by separating data propagation from finality mechanisms.

The protocol operates in synchronized rounds, where each DAG vertex in a round consists of two components:
\begin{itemize}
    \item Data Payload: Contains transactions and references to at least $2f+1$ certificates from the prior round.
    \item Certificate: A quorum of $2f+1$ cryptographic signatures attesting to the validity of the vertex and its dependencies.
\end{itemize}

During each round, replicas broadcast their proposed vertices to the network. A vertex becomes certified once 
$2f+1$ signatures are collected, enabling it to serve as a dependency for new vertices in subsequent rounds. This iterative process ensures liveness while preserving the DAG’s causal ordering.

Vertex finalization occurs at fixed intervals, typically every two rounds in Tusk~\cite{narwhat-tusk} or three rounds in DAG-Rider~\cite{keidar2021all}. A designated leader (selected via round-robin scheduling~\cite{rasmussen2008round} or distributed randomness~\cite{boneh2001short}) proposes a vertex for commitment. A leader vertex in round $r$ is eligible to be committed during round $r+2$ (as in Tusk):
1) The replica must have received at least $2f + 1$ vertices from round $r + 1$, 
and 2) The leader vertex must be referenced by a minimum of $f + 1$ vertices in round $r + 1$.

DAG-based protocols provide the following properties:
\label{ss:dag_propoties}
\begin{itemize}
      \item Validity: if an honest replica $R$ has a vertice $B$ in its local view of the DAG,
            then $R$ also has all the causal history of $B$.
      \item Consistency: if an honest replica $R$ obtains a vertice $B_r$ in round $r$ from replica $P$,
            then, eventually, all other honest replicas will have $B_r$.
      \item Completeness: if two honest replicas have a vertice $B_r$ in round $r$,
            the causal histories of $B_r$ are identical.
\end{itemize}

\vspace{-2mm}
\section{\sysname Overview}
\sysname{} advances smart contract execution efficiency using an innovative sharding architecture augmented by a dynamic shard reconfiguration mechanism. This mechanism counters potential censorship attacks, such as post-execution transaction suppression or biased transaction selection, ensuring network integrity.

Unlike conventional sharding systems that partition replicas into isolated groups governed by separate consensus protocols, \sysname{} employs a single unified consensus protocol jointly maintained by all replicas. Each replica operates as a shard proposer, processing transactions specific to its assigned shard. To coordinate cross-shard transaction ordering and execution, \sysname{} leverages a DAG-based consensus protocol, enabling global agreement on transaction validity while preserving shard-level parallelism via the predetermined leaders underlying the consensus.

\sysname{} employs both $EOV$ model and $OE$ model to address $\SST{}s$ and $\CST{}s$.
\begin{itemize}
    \item $\SST{}s$ ($EOV$ Model):
    Transactions confined to a single shard are processed non-deterministically by a Concurrent Executor ($CE$) within the corresponding shard proposer. After the local preplay, results undergo a consensus across all shards to ensure state consistency.
    \item $\CST{}s$ ($OE$ Model):
    $\CST{}s$ employ an optimistic concurrency control protocol with deterministic finalization. Atomic commitment is achieved post-consensus, allowing tentative execution optimistically druing guaranteeing rollback-free confirmation.
\end{itemize}

\begin{figure}[t]
    \centering
    \includegraphics[width=0.48\textwidth]{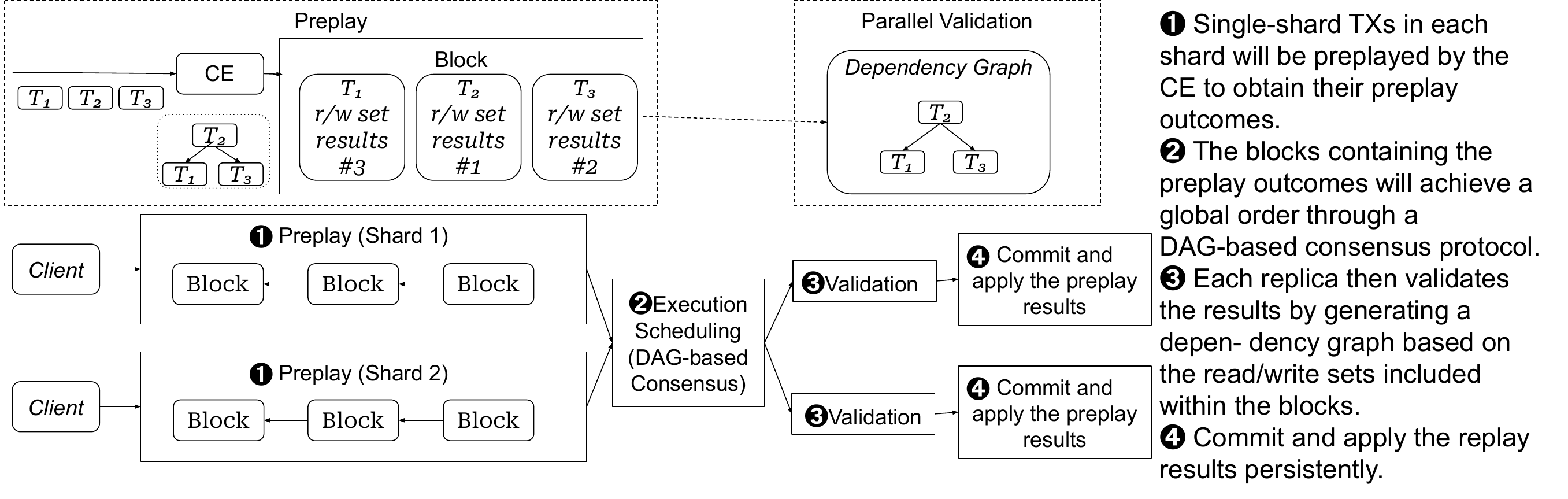}
    \vspace{-6mm}
    \caption{
        The dataflow of $\SST{}s$.
        }
    \vspace{-4mm}
    \label{fig:single_shard_fabric}
\end{figure}

\Cref{fig:single_shard_fabric}  demonstrates an example 
where two shard proposers propose two single-shard transactions. 
Further details on $\SST{}s$ and $\CST{}s$ are provided 
in Section~\ref{ss:single_txn} and \ref{ss:cross_txn}.

\sysname{} also allows the migration of each shard to another replica to avoid a censorship attack, 
such as dropping transactions.

\vspace{-2mm}
\subsection{System, Threat and Data Model}
\label{ss:system_model}

In this section, we describe the system, threat and the data model.
We leave the discussion of the $\SST{}s$, the $\CST{}s$,
and the shard reconfiguration in 
 in \Cref{ss:single_txn},
\ref{ss:cross_txn},
and \ref{ss:group_service}.


\vspace{-2mm}
\para{System Model}
Thunderbolt is composed of $n$ replicas, 
each of which serves a dual purpose: 
functioning as a shard and a replica. 
As a shard, a node maintains a distinct set or partition of data. 
In its role as a replica, it preserves a copy of the transaction log. 
Thus, each replica may also be designated as a shard proposer. 
Furthermore, clients direct transactions to the appropriate shards, 
and every replica engages in the consensus process 
to establish a cohesive order for these transactions.

In summary, a node in \sysname{} assumes three essential roles:
1). It operates as a shard proposer, managing $\SST{}s$ 
within each shard in an independent manner.
2). It acts as a replica that contributes to the consensus process.
3). It serves as a leader that commits $\CST{}s$ 
in a total order in accordance with the consensus protocol.
For clarity in the accompanying illustrations, 
the terms "replicas" and "shard proposers" may be used interchangeably. 
We will utilize the term "shard proposers" when delineating shard procedures 
and will revert to "replicas" in discussions regarding the consensus protocols.

\para{Threat model}
We consider a set of $n$ replicas (or shards), where at most $f$ of these replicas can be faulty and $n = 3f + 1$. 
The $f$ faulty replicas may exhibit any arbitrary behavior, 
including Byzantine failures, while the remaining replicas are assumed to be honest and will adhere to 
the protocol’s specifications at all times. 
Additionally, we assume that clients are not trustworthy and would not expect to send transactions to all of the shards associated with their transactions. 
The network is expected to be eventually synchronous~\cite{dwork1988consensus}, meaning that messages sent from a replica will eventually arrive within a global stabilization time ($GST$), 
which remains unknown to the replicas. 
Communications between replicas use authenticated point-to-point channels, 
with messages signed by the sender using a public-private key pair for authentication.


%
%


\para{Data model}
The data model assumes that each transaction includes a contract code
with functions to access data belonging to the sender in the shard.
The contract involves two types of operations:
<$Read, K$> and <$Write, K, V$>.
Here, $K$ represents the key required for access and $V$ is the value that needs to be written
to the key $K$.
The contract code is Turing-complete
and users could not obtain any information without execution. 
We also assume that the functions of the contract are idempotent. 

Our system is designed with the understanding that data must be partitioned and that each key is assigned a shard ID ($SID$) before it can be utilized. 
These $SIDs$ are predefined and recognized across all shards.  
They fulfill a dual function: they guide transactions to the appropriate shard proposer and 
support parallel processing among multiple shards, 
thereby enhancing overall system efficiency (\Cref{ss:cross_txn}).
The actual method for partitioning is orthogonal to our work and 
any existing techniques can be utilized~\cite{cattell2011scalable,chen2014split, zaharia2016apache, chaiken2008scope, navathe1984vertical}. 

\vspace{-2mm}
\section{Single Shard Transactions}
\label{ss:single_txn}

\sysname{} processes $\SST{}s$ through three core components: preplay, execution scheduling, and validation. During each round, a shard proposer initiates the workflow by preplaying a batch of $\SST{}s$. This generates a block containing critical preplay outcomes, which the proposer propagates to other shards via a DAG-based consensus protocol. During consensus, these blocks undergo parallel validation across shards. Once a replica commits a block, it applies the preplay results to its storage.

\para{Preplay}
\label{ss:execution}
In \sysname{}, shard proposers play a pivotal role by preplaying transactions to determine their outcomes before block creation. A concurrent executor ($CE$) is employed to preprocess batches of transactions efficiently, producing detailed outputs for each transaction. These outputs include:  the read/write sets accessed during execution, the corresponding execution results, and a scheduled execution order (as depicted in \Cref{fig:single_shard_fabric}).

The scheduled order establishes a deterministic serialized sequence, ensuring transaction results remain consistent when executed in the prescribed order. The read/write sets reveal the specific data accessed by each transaction. Crucially, these sets cannot be predetermined and are derived exclusively via the preplay process.

\para{Execution Scheduling}
\sysname{} integrates with a DAG-based dissemination layer that employs a consensus protocol to establish a total block order across replicas. In each round 
$r$, a shard proposer $R$ proposes a $CE$-generated block to the DAG, creating a new vertex in the graph. 
This vertex links to all prior vertices, including those proposed by $R$ in round $r-1$. 
For clarity, "vertices" in the DAG are hereafter referred to as "blocks" in the following sections, implying they have been certified by the protocol.

\para{Validation}
\label{ss:validation}

When a replica receives data for round $r$ via the DAG (\Cref{ss:preliminaries}), \sysname{} initiates a rigorous validation process to verify the integrity of preplay results within the blocks. Validators construct a local dependency graph using the read/write sets to enable parallel transaction validation rather than sequential checks, optimizing system throughput. Notably, blocks from round $r-1$ are validated before those from round $r$ for the same shard proposer.

During re-execution, validators confirm that the computed read sets match the values recorded in the block. A valid dependency graph guarantees consistent read-set results and ensures the final state of each key aligns with the block’s declared values. If discrepancies in read-set values are detected, the block is flagged as invalid and discarded. Until blocks are committed, replicas retain preplay results in local storage, either to process $\CST{}s$ (\Cref{ss:cross_txn}) or until DAG reconfiguration occurs (\Cref{ss:group_service}).

\vspace{-2mm}
\section{Cross-shard transactions}
\label{ss:cross_txn}
$\CST{}s$ involve multiple shards and require
consensus to establish a total execution order. This total order ensures that all $\CST{}s$ are executed consistently across the involved shards.
However, each shard proposer preplays $\SST{}s$ independently and replicates the results; in contrast, 
the total order for the $\CST{}s$ must be determined first before they can be executed (for instance, preplay optimization cannot be employed for $\CST{}s$).
Therefore, \sysname{} must coordinate the sequencing between the $\CST{}s$ and the $\SST{}s$ to guarantee consistent execution outcomes across all replicas.
Prior approaches often rely on coordinators to establish inter-shard order~\cite{cheng2024shardag,huang2022brokerchain,
dang2019towards, hong2022scaling, kokoris2018omniledger,zamani2018rapidchain,
amiri2021sharper, qi2022dag}, but these introduce communication overhead that degrades performance.

\vspace{-2mm}
\subsection{Rules for Proposals}
To address issues, 
\sysname{} leverages the DAG’s predetermined leaders to 
enforces a consistent partial order between single-shard and cross-shard transactions 
via the following rules:

\begin{enumerate}[label=G\arabic*), ref={G\arabic*}]
\item \label{enum:exe_rule}
If a leader $L$ commits both a $\SST{}$ and a $\CST{}$, the $\SST{}$ must be committed first.
\item 
\label{enum:order_rule}
If leader $L_i$ commits $\CST{}$ $X$ in round $i$, any $\SST{}$ $Y$ committed by leader  $L_j$ in round $j$ (where $j>i$) cannot execute until $X$ is finalized.
\end{enumerate}

\begin{figure}[t]
    \centering
    \includegraphics[width=0.48\textwidth]{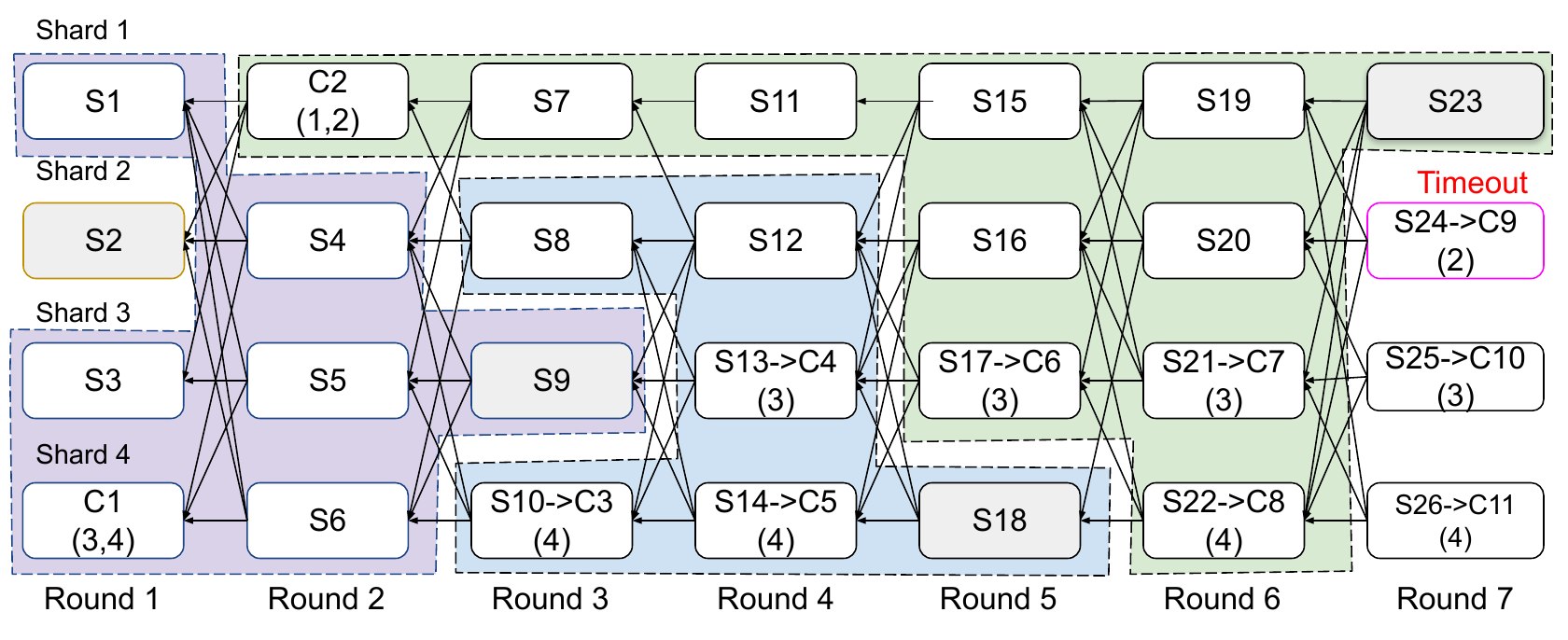}
    \vspace{-4mm}
    \caption{
        An example illustrates the commitment of each leader in the process of 
        converting $\SST{}s$ to $\CST{}s$ on Tusk. 
        In this DAG, $Si$ denotes a $\SST{}$, 
        while $Ci (\{X\})$ represents a $\CST{}$ associated with shards $\{X\}$. 
        The leaders in the odd rounds are selected using round-robin selection.
        The blocks committed by the same leader have borders of the same color.
        $Si$ will be converted to $Ci$ if there is any conflict blocks ($S10$) 
        or could not receive the leader block in time ($S24$).
        }
    \vspace{-4mm}
    \label{fig:cross_fabric}
\end{figure}

It is worth knowing that Leaders are predetermined per round using deterministic methods (e.g., round-robin or global random coins). To enforce these rules, \sysname{} applies the following proposal rules:
\begin{enumerate}[label=P\arabic*), ref={Rule P\arabic*}]
    \item 
    $\CST{}s$ are submitted directly to the DAG, bypassing the $CE$.
    \item \label{enum:cross_commit}
    Leaders committing a batch of transactions must finalize all $\SST{}s$ before $\CST{}s$.
\item \label{enum:ss_cs} 
If a shard proposer $SL$ proposes a $\SST{}$ $X$ in round $r$, and the current round’s leader $L$ differs from $SL$, $SL$ must:
    \begin{itemize}
        \item Wait for $L$’s proposal before preplaying $X$.
        \item Convert $X$ to a $\CST{}$ if any uncommitted $\CST{}$ $Y$ in $L$’s history conflicts with $X$.
        \item Otherwise, preplay $X$ and submit the results.
    \end{itemize}
    
    \item \label{enum:ss_cs_2} 
    If a shard proposer $SL$ proposes a $\SST{}$ $X$ in round $r$, and a prior leader’s uncommitted $\CST{}$ $Y$ (in round q < $r$) conflicts with $X$, $SL$ converts $X$ to a $\CST{}$.
    
    \item \label{enum:leader_commit}
    If leader $L$ in round $r$ commits a $\CST{}$ $X$ related to shard A but lacks A’s proposal in round $r-1$, $L$ defers committing A and A’s subsequent proposals.
    \item \label{enum:time_out} 
    If a shard proposer $SL$ proposes a $\SST{}$ $X$ in round $r$ but the leader $L$’s proposal is delayed beyond a timeout, $SL$ converts $X$ to a $\CST{}$.
\end{enumerate}
These rules enable \sysname{} to process $\CST{}s$ without blocking shards while maximizing parallelism for single-shard transaction preplay.
\begin{example}
\Cref{fig:cross_fabric} illustrates \sysname{}’s handling of single-shard and cross-shard transactions on Tusk:
$\SST{}s$ will be executed before $\CST{}s$ under the same leader (\ref{enum:cross_commit}).
$S10$ is converted to $\CST{}$ $C1$ due to its dependency on $S9$’s leader history (\ref{enum:ss_cs}).
$S{13}$ and $S{14}$ become $C4$ and $C5$, respectively, because $C1$ remains uncommitted until round $5$  (\ref{enum:ss_cs_2}).
    Leader $S18$ in round $5$ skips $C2$ and subsequent transactions due to missing $S11$ (\ref{enum:leader_commit}).
In round $7$, $S24$ converts $S4$ to $\CST{}$ $C9$ after failing to receive leader $S23$’s proposal (\ref{enum:time_out}).
\end{example}


\vspace{-2mm}
\subsection{Parallel Execution}
\label{ss:paral_exe}
When processing $\CST{}s$, \sysname{} preserves all sharding metadata for each transaction. Rather than sequential execution, \sysname{} employs deterministic concurrency control mechanisms, such as QueCC~\cite{qadah2018quecc}, to construct dependency graphs using cross-shard metadata (SID). This enables parallel execution while maintaining consistency across shards.

\begin{figure}[t]
    \centering
    \includegraphics[width=0.48\textwidth]{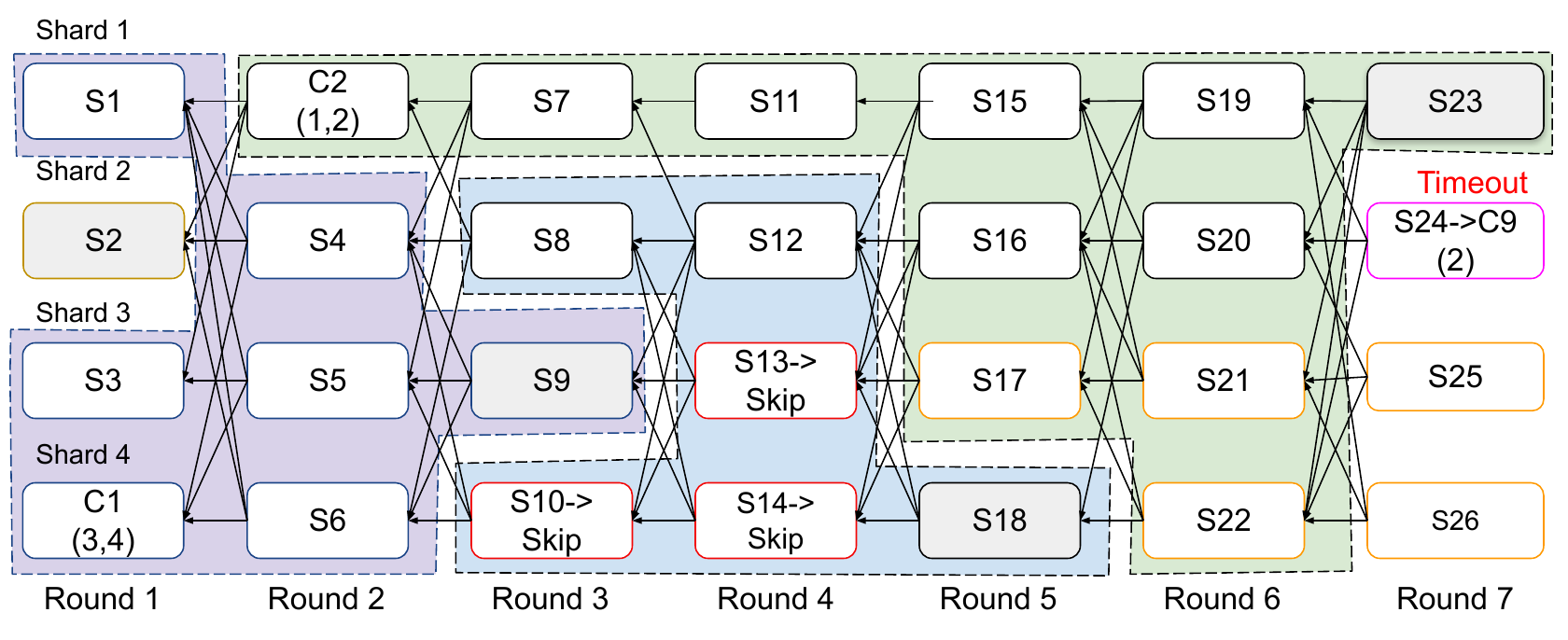}
    \vspace{-4mm}
    \caption{
        An example of proposing skipping blocks to 
        restart the preplay of the $\SST{}s$. 
        The $\SST{}s$ ($S17$, $S21$, and $S22$), which would be converted to $\CST{}s$ after round $5$ in \Cref{fig:cross_fabric},
        can replay their executions before delivering to the DAG.
        }
    \vspace{-4mm}
    \label{fig:sst_restart}
\end{figure}

\begin{figure*}[t]
    \centering
    \includegraphics[width=0.96\textwidth]{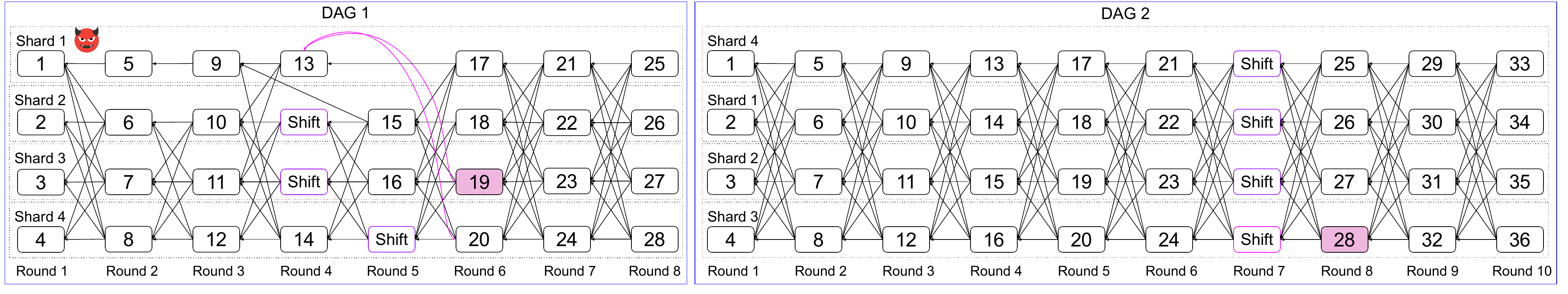}
    \vspace{-4mm}
    \caption{
        An example of DAG reconfiguration.
        Shard 1 in DAG 1 which is a malicious replica that delays the blocks from round 2, 
        which triggers a reconfiguration to a new DAG ($K = 2$). 
        Additionally, a further reconfiguration to DAG 2 will take place 
        after the replicas have proposed six rounds ($K^{'} = 6$).
    }
    \vspace{-2mm}
    \label{fig:shard_rotate}
\end{figure*}

\vspace{-2mm}
\subsection{Message Failures}
In practical deployments, network unreliability may delay message delivery. To mitigate this:
\begin{itemize}
    \item If a leader $L$ cannot include all $\SST{}s$ linked to a $\CST{}s$ due to network delays (e.g., missing shard proposals), $L$ bypasses the affected $\CST{}$ and subsequent transactions from the same shard ($C2$ in \Cref{fig:cross_fabric} on \ref{enum:leader_commit}). This prevents violations of global ordering \ref{enum:order_rule} by excluding incomplete transaction sets. These transactions are later finalized by subsequent leaders.
    \item 
If a shard proposer fails to receive the leader’s proposal within a round’s timeout window, it cannot preplay its $\SST{}$ (\ref{enum:time_out}). The proposer instead promotes the transaction to a $\CST{}$ and submits it directly to the DAG (such as $S{24}$ in \Cref{fig:cross_fabric}).

\end{itemize}

\vspace{-2mm}
\subsection{Preplay Recovering}
\label{ss:preplay_recover}
Under \ref{enum:ss_cs}, a shard proposer $SL$ must convert a $\SST{}$ to a $\CST{}$ if it detects conflicting uncommitted $\CST{}s$ in prior leader histories. While this ensures safety, it forfeits the performance benefits of preplay, such as the blocks in shard $3$ after round $3$ in \Cref{fig:cross_fabric}.

To recover preplay, $SL$  must verify that all conflicting $\CST{}s$  have been finalized by preceding leaders. Since a leader $CL_r$ in round $r$ is finalized within two subsequent rounds (\Cref{ss:preliminaries}), $SL$ can safely preplay new single-shard transactions once: 1) It receives $2f + 1$ certificates in round r + 1, and
2) $CL_r$ is referenced by at least $f + 1$ blocks in round $r + 1$.
If $SL$ identifies any conflicting $\CST{}s$ while proposing $\SST{}s$, $SL$ submits skip blocks to the DAG until prior leaders are finalized. For instance, $S10$, $S13$, and $S14$ in \Cref{fig:sst_restart} are converted to skip blocks. Consequently, subsequent transactions, like $S17$ and $S22$, are reverted to the $EOV$ model, restoring preplay capabilities.

\vspace{-2mm}
\section{Shards Reconfiguration}
\label{ss:group_service}
In Byzantine environments, compromised replicas may enable attackers to launch censorship attacks, threatening the integrity of transactions within their assigned shards.

Thunderbolt employs a round-robin selection mechanism~\cite{shreedhar1995efficient} to rotate shard proposers if a leader fails to propose transactions for $K$ rounds or at fixed intervals of $K^{'}$ rounds (where $K^{'} > K$).

This rotation serves dual purposes:
\begin{itemize}
    \item 
Preventing transaction duplication (DDOS~\cite{mirkovic2004taxonomy,douligeris2004ddos}): Proposers perform local deduplication to block malicious clients from flooding the system with redundant transactions, a known challenge in DAG-based protocols~\cite{sui-lutris,spiegelman2022bullshark,mysticeti}.
\item
Mitigating censorship: Regular rotation limits the window for a compromised proposer to disrupt operations.
\end{itemize}


Unlike traditional consensus protocols
that depend on notification messages to alter primary replicas,
\sysname{} introduces an innovative mechanism that uses the underlying DAG protocols
to facilitate the seamless transition to a new DAG and reconfigure shard proposers.
We leverage a round-robin approach to select a new proposer that if the current proposer of shard X is replica $R_i$,
the subsequent proposer of shard X will be $R_{(i\ mod\ n)+1}$.

However, the transmission of blocks to a new proposer may experience delays or omissions due to network issues or the actions of a malicious proposer.
If the new proposer for round $r$ is unable to receive the proposal committed in round $r - 1$ from the previous proposer, the new proposer will stop operations until the block arrives to ensure safety.

To address this challenge, \sysname{} introduces Shift blocks to reach agreements
among shards when a shard reconfiguration should be initiated
and switch to a new DAG to process further transactions.

A replica $R$ in a shard broadcasts a Shift block in round $r$ under the following conditions:
\begin{enumerate}
    \item $R$ receives no block from a shard proposer after round $r-K$.
    \item $R$ has proposed blocks for at least $K^{'}$ rounds.
    \item $R$ received $f+1$ Shift blocks from distinct replicas at round $r-1$.
    \item $R$ does not broadcast the Shift block before.
\end{enumerate}

\begin{example}
In \Cref{fig:shard_rotate}  where $K=2$ and $K^{'}=6$, shards 2 and 3 broadcast Shift blocks in round 4 after failing to receive blocks from shard 1 in rounds 2–3. Despite receiving a block in round 4, shard 4 broadcasts a Shift block in round 5 upon receiving two Shift blocks from peers, prioritizing liveness assurance.
\end{example}

\begin{figure}[t]
    \centering
    \includegraphics[width=0.48\textwidth]{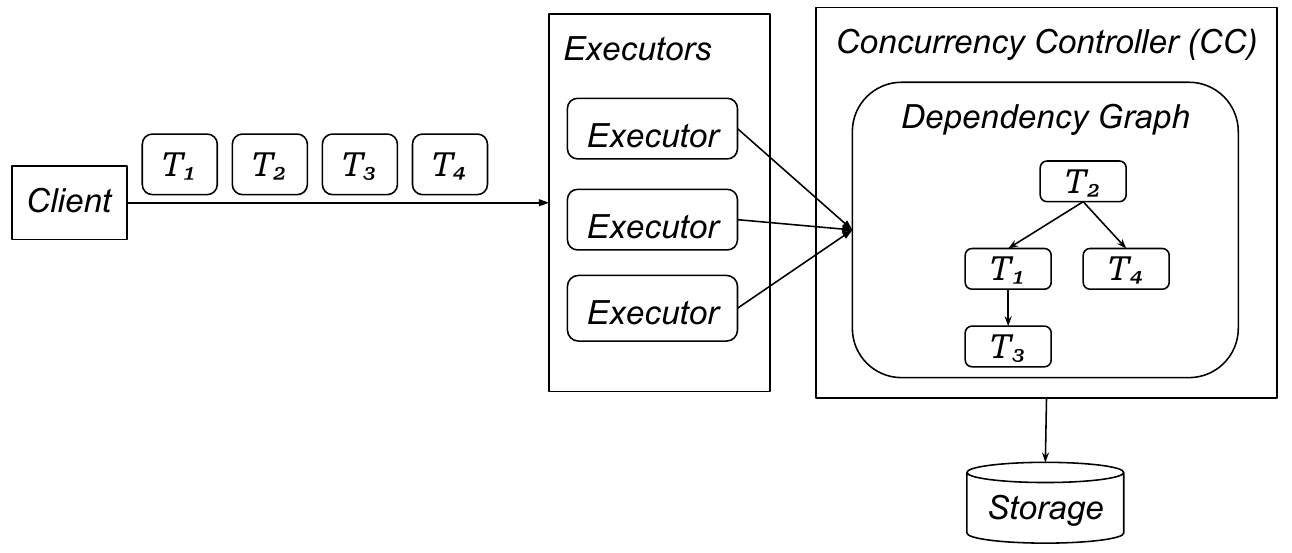}
    \vspace{-4mm}
    \caption{
        The architecture of the $CE$ consists of a set of executors that execute transactions. 
        The concurrency controller utilizes a dependency graph 
        to determine the order of transactions and their execution results.}
    \label{fig:x_protocol_intro}
    \vspace{-4mm}
\end{figure}

\para{Non-blocking Reconfiguration}
\label{ss:non-blocking}
Since at $2f+1$ honest replicas will commit the same block 
during the same round (\Cref{ss:dag_propoties}), 
we designate the round of the first commit block that includes $2f + 1$ Shift blocks 
as the ending round for the current DAG.
Then,
\textbf{each replica will begin a new DAG from the same ending round 
to ensure the system's safety.} 
For instance, shard $2$ will propose block $15$ at round $5$ after proposing a Shift block at round $4$. 
Finally, block $19$ from shard $3$ at round $6$ is selected as the leader during the consensus process. 
Subsequently, shard $3$ commits all historical blocks from block $19$, including the Shift blocks from other shards. 
At this point, all shards will transition to the new DAG (DAG 2) and 
start executing transactions within the new shard.
Moreover, each replica will propose a Shift block every $K^{'} = 6$ blocks in the new DAG (DAG 2) 
to facilitate a transition to the next DAG. 
This measure is intended to protect against censorship attacks, 
which may involve dropping transactions, failing to propose blocks, 
or prioritizing certain transactions over others. 
Additionally, the non-blocking mechanism provides protection against malicious proposers 
as the reconfiguration process requires a minimum of $2f+1$ Shift blocks to be effective.

\para{Uncommitted Transactions}
Due to the two-round leader commitment latency, transactions uncommitted in the ending round of a DAG are discarded and must be resubmitted. Only transactions from the last two rounds or those excluded by the leader are affected. Clients will automatically retransmit transactions lost due to the reconfiguration.

\para{Censorship Attacks}
Malicious replicas may attempt censorship via message drops, transaction rescheduling, DDoS attacks, or proposal halting~\cite{mirkovic2004taxonomy,douligeris2004ddos}. As each replica governs an entire shard in \sysname{}, such attacks can paralyze specific shards. 
The reconfiguration mechanism counters this by periodically reassigning shards to new replicas, 
limiting the impact window of compromised replicas.

 \begin{table}[t]
     \renewcommand\arraystretch{1.3}
     \centering
     \footnotesize
     \begin{tabular}{c|c|c|c|c}
         \hline
         Time & Transactions                & Operations                              & \makecell{Dependencies}           & \makecell{Execution\\ Order} \\
         \hline
             0    & Initial DB                  & \makecell{$D=3$}                                        & \{\}                                                & \{\}                                   \\
         \hline
         1    & $T_1$:(W, $D$, 3)           & $T_1$ writes $D=3$                      & \{$T_1$\}                         & \{\}                    \\
         \hline
         2    & $T_2$:(R, $D$, 3)           & \makecell{$T_2$ reads $D$ on $T_1$:                                                                   \\ $(D=3)$}                               & \{$T_1\rightarrow{}T_2$\}                           & \{\}                      \\
         \hline
         3    & $T_3$:(R, $D$, 3)           & \makecell{$T_3$ reads $D$ on $T_1$:                                                                   \\ $(D=3)$}                               & \{\makecell{$T_1\rightarrow{}T_2$\\$T_1\rightarrow{}T_3$}\} & \{\}                      \\
         \hline
         4    & $T_3$: Commit               & Wait for $T_1$                          & \{\makecell{$T_1\rightarrow{}T_2$                           \\$T_1\rightarrow{}T_3$}\} & \{\}                      \\
         \hline
         5    & $T_1$:(W, $D$, 5)           & \makecell{$T_1$ writes $D=5$.                                                                         \\ \textcolor{red}{Abort $T_2, T_3$}}           & \{$T_1$\}                                       & \{\}                    \\
         \hline
         6    & \makecell{$T_3$:(R, $D$, 5)                                                                                                         \\(\textcolor{blue}{Re-execute})}  & \makecell{$T_3$ reads $D$ on $T_1$:\\ $(D=5)$} & \{$T_1\rightarrow{}T_3$\}                      & \{\}                      \\
         \hline
         7    & $T_1$: Commit               & Commit $T_1$                            & \{$T_1\rightarrow{}T_3$\}         & \{$T_1$\}               \\
         \hline
         8    & $T_3$: Commit               & Commit $T_3$                            & \{$T_1\rightarrow{}T_3$\}         & \{$T_1, T_3$\}          \\
         \hline
         9    & $T_2$: (W, $D$, 3)          & \makecell{\textcolor{red}{Invalid}\\ \textcolor{red}{and re-execute}} &                                   &                         \\
         \hline
         10   & \makecell{$T_2$:(R, $D$, 5)                                                                                                         \\(\textcolor{blue}{Re-execute})}  & \makecell{$T_2$ reads $D$ on $T_1$:\\ $(D=5)$} & \{$T_2$\}                                       & \{$T_1,T_3$\}              \\
         \hline
         11   & $T_2$: (W, $D$, 2)          & $T_2$ writes $D=2$                      & \{\makecell{$T_1\rightarrow{}T_2$                           \\$T_1\rightarrow{}T_3$}\}                                      & \{$T_1,T_3$\}              \\
         \hline
         12   & $T_2$: Commit               & Commit $T_2$                            & \{\makecell{$T_1\rightarrow{}T_2$                           \\$T_1\rightarrow{}T_3$}\}                                           & \{$T_1, T_3, T_2$\}        \\
         \hline
     \end{tabular}
     \caption{An example of generating the dependency while executing 
                 the transactions $\{T_1, T_2, T_3\}$ that access the data $D$ 
                 and determining the execution order.}
     \label{table:protocol_example}
     \vspace{-4mm}
 \end{table}

\vspace{-2mm}
\section{Concurrent Executor}
\label{ss:execution_engine}

The concurrent executor ($CE$) is a critical component enabling \sysname{} to process $\SST{}s$ concurrently during the preplay phase. As a nondeterministic concurrency control executor, $CE$ generates a serialized execution order, read/write sets, and execution results for transaction batches. These outputs allow any replica to independently verify correctness, even though the $CE$-derived order may differ from the transactions’ arrival sequence.

The architecture of $CE$ is illustrated in \Cref{fig:x_protocol_intro},
where a set of executors executes the transactions,
and a concurrency controller ($CC$) determines the execution results among the transactions.

Transactions progress through a two-phase data flow: an execution phase (operations processing) and a finalization phase (commit/abort decisions). Table~\ref{table:protocol_example} exemplifies this workflow using transactions $\{T_1, T_2, T_3\}$.

\begin{figure}[t]
    \centering
    \includegraphics[width=0.36\textwidth]{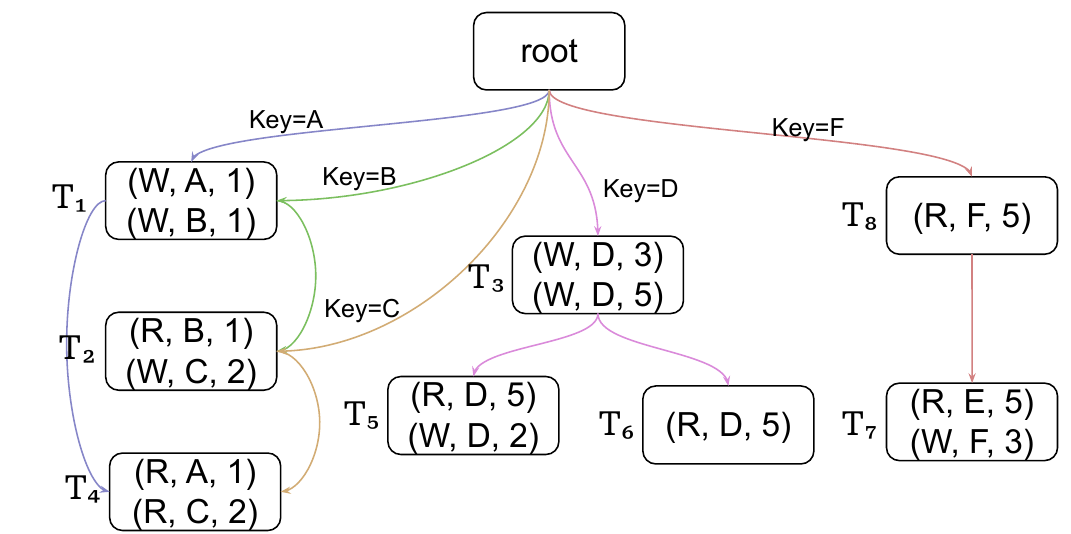}
    \vspace{-4mm}
    \caption{Dependency Graph on $\sysname{}$. Edges with the same color represent a dependency graph with a specific key.}
    \vspace{-6mm}
    \label{fig:dg}
\end{figure}


\vspace{-2mm}
\subsection{Execution Phase}
During the execution phase, the executors access the data within $CC$ directly and
$CC$ maintains a dependency graph to keep track of the relationship between transactions
and all the results are stored in the graph directly to avoid accessing the disk IO.
The critical characteristic of $CC$ is that
$CC$ only maintains the graph based on the current operations among the transactions 
without requiring any read/write set knowledge.
Furthermore, $CC$ mechanism is nondeterministic, 
which means it can arrange transactions in any order based on their current states. 
The order of two transactions will not be established until a dependency is created.
For instance, in the case of two conflicting transactions, denoted as $T_1$ and $T_2$, 
which both only write the same key, either transaction may be ordered first. 
Both execution orders, $[T_1, T_2]$ and $[T_2, T_1]$, are considered valid. 
A dependency is established based on the commit times of these transactions. 
If $T_1$ commits before $T_2$, the execution order $[T_1, T_2]$ indicates that $T_1$ is prioritized over $T_2$. 
Additionally, a third transaction, $T_3$, can influence this dependency. 
A dependency is formed if $T_3$ reads the value following $T_2$’s write before $T_1$ and $T_2$ commit, 
resulting in a unique execution order of $[T_1, T_2, T_3]$. 
As a result, a dependency indicator will be generated for the two conflicting transactions 
once a read-write conflict arises or when both transactions have been committed.

While receiving an operation from a transaction $T$ sent by the executors identified accessing the key $K$,
denoted as $O_k$, $CC$ checks the relationships among the transactions.
If $T$ conflicts with other transactions or has been aborted by other transactions,
the operation $O_k$ will not be considered valid.
For instance, $T_2$ at time 9 in Table~\ref{table:protocol_example} is an example of an invalid operation,
as it was aborted by $T_1$ at time 5 due to its outdated read in $D$.
Transaction $T$ will be aborted in such cases and require reexecution.
Otherwise, if the operation $O_k$ is valid,
it will be added to the dependency graph (\cref{ss:def_g}) and obtain the operation result,
such as the value $V$ that $O_k$ intends to read.
We can obtain the operation results from other transactions directly based on the dependency graph
to allow reading uncommitted data, such as $T_2$ reads $D$ on $T_1$ at time 2.

\vspace{-2mm}
\subsection{Finalization Phase}
\label{ss:ee_commit_phase}
During the finalization phase, the executor informs $CC$
that the executor has completed all the operations.
Then $CC$ will update the results to the storage asynchronously once all its dependencies have been committed and assign the execution order to the transactions.
If $CC$ has terminated the transaction due to conflicts with other transactions,
$CC$ aborts the transaction and notifies the executor to restart the execution.

\vspace{-2mm}
\section{Dependency Graph in CC}
\label{ss:dep}
This section explains the dependency graph $G$, 
which is central to the $CC$ component. 
It plays a crucial role in maintaining the causal relationships between transactions 
during the preplay process in $CE$. 
The $CC$ component ensures that the sequential order of execution defined by $G$ is valid.

\vspace{-2mm}
\subsection{Dependency Graph Construction}
\label{ss:def_g}
A Dependency Graph is a graph $G(V, E)$
that plays a crucial role in tracking the causal relationship between transactions in $CC$.
Each node $v\in V$ represents a specific transaction.
Additionally, each edge $e (u, v, k) \in E$ indicates a connection between two transactions $u$ and $v$
on a key $K$.
This relationship is represented as $u \rightarrow{}_k v$.
For example, in \Cref{fig:dg},
transaction $T_5$ generates an edge $e (T_3,T_5, D)$
from $T_3$ because $T_5$ acquires the value 3 of key $D$ from $T_3$.
Without loss of generality,
we have assigned a root node denoted $R$ and added edges $e (R, u, k)\in E$
for each $u \in V$ that accesses the key $K$ but does not have any incoming edge on key $K$, 
such as $T_7$ and $T_8$ in \Cref{fig:dg}. 

If the graph $G$ is acyclic, a sequential order can be established
by generating a topological order.
It is crucial that every transaction
must obtain the same causal order in any topological order from $G$ to ensure consistency.
Therefore, $G$ is considered a valid graph only if any sequential order generated
from the topological order is a valid serialization order and produces the same results.
By following any correct order, all transactions will yield the same execution results.
However, because of the nondeterministic characteristic, 
the results may not be the same as those executed in their arrival order.
Each node $u$ maintains all records of the operations triggered by a transaction $u$,
including the resulting values.
The sequence of the linking nodes establishes the order of commitment 
between the two transactions. 

\begin{figure}[t]
    \centering\begin{tabular}{ccc}
        \includegraphics[width=0.115\textwidth]{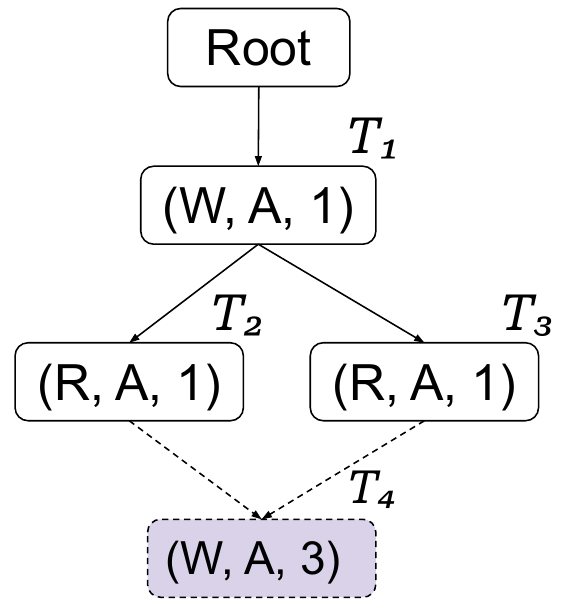} &
        \includegraphics[width=0.17\textwidth]{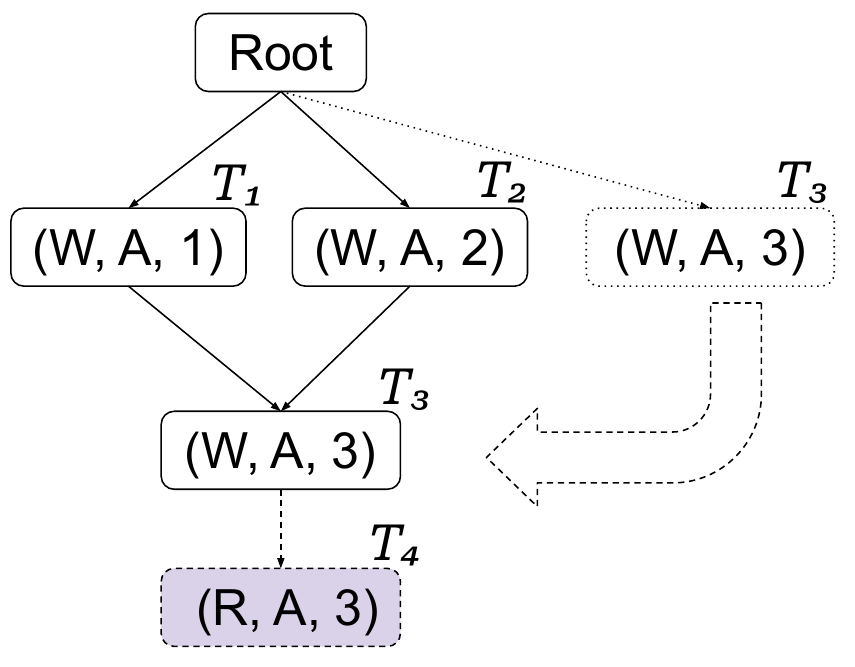} &
        \includegraphics[width=0.15\textwidth]{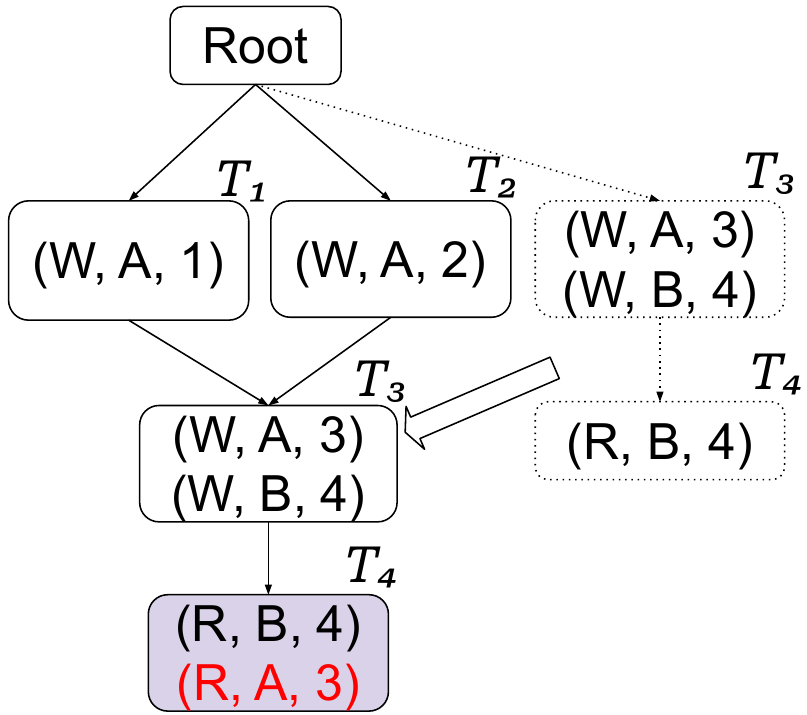} \\
        a) & b) & c) \\
    \end{tabular}
    \vspace{-4mm}
    \caption{
        An example of incorporating operations into the dependency graph and 
        adjusting the graph to ensure correctness.
        a) $T_4$ writes A=3.  b) $T_4$ reads A.
        c) $T_4$ reads key $A$ on its existing node and retrieves the value from $T_3$.}
    \vspace{-6mm}
    \label{fig:add_new}
\end{figure}


Since a transaction is an atomic commitment, 
we combine the internal operations to simplify the in-node states.
However, to trace the conflicts between two nodes,
we must retain the first operation if it is a read and the last operation if it is a write,
to ensure that the causal relationship is not lost. 
Thus, we remain at most two operations in the nodes: the first read and the last write.

To help illustrate the algorithm,
we define the types of each node depending on the operations it contains on a key:
\begin{itemize}
    \label{ss:rw_node}
    \item A node $v \in V$ is a read node $R^{k}_v$ if the first operation on key $K$ is a read.
    \item A node $v \in V$ is a write node $W^{k}_v$ if $v$ contains write operations on key $K$.
    \item The root node $R$ is a write node.
\end{itemize}


\vspace{-2mm}
\subsection{Generating New Nodes} \label{ss:add_op}
This section presents the process of adding operations
from a new transaction to the dependency graph $G$.

$CC$ creates a new node whenever an operation $O_k$ is received from a new transaction $T$.
If $O_k$ is a write operation, $T$ needs to establish a connection to each casual relation.
To avoid pointing to the root and assuming that the earlier transaction will commit first,
the non-write nodes $v$ on key $K$, which only contains reads,
without any outgoing edges (not dependent by other nodes) are selected,
and edges $e(v,u,k)$ are added, pointing to $u$ (\Cref{fig:add_new} (a)).

On the other hand, when a read operation is performed, 
$CC$ selects the latest write node $u$ to obtain the latest value
or selects the root to read the data value from storage if no write nodes exist. 
If the write node $u$ is selected, 
we need to make all other write nodes $W^k_v$
contain a path to $u$ to guarantee the correctness for the read after write between $u$ and $v$.
Finally, the operation and its result <Type, Key, Result> will be written into the node $u$.
An example of adding a new read operation on $A$ from $T_4$ is depicted in \Cref{fig:add_new} (b).
$T_4$ selects $T_3$ to read to obtain $A=3$ and adds an edge from $T_3$
and a record <R, $A$, 3> is logged down in the node.
Then $T_3$ will add two edges from $T_1$ and $T_2$, respectively.


\begin{figure}[t]
    \centering\begin{tabular}{ccc}
        \includegraphics[width=0.12\textwidth]{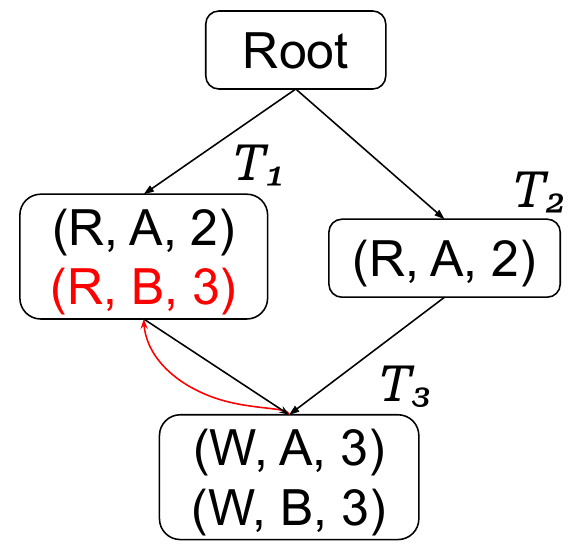} &
        \includegraphics[width=0.10\textwidth]{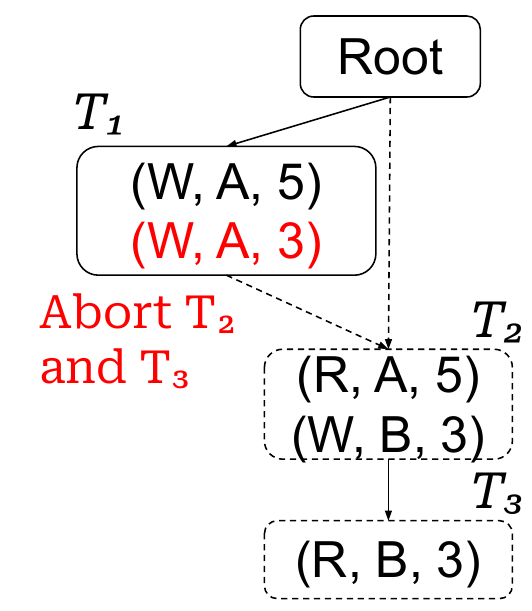} \\
        a) & b) 
    \end{tabular}
    \vspace{-4mm}
    \caption{Cycle of conflicts and cascading aborts.
        a) $T_1$ reads $B$ and adds a dependency from $T_3$ following the rules in \Cref{ss:add_op}.
        b) Cascading aborts from $T_1$ since $T_1$ wants to write $A$  that breaks the read on $T_2$.}
    \vspace{-5mm}
    \label{fig:cycle}
\end{figure}

\vspace{-2mm}
\subsection{Operations on Existing Nodes}
When receiving an operation $O_k$ for the key $K$ from an existing transaction $T$ in $G$,
$CC$ will select the corresponding node $u$ to attach the record.
If $O_k$ is a read operation,
the result will be directly retrieved if $u$ contains the record for key $K$.
Otherwise, it will proceed with the new node operation as specified in \Cref{ss:add_op} 
to choose a previous one to access the value.
\Cref{fig:add_new} (c) illustrates an instance where $T_4$ reads key $A$
as its second operation and retrieves the value from $T_3$.
If $O_k$ is a write operation, the operation will be appended to the node.
\vspace{-2mm}
\subsection{Conflict Detection}
\label{ss:data_conflict}
Appending the records to an existing node may lead to transaction conflicts.
For instance, a transaction updates the value again but it has been read by another transaction or a dependency cycle is created due to the dependency on another key 
since we always find the latest write to retrieve the value.
\Cref{fig:cycle} (a) depicts a scenario in which $T_1$ attempts to retrieve
the value of $B$ from $T_3$, which has established a dependency from $T_1$ due to key $A$, which results in the creation of a dependency cycle.
In this case, $CC$ will try to read the value from its ancestor, 
like $B$ reads the value from the $Root$ \Cref{fig:cycle} (a).
If there is still any conflict with other transactions, $CC$ will trigger the abort process. 

Once conflicts are detected, $CC$ triggers an abort process as follows:
\begin{enumerate}
    \item If $u$ only contains read operations, abort $T$ itself.
    \item If $u$ contains write operations, cascading abort from $T$.
\end{enumerate}


In \Cref{fig:cycle} (b), we need to abort $T_2$ and $T_3$ since $T_1$ contains a write operation.
However, in \Cref{fig:cycle} (a), we only need to remove $T_1$ and keep $T_3$ alive.

%
%

\vspace{-2mm}
\section{\sysname{} Correctness Analysis}
\label{ss:correctness}

In this section, we conduct an analysis of the safety and liveness properties of \sysname{}. 
Safety is defined as if two conflicting transactions, $T_1$ and $T_2$, 
are executed within an honest replica in a specific order $O$ (for instance, $T_1 < T_2$), 
it is expected that all other honest replicas will also execute $T_1$ and $T_2$ in the same order $O$.
Liveness, on the other hand, is characterized by the assurance that client requests will consistently 
receive a response.

\para{Proof of Safety}
We prove the safety of \sysname{}. We will first analyze the safety \textbf{within the same DAG}.

\begin{theorem}
    \label{th:safety1}
In the case of two conflicting transactions, $T_1$ and $T_2$, 
which may occur as either $\SST{}s$ within the same shard 
or as $\CST{}s$, if $T_1$ is executed prior to 
$T_2$ in an honest replica, 
any other honest replica will also execute $T_1$ before $T_2$.
\end{theorem}

\begin{proof}
The DAG protocol ensures that the sequence of two blocks proposed by the same proposer
aligns with their commit order. 
Consequently, the order of two $\SST{}s$ 
proposed by the same proposer will be preserved. 
Furthermore, the order of $\CST{}s$ is established 
by the consensus, resulting in a consistent global order. 
This leads to uniform execution of transactions $T_1$ and $T_2$
across all honest replicas, 
irrespective of whether the transactions are $\SST{}s$ or $\CST{}s$.
\end{proof}
\begin{theorem}
    In the scenario where two conflicting transactions 
    occur,where $T_1$ is a $\SST{}$ and $T_2$ as a $\CST{}$,
    if an honest replica executing $T_1$ prior to $T_2$, 
    any other honest replicas will execute $T_1$ before $T_2$ as well.
\end{theorem}
\begin{proof}
    Consider a contradiction scenario in which replica $R_1$ executes $T_1$ prior to $T_2$, 
    while replica $R_2$ executes $T_2$ before $T_1$. 
    We also suppose leader $X$ commits $T_1$ and leader $Y$ commits $T_2$. 
    Then it can be shown that $R_1$ is the proposer of $T_1$ 
    and $T_1$ is not converted into a cross-shard transaction.
    Since $R_2$ executes $T_2$ ahead of $T_1$, 
    the round of leader $Y$ must be less or equal to the round of leader $X$ ($ Y\leq X$).
    However, if the round $r$ to $T_1$ occurs before $Y$ ($r < Y$), 
    then $T_2$ cannot be committed (\ref{enum:leader_commit}). 
    Conversely, if $r \geq Y$, then according to \ref{enum:ss_cs_2}, $T_1$ should be converted into a $\CST{}$
    if $T_2$ remains uncommitted.
    Thus, it is impossible for $R_1$ to execute $T_1$ before $T_2$ while $R_2$ executes $T_2$ before $T_1$.
\end{proof}
\begin{theorem}
    When two conflicting transactions, $T_1$ and $T_2$, 
    are executed by honest replicas in the order of $T_1 < T_2$ 
    across two different DAGs, it is ensured that all other honest replicas will 
    execute $T_1$ and $T_2$ in the same order within a single DAG.
\end{theorem}
\begin{proof}
    \Cref{ss:non-blocking} illustrates that all honest replicas will transition to the new DAG starting 
    from the same ending round. Refer to the consistency and completeness property of the DAG (\Cref{ss:dag_propoties}),  all honest replicas will execute transactions $T_1$ and $T_2$ within the same DAG.
\end{proof}
Consequently, we can draw the conclusion that: If two transactions are isolated, they can be executed in any order.
Otherwise, all the honest replicas will execute them in the same order. Thus, \sysname{} holds its safety guarantees.
\para{Proof of Liveness}
In an environment where all replicas operate effectively, 
they will propose blocks within the same DAG. 
Each block suggested by the respective shard proposers will ultimately be committed. 
When a malicious replica is identified, the honest replicas will respond by proposing a Shift block.
If fewer than $2f + 1$ Shift blocks are proposed, 
the DAG will remain unchanged, 
and all replicas will continue functioning within the current DAG while proposing new blocks. 
Conversely, if there are $2f + 1$ Shift blocks, all honest replicas will transition to the new DAG 
within the same round, as detailed in \Cref{ss:dag_propoties}.
After a minimum of $2f + 1$ honest replicas successfully relocate to the new DAG, 
they will be empowered to propose new blocks. Provided that all replicas maintain proper behavior, 
they will consistently propose blocks within the same DAG, 
ensuring that each shard proposer's proposed blocks are duly committed.
\vspace{-2mm}
\section{Correctness of serializability on $CC$} 
\label{sec:cc_correctness}
In this section, we will present the proofs that establish the correctness of $CC$. 
We will start by defining key concepts related to serializability, 
followed by a detailed analysis of our findings regarding correctness.

We consider a set of transactions $T = \{ T_i \}$ (where \( 1 \leq i \leq n \)) 
alongside a predefined commit order $CO$. 
$CC$ generates a sequential order denoted as $SO = [T_1, T_2, \dots, T_n]$, 
producing an outcome represented as $OUT = [OUT_1, OUT_2, \ldots, OUT_n]$ that each transaction reads and writes some values on some keys. 
Let $SE$ denote one of the possible sequential orders of $SO$, with $OUT'$ reflecting its corresponding outcomes. 
$CC$ is deemed serializable if the condition $OUT = OUT'$ holds true.
\begin{definition}[Read-Complete]
    If $T_i$ reads a value updated by $T_j$ in $SO$, $T_i$ will also read the value updated by $T_j$ in $SE$.
    If transaction $T_i$ reads a value updated by a transaction $T_j$ in $SO$, 
    then $T_i$ will also read the value updated by $T_j$ in $SE$.
\end{definition}

\begin{definition}[Write-Complete]
    If transactions $T_i$ and $T_j$ both write new values to key $K$, 
    but $T_i$ commits before $T_j$, which generates an order $T_i < T_j$ in $SO$, 
    then $T_i$ will also write the values to $K$ before Tj when in $SE$.
\end{definition}

\begin{theorem}
 $CC$ is considered both Read-Complete and Write-Complete when the dependency graph $G$ is always valid.
\end{theorem}

\begin{proof}
    Firstly, if $G$ is valid, we know that when a read node $R_v^{k}$ (representing a node u that reads on key K) 
    retrieves a value from a write node $W_u^{k}$ on key $K$, all the write nodes that are updating values on $K$ 
    either have a direct path to $u$ or a path from $v$. This guarantees the correctness of read-after-write operations.
    Consequently, if transaction $T_i$ reads values updated by transaction $T_j$ on key $K$, 
    where $T_j$ always exists (with the root being a write node), 
    it follows that no other transactions updating values on $K$ will be located between $T_i$ and $T_j$. 
    Therefore, for any execution order $SO$ generated by $G$, where $T_i$ reads the value updated by $T_j$, 
    it can be inferred that $T_j$ will also read the same value from $T_i$ in $SE$. 
    This demonstrates that $CC$ is Read-Complete.
    Secondly, since $SO$ is an execution order in $CC$, if $T_i$ commits before $T_j$, then $T_i$ will precede $T_j$ in $SO$. 
    Consequently, $T_j$ will update the values after $T_i$ in $SE$ as well. 
    Therefore, we can conclude that $CC$ is Write-Complete.
\end{proof}

\begin{theorem}
    $CC$ is serializable if $CC$ is both Read-Complete and Write-Complete.
\end{theorem}

\begin{proof}
    Since $CC$ is Read-Complete, every transaction $T_i$ that reads the values of certain keys will retrieve the same values 
    regardless of the order generated by $G$. 
    Additionally, since $CC$ is Write-Complete, every transaction $T_i$ that updates the values of certain keys will maintain a consistent order 
    among all the orders generated by $G$. 
    Therefore, the transactions will yield the same outcomes in any generated orders $SO$ and $SE$ ($OUT = OUT'$).
\end{proof}

\vspace{-2mm}
\section{Concurrency Executor Evaluation}
\label{s:eval}
This section evaluates \sysname{} by assessing its performance on the $CE$
and the \sysname{} protocol. 
We implement all the baseline comparisons using Apache ResilientDB (Incubating)~\cite{apache-resdb, rcc}. 
Apache ResilientDB is an open-source incubating blockchain project 
that supports various consensus protocols. 
It provides a fair comparison of each protocol by offering a high-performance framework. 
Researchers can focus solely on their protocols without considering system structures such as the network and thread models.

We will begin by comparing $CE$ with two baseline protocols:
$OCC$~\cite{kung1981optimistic} and \TPL{}~\cite{soisalon1995partial}. 
Additionally, we will analyze the performance of \sysname{}, 
which is built on Tusk~\cite{narwhat-tusk}, 
and we will also use Tusk as a baseline for our comparisons. 
For our input workload, we will utilize SmallBank~\cite{smallbank, alomari2008cost}, 
a benchmarking suite that simulates common asset transfer transactions. 
This suite is also used to evaluate variant block systems
~\cite{thakkar2020scaling,li2023auto, lai2023private, zhang2023cchain, ruan2020transactional, peng2022neuchain, gorenflo2020fastfabric}.

\vspace{-2mm}
\subsection{Baseline Protocols}
We implement \OCC{}~\cite{kung1981optimistic} and \TPL{}~\cite{yu2014staring, yu2016tictoc} to compare the performance against our concurrent executor.
We set up our experiments on AWS c5.9xlarge consisting of 36 vCPU, 72GB of DDR3 memory.
We use LevelDB as the storage to save the balance of each account.

\para{\OCC{}}
Each executor is responsible for locally executing transactions.
When an operation within a transaction $T$ requires reading the value of a key $K$
that the executor has not previously accessed during the execution,
the executor will retrieve the value from the storage.
Each value also contains a version to indicate the time the value was obtained.
Any write operation will update the values locally.
Upon completion of $T$, all the updated values will be forwarded to a central verifier.
The verifier will cross-check the value versions by comparing them with the current versions in the storage.
If there is a mismatch, the commit will be rejected, necessitating the re-execution of $T$.

\para{\TPL{}}
Each executor performs transactions by directly accessing the storage through a central controller.
When an operation within a transaction $T$ requires the read or update of a key $K$,
the controller will lock $K$ to prevent conflicts.
If an operation seeks to access $K$ but discovers that another executor has locked it,
the executor will release all locks and re-execute $T$.
Upon the completion of $T$, all the results will be transmitted to storage,
and all locks will be released.

\vspace{-2mm}
\subsection{Experiment Setup With Smallbank}
SmallBank~\cite{smallbank} is a transactional system that comprises six distinct transaction types,
five of which are designed to update account balances,
while the remaining transaction is a read-only query that retrieves both checking
and saving the account details of a user.
Our focus is on two types of transactions: SendPayment and GetBalance,
which are used to transfer funds between two accounts and retrieve account balances, respectively.
Our objective is to evaluate the performance under varying read-write balance workloads.
During a SendPayment transaction, the account balances are updated by reading the current balance and then writing the new values back.
We created 10,000 accounts and conducted each experiment 50 times to obtain the average outputs.

We evaluated the impact of parallel execution.
We measured the performance by uniformly selecting GetBalance with a probability
of $P_r$ while SendPayment with $1 - P_r$.
We follow a Zipfian distribution to select accounts as transaction parameters and set the Zipfian parameter $\theta{}$.
The value of $\theta{}$ determines the level of account contention,
with higher values leading to higher contention.
We focus only on data workloads with high contention by setting $\theta{} = 0.85$.

\begin{figure*}[t]
    \centering
    \scalebox{0.82}{\ref{workermain}}\\[3pt]
    \begin{tabular}{ccc}
        \EvalEEWorkerTPS &
        \EvalEEWorkerLat &
        \EvalEEWorkerRetry\\
        \multicolumn{3}{c}{ a) The read-write balanced workflow ($P_r=0.5$).}\\
        \EvalEEWorkerUpdateTPS &
        \EvalEEWorkerUpdateLat &
        \EvalEEWorkerUpdateRetry \\
        \multicolumn{3}{c}{ b) The update only workload ($P_r=0$).}\\
    \end{tabular}
    \vspace{-2mm}
    \caption{Evaluation of $CE$ on different numbers of executors.} 
    \vspace{-4mm}
    \label{fig:worker-eval}
\end{figure*}

\vspace{-2mm}
\subsection{Impact from Concurrent Executor}
We first evaluate the impact of increasing the number of executors to execute the transactions,
then measure the aborts produced by each protocol.
We ran two batch sizes $b300$ and $b500$ for each protocol:
\sysname{}-b300, \sysname{}-b500, \OCC-b300, \OCC-b500,
\TPL{}-b300, and \TPL{}-b500.
We set $P_r = 0.5$ to measure a read-write balanced workflow and $P_r = 0$ on an update-only workflow.

\para{Number of Executors}
In the read-write balanced workflow,
the results depicted in \Cref{fig:worker-eval} (a) show that 
\TPL{} protocols with different batch sizes
all experience a performance drop when increasing the number of executors beyond $8$.
However, \sysname and \OCC{} protocols with all the batch sizes
obtain their highest throughputs on $12$ executors and maintain stable throughput.
\sysname{}-b500 obtained $43K$ TPS while \OCC{}-b500 achieved $35K$ TPS.

In the update-only workflow,
the results shown in \Cref{fig:worker-eval} (b) indicate that
\OCC{} and \TPL{} stopped increasing earlier in $4$ executors (both around $22K$ TPS)
while \sysname provides a peek throughput ($28K$ TPS) in $12$ executors.

These experiments demonstrate that
all the protocols do not obtain significant benefits for a large number of executors
in a high-competition workflow.
However, \sysname can still achieve more parallelism with more executors.

\para{Evaluation of Abort Rates}
As we increased the number of executors,
we also measured the average number of re-executions for the transactions.
The results in \Cref{fig:worker-eval} 
indicate that when the number of executors goes beyond $8$,
all \TPL{} protocols experience a significant increase in the rate of abortions,
leading to a drop in throughput from $24k$ to $18k$ in the read-write balanced workflow.
While \OCC{} protocols provide a lower rate within the read-write balanced workflow.
However, \sysname achieves the lowest abortions,
with \sysname-b500 reducing $50\%$ of the abortions from \OCC{}-b500
and $90\%$ from \TPL{}-b500 in all the experiments.

\para{Evaluation of $\theta{}$}
We will now provide an evaluation with various $\theta$ values, 
using a read-write balance workload of $P_r = 0.5$. 
In particular, we examine high contention workloads where $0.75 \leq \theta \leq 0.9$, 
as this range represents the primary focus of our study. 
The results illustrated in \Cref{fig:pr_eval} (a and b) 
demonstrate that at $\theta = 0.75$, 
both \OCC{} and \sysname{} show comparable performance. 
However, as $\theta$ increases to 0.9, \OCC{} experiences a significant decline in performance, 
while \sysname{} continues to achieve higher throughput levels. 
In contrast, the \TPL{} approaches steady throughput, attributable to its locking strategy.

\begin{figure}[t]
    \centering
    \scalebox{0.82}{\ref{alphamain}}\\[3pt]
    \begin{tabular}{cc}
        \EvalAlphaTPS &
        \EvalAlphaLat \\
        \EvalWorkloadTPS &
        \EvalWorkloadLat
    \end{tabular}
    \vspace{-4mm}
    \caption{Throughput and average latency within varying $\theta{}$ settings with $P_r=0.5$ (a and b)
    and varying $P_r$ settings with $\theta{}=0.85$ (c and d).}
    \vspace{-4mm}
    \label{fig:pr_eval}
\end{figure}

\para{Evaluation of $P_r$}
In this evaluation, we will look at how different values of $P_r$ affect the read and write ratio 
in the workload, with $\theta{} = 0.85$. 
The results in \Cref{fig:pr_eval}(c and d) show that when $P_r = 1$ (all read), 
all protocols behave similarly. 
However, the \OCC{} protocol performs slightly better because it allows non-blocking local executions.
When conflicts happen, the \TPL{} shows a sharp decline in performance. 
In contrast, both \sysname{} and \OCC{} also perform worse when we decrease $P_r$, 
which leads to more conflicts. 
On the other hand, as we increase the value of $P_r$, 
\sysname{} achieves better throughput and lower latency than \OCC{}, 
even when all operations are write-only.
All protocols show similar latency at $P_r = 1$, but as $P_r$ increases, 
the latency for \TPL{} rises, while \sysname{} continues to perform better than OCC.



\vspace{-2mm}
\section{System Evaluation}
We conducted evaluations to determine the impact of \sysname built on Tusk.
In our evaluation, we compared the performance of \sysname with Tusk,
which executes transactions in order after reaching a total order after DAG protocols.
We evaluated the impact of different components of the protocol by comparing the results between three systems:
\begin{itemize}
    \item \sysname{}: Leverage $CE$ + parallel verification.
    \item \sysname{}-OCC: Combine $OCC$ + parallel verification.
    \item Tusk: Utilize the $OE$ model.
\end{itemize}
As the serialized verification will execute the transactions in order to verify the results,
we will expect that any DAG-based protocols with serialized verification will provide the same behavior with Tusk.
We also leveraged SmallBank as the input workload.

Each replica was configured with $CE$ comprising $16$ executors to process transaction batches of $500$, alongside $16$ validators to verify blocks post-consensus. We scaled the system from $8$ to $64$ replicas. By default, $K^{'}$ was set to a sufficiently large value to disable shard rotation. In the final phase of our evaluation, we examined the impact of varying $K^{'}$ values, which govern the frequency of shard reconfiguration (\Cref{ss:rotation}).

\para{SmallBank}
Throughout the system evaluation, we focus on the SmallBank workload with a read-write balanced scenario ($P_r=0.5$), 
where half of the transactions are read-only.

In the smallbank workload, Transaction addresses were selected from a pool of $1000$ users with a skew parameter $\theta{} = 0.85$ to simulate a high-contention workload.

The LAN results, as shown in \Cref{fig:framework_smallbank}, demonstrate that \sysname{} significantly outperformed Tusk’s sequential execution model, achieving a $50x$ speedup. Specifically, \sysname{} reached $500K$ TPS with $64$ replicas, compared to Tusk’s $11K$ TPS. This improvement highlights the benefits of executing transactions in parallel.

We also compared Thunderbolt with \sysname{}-OCC, which replaces the $CE$ with $OCC$. While \sysname{}-OCC matched \sysname{}’s throughput at $8$ replicas, it lagged behind at scale, achieving only $400K$ TPS with $64$ replicas. Furthermore, \sysname{} maintained a low transaction latency of $5$ seconds, whereas Tusk’s latency soared to $100$ seconds under the same conditions.

The WAN results demonstrate similar behavior but with higher latency. 
The latency gap between \sysname{} and Tusk become smaller because the WAN latency begins to dominate cost, becoming bottleneck.

\begin{figure}[t]
    \centering
    \scalebox{0.82}{\ref{frameworkmain}}\\[3pt]
    \begin{tabular}{cc}
        \EvalFrameworkTPS &
        \EvalFrameworkLat \\
        \EvalWANTPS & 
        \EvalWANLat
    \end{tabular}
    \vspace{-4mm}
    \caption{Throughput and average latency within different replicas within $P_r=0.5$ in LAN and WAN.}
    \vspace{-4mm}
    \label{fig:framework_smallbank}
\end{figure}

\para{Cross-shard Transactions}
Next, we evaluated the impact of $\CST{}s$ using $16$ replicas. We randomly assigned a percentage $P\%$ ($0 < P \leq 100$) of transactions to be processed by two shards. Additionally, we assessed the benefits of parallel execution by comparing \sysname{}-OCC.

As shown in \Cref{fig:framework_cross}, the performance of both \sysname{} and \sysname{}-OCC declined as the percentage $P$ increased. In scenarios with only single-shard transactions ($P = 0$), both systems achieved $100K$ TPS. However, when $P$ increased to $8\%$, \sysname{}-OCC’s throughput dropped to $16K$ TPS, while \sysname{} maintained a significantly higher throughput of $64K$ TPS. \sysname{}-OCC’s performance aligned closely with Tusk, achieving approximately $10K$ TPS. In contrast, \sysname{} delivered $19K$ TPS even when all transactions were cross-shard, demonstrating the advantages of its parallel execution model and non-deterministic ordering on $CE$.

Latency metrics further highlighted \sysname{}’s superiority. Under high-contention workloads, \sysname{} achieved a transaction latency of $24s$ seconds, while \sysname{}-OCC’s latency was nearly double at $50$ seconds.

\begin{figure}[t]
    \centering
    \scalebox{0.82}{\ref{crossmain}}\\[3pt]
    \begin{tabular}{cc}
        \EvalCrossTPS &
        \EvalCrossLat
    \end{tabular}
    \vspace{-2mm}
    \caption{Throughput and average latency within different ratios of cross-shard transactions within 16 replicas.}
    \vspace{-4mm}
    \label{fig:framework_cross}
\end{figure}

\begin{figure}[t]
    \centering
    \begin{tabular}{cc}
        \EvalRotateTPS &
        \EvalRotateLat
    \end{tabular}
    \vspace{-2mm}
    \caption{Throughput and average latency within different reconfiguration periods within 8 replicas.  }
    \vspace{-4mm}
    \label{fig:framework_rotate}
\end{figure}

\para{Reconfiguration Periods}
\label{ss:rotation}
Now, we analyze the performance using different reconfiguration periods $K^{'}$
to transition the shard proposers into a new DAG on $8$ replicas.
Figure~\ref{fig:framework_rotate} demonstrates that \sysname exhibited
lower performance with smaller $K^{'}$ values ($80K$ TPS with $K^{'}=10$),
attributed to the costly transition between DAGs.
Conversely, when $K^{'}$ was increased to over $1000$,
\sysname demonstrated significantly improved stability,
achieving a throughput of $180K$ TPS.
Additionally, the average latency decreased from $1.9s$ to $1.7s$ as $K^{'}$ increased from $10$ to $5000$.
\Cref{fig:shard_rec} also shows the average run time of committing proposals per 100 rounds, that is $\frac{1}{100}\sum(T_{commit(i)} - T_{commit(i-1)}$) where $T_{commit(i)}$ is the time of committing round $i$.
We set $K^{'}$ as 300 and it was demonstrated that \sysname will not get stuck during the reconfiguration. The runtime of each round is around $0.07s$ to $0.1s$.

\para{Failures}
Finally, we evaluated the impact of replica failures
within $16$ replicas. 
We forced $f$ replicas ($f=1$ or $f=2$) to stop working during the experiments.
We randomly designated a percentage $P\%$ ($0 < P \leq 100$) 
of the transactions to be processed by two shards. 

\Cref{fig:framework_fail} reveals that \sysname still can provide 
higher throughputs when some shards stop working.
When one replica failed to propose transactions, \sysname\\1 ($f=1$) obtained $78K$ TPS
working on all $\SST{}s$ ($P=0$) and $17K$ TPS on all $\CST{}s$
while when two replicas failed to propose transactions, \sysname\\2 ($f=2$) obtained $66K$ TPS 
on all $\SST{}s$ and $15K$ TPS on all $\CST{}s$.
However, the results show that the latency remains stable even if some replicas fail to process, benefiting from the leader rotation of the DAG protocols.

\begin{figure}[t]
    \centering
    \begin{tabular}{cc}
        \ReconfigLat 
    \end{tabular}
    \vspace{-4mm}
    \caption {Average latency per 100 rounds and reconfigure the shard per 300 rounds. }
    \vspace{-4mm}
    \label{fig:shard_rec}
\end{figure}

%
%
%

\vspace{-2mm}
\section{Related work}
\label{s:related}

\para{Sharding on DAG-based protocol}
Several studies 
~\cite{dang2019towards, wang2019sok,kokoris2018omniledger,zamani2018rapidchain,al2017chainspace, ext_byshard, geobft, ringbft} 
have underscored the necessity of implementing sharding to enhance scalability within blockchain systems. Directed Acyclic Graph (DAG)-based blockchains 
~\cite{churyumov2016byteball, baird2016swirlds, li2020decentralized, yu2020ohie, xu2021occam,
narwhat-tusk, spiegelman2022bullshark,keidar2021all,keidar2022cordial,
stathakopoulou2023bbca,mysticeti,arun2024shoal,
spiegelman2023shoal,shrestha2024sailfish,
malkhi2023bbca, 
raikwar2024sok,
spiegelman2022bullshark, narwhat-tusk,keidar2021all, xiao2022nezha} 
present a promising alternative for improving concurrent transaction processing through the utilization of the DAG data structure. 
However, it is noteworthy that only a limited number of approaches offer effective sharding strategies for DAG protocols 
~\cite{dang2019towards, huang2022brokerchain, cheng2024shardag, amiri2021sharper, qi2022dag} 
These strategies often depend on account-based mechanisms to achieve eventual atomicity, 
employing incentive mechanisms through either the two-phase commit (2PC) protocol or additional coordinating entities.

\begin{figure}[t]
    \centering
    \scalebox{0.82}{\ref{failmain}}\\[3pt]
    \begin{tabular}{cc}
        \EvalFailTPS &
        \EvalFailLat
    \end{tabular}
    \vspace{-2mm}
    \caption{Throughput and average latency within different ratios of cross-shard transactions 
            within 16 replicas when $f$ ($f=1$ or $f=2$) replicas failed.}
    \vspace{-4mm}
    \label{fig:framework_fail}
\end{figure}

\para{Execute-Order-Validate}
Hyperledger Fabric~\cite{hyperledger-fabric} introduces the Execute-Order-Validate (EOV) framework, 
which demonstrates high performance primarily in low-contention workloads. 
To enhance this framework, various techniques have been developed
 ~\cite{thakkar2018performance, gorenflo2020fastfabric, thakkar2021scaling,
sharma2019blurring, gorenflo2020xox, ruan2020transactional},
including methods for reordering transactions within a block and improving execution processes. 
In contrast, \sysname{} addresses scalability by distributing transactions across multiple shards, 
thereby enabling each shard proposer to execute transactions in parallel.

\para{Concurrent Execution}
Deterministic approaches 
~\cite{qadah2018quecc, faleiro2017high, yao2016exploiting} 
have been proposed to improve the efficiency of transaction execution by creating a dependency graph. 
This graph allows for concurrent execution while minimizing conflicts between transactions. 
Furthermore, segmenting transactions 
~\cite{shasha1995transaction,shasha1992simple,chen2021schain,qin2021caracal} 
has been introduced as an effective method for reducing these conflicts.
CHIRON~\cite{neiheiser2024chiron} and BlockSTM ~\cite{gelashvili2023block}
offer nondeterministic execution by extracting dependencies from smart contracts or by determining the execution order based on transaction arrival times. 
Conversely, \sysname{} adopts a different methodology that it does not depend on arrival times or read/write set information for transaction processing. 
Instead, it dynamically assigns execution orders to effectively minimize transaction conflicts.

%

\para{Concurrent Consensus}
A robust and scalable blockchain framework is crucial for the implementation of real-world applications~\cite{bbook, bedrock, chemistry, ccbook}. 
For example, the PoE~\cite{poe} model utilizes speculative execution protocols, 
and protocols, including RCC~\cite{rcc}, FlexiTrust~\cite{sgxsuyash}, and SpotLess~\cite{spotless}, 
incorporate multiple leaders to enhance parallel processing capabilities. 
However, these protocols currently do not support reconfiguration.
\vspace{-2mm}
\section{Conclusions}
We introduce \sysname{}, an innovative sharding system designed to enhance smart contract execution by 
integrating the Order-Execute and Execute-Order-Validate models
to efficiently handle both single-shard transactions ($\SST{}s$) and cross-shard transactions ($\CST{}s$).
We have developed a concurrent executor that significantly improves the performance of $\SST{}s$ 
without requiring prior knowledge of read/write sets. 
Additionally, \sysname{} effectively distributes transactions across multiple shards and utilizes 
Directed Acyclic Graph (DAG)-based protocols to maintain consistency between single-shard and cross-shard transactions.
A key feature of \sysname{} is its ability to utilize the inherent properties of the DAG to facilitate a non-blocking transition 
to a new DAG structure. This allows for the rotation of proposers for each shard, helping to prevent malicious activity from any single proposer. 
Our performance evaluations show that \sysname{} achieves an impressive speedup of up to $50 \times$ compared to the native execution provided by Tusk.

\begin{acks}
  This work is partially funded by NSF Award Number 2245373.
\end{acks}

\vspace{10mm}
\begingroup\noindent\raggedright\textbf{Artifact Availability:}\\
The source code, data, and/or other artifacts have been made available at \url{https://github.com/apache/incubator-resilientdb/tree/Thunderbolt}.
\endgroup

\bibliographystyle{NewEDBTStyle-2025/ACM-Reference-Format}

\bibliography{references}

\clearpage
\begin{appendices}
\end{appendices}

\end{document}